\documentclass[]{article}
\usepackage{jheppub}
\usepackage{amsmath}
\usepackage[english]{babel}
\usepackage{amssymb}
\usepackage{mathtools}
\usepackage{color}
\definecolor{ao}{rgb}{0.0, 0.5, 0.0}

\usepackage{graphicx}
\usepackage{longtable}
\usepackage{subcaption}
\usepackage{tikz}
\usepackage{comment}
\usepackage{capt-of}
\usepackage{amsthm}
\usepackage{hyperref}
\hypersetup{
    colorlinks=true,
    linkcolor=blue
    }

\usepackage{bm}

\newtheorem{remark}{Remark}
\newtheorem{theorem}{Theorem}
\newtheorem{proposition}{Proposition}
\newtheorem{corollary}{Corollary}
\newtheorem{definition}{Definition}
\newtheorem{lemma}{Lemma}

\def\0{\mbox{\tiny $0$}}
\def\1{\mbox{\tiny $1$}}
\def\2{\mbox{\tiny $2$}}
\def\3{\mbox{\tiny $3$}}
\def\4{\mbox{\tiny $4$}}
\def\5{\mbox{\tiny $5$}}
\def\6{\mbox{\tiny $6$}}
\def\7{\mbox{\tiny $7$}}
\def\8{\mbox{\tiny $8$}}
\def\9{\mbox{\tiny $9$}}

\def\k{k_{_{B}}}

\def\r{\rangle}
\def\l{\langle}
\def\m{\bar{m}}
\def\p{\bar{p}}
\def\q{\bar{q}}

\def\pp{p_{12}}
\def\ppp{p_{11}}
\def\qq{q_{12}}

\def\b{\beta^{'}}
\def\bg{\beta_G^{(1)}}
\def\bgg{\beta_G^{(2)}}
\def\bsk{\beta_{SK}^{(1)}}
\def\bskk{\beta_{SK}^{(2)}}

\newcommand{\SOMMA}[2]{\displaystyle\sum\limits_{#1}^{#2}}
\newcommand{\sommaSigma}[1]{\displaystyle\sum\limits_{\lbrace#1\rbrace}}

\newcommand{\pder}[2]{\ensuremath{\dfrac{\partial #1}{\partial #2}}}

\long\def \beq#1\eeq {\begin{equation} #1 \end{equation}}
\long\def \beaq#1\eeaq {\begin{equation}\begin{aligned} #1 \end{aligned}\end{equation}}
\long\def \bes#1\ees {\begin{equation}\begin{split} #1 \end{split} \end{equation}}
\long\def \bea#1\eea {\begin{eqnarray} #1 \end{eqnarray}}
\long\def \bse[#1]#2\ese {\begin{subequations}\label{#1}\begin{align} #2 \end{align}\end{subequations}}

\newcommand{\dt}{\frac{\partial}{\partial t}}

\graphicspath{{Plots/}} 
\title{Replica symmetry breaking in dense neural networks}

\author[a,b,c]{Linda Albanese,}
\author[a,b]{Francesco Alemanno}
\author[a,c]{Andrea Alessandrelli,}
\author[a,b]{Adriano Barra,}

\affiliation[a]{Dipartimento di Matematica e Fisica,  Universit\`a  del Salento, Via per Arnesano, 73100, Lecce, Italy}

\affiliation[b]{Istituto Nazionale di Fisica Nucleare, Campus Ecotekne, Via Monteroni, 73100, Lecce,Italy}

\affiliation[c]{Scuola Superiore ISUFI, Campus Ecotekne, Via Monteroni, 73100, Lecce, Italy}

\abstract{Understanding the glassy nature of neural networks is pivotal both for theoretical and computational advances in Machine Learning and Theoretical Artificial Intelligence. Keeping the focus on dense associative Hebbian neural networks (i.e. Hopfield networks with polynomial interactions of even degree $P >2$), the purpose of this paper is two-fold: at first we develop rigorous mathematical approaches to address properly a statistical mechanical picture of the phenomenon of {\em replica symmetry breaking} (RSB) in these networks, then -deepening results stemmed via these routes- we aim to inspect the {\em glassiness} that they hide.
\newline 
In particular, regarding the methodology, we provide two techniques: the former (closer to mathematical physics in spirit) is an adaptation of the transport PDE to the case, while the latter (more probabilistic in its nature) is an extension of Guerra's interpolation breakthrough. Beyond coherence among the results,  either in replica symmetric and in the one-step replica symmetry breaking level of description, we prove the  Gardner's picture (heuristically achieved through the replica trick) and we identify the maximal storage capacity by a ground-state analysis in the Baldi-Venkatesh high-storage regime.
\newline
In the second part of the paper we investigate the glassy structure of these networks: at difference with the replica symmetric scenario (RS), RSB actually stabilizes the spin-glass phase. We report huge differences w.r.t. the standard pairwise Hopfield limit: in particular, it is known that it is possible to express the free energy of the Hopfield neural network  (and, in a cascade fashion, all its properties) as a linear combination of the free energies of an hard spin glass (i.e. the Sherrington-Kirkpatrick model) and a soft spin glass (the Gaussian or ''spherical'' model). While this continues to hold also in the first step of RSB for the Hopfield model, this is no longer true when interactions are more than pairwise (whatever the level of description, RS or RSB). For dense networks  solely the free energy of the hard spin glass survives. As the Sherrington-Kirkpatrick spin glass  is full-RSB (i.e. Parisi theory holds for that model), while the Gaussian spin-glass is replica symmetric, these different representation theorems prove a huge diversity in the underlying glassiness of associative neural networks.}

\begin{document}

\maketitle


\section*{Introduction}

As  the raise of Artificial Intelligence (AI) keeps spreading neural networks and learning algorithms in  countless meanders of society and scientific research, a rationale behind such an empirical progress continues to be a urgent priority in the agendas of theoreticians worldwide:
en route toward a Theory for AI (where all the spontaneous information processing skills that neural networks and learning machines enjoy would be somehow expected and no longer surprising)  {\em statistical mechanics of complex systems} (namely, Parisi {\em spin glass theory}) is a longstanding pillar. 
Glassy statistical mechanics has been indeed the main methodological approach allowing a post-winter pioneering -but exhaustive- picture of the Hopfield associative memory, achieved by Amit-Gutfreund-Sompolinksy (AGS) in the eighties \cite{Amit,AGS}: since the AGS milestone, it became evident that spin glasses and neural networks were intimately related and progresses in Computer Science arose from this relation quickly enlarged to computational complexity \cite{Ksat},  machine learning \cite{NoiConLenka},  combinatorial optimization \cite{postino}, error correcting codes \cite{Nishimori} and much more (see e.g. \cite{MezardMontanari,Steven}). A main reward in the usage of  glassy statistical mechanics, beyond a good comprehension of machine operational modes (welcome in eXplainable AI, XAI)   lies in painting phase diagrams for the neural architecture under inspection: this is a main route toward Optimized AI (OAI) as we briefly explain. Phase diagrams are plots in the space of the tunable parameters of the machine where its different operational modes naturally emerge and are split  by {\em phase transitions} much similar in spirit to those phase transitions that split the three different macroscopic behaviors of a glass of water in its phase diagram in Physics (i.e. the three regions: vapour, ice and liquid, in the space of its control parameters, namely pressure, temperature and volume). The knowledge of the phase diagram constitutes precious information as it allows setting the network in the desired operational regime {\em a priori}, before training and energy consumption. 
\newline
Glassy statistical mechanics is thus  the methodological leitmotif of the paper,  while the subject of the investigation  are dense Hebbian networks, i.e. generalizations of the Hopfield model where neurons -rather than interacting pairwisely-    interact in P-ples (such that when $P=2$ the Hopfield reference is recovered). Indeed dense neural networks \cite{denseHop1} are now taking hold, due to the fact that they have excellent properties of pattern recognition and image detection, remaining robust against adversarial attacks \cite{zhang, denseHop2,DAM-C}.  
\newline
As it is clearly emerging in these years by a  plethora of investigations (see e.g. \cite{ZecchinaOld, Zecchina-New,Kuhn,Remi,Gavin}), Replica Symmetry Breaking (RSB) is by far a crucial mechanism that should be better understood in modern information processing networks: despite working under the Parisi's replica symmetry breaking (RSB) scheme is notoriously challenging \cite{Talabook}, due to a series of breakthroughs that  Guerra obtained in their mathematical treatment in the past two decades (see e.g. \cite{Guerra}), times are ripe for such investigations, at least at the first step of RSB (that is the solely addressed here).
\newline 
Before we start reporting our results, we highlight that there are two -rather different- storage scalings (that results in manifestly different operational regimes) that these networks can hold: the Baldi$\&$Venkatesh {\em high storage regime} \cite{baldi} and  a new {\em high resolution regime} discovered last year \cite{PRLNN}.  
\begin{itemize}
\item Regarding the former, since the pioneering analyses by Baldi $\&$ Venkatesh \cite{baldi}, Bovier $\&$  Niederhauser \cite{Bovier} and Elisabeth Gardner \cite{gardner}, it became clear that the maximal storage capacity for these systems satisfies the following scaling:  calling $K$ the amount of patterns to store and $N$ the neurons in the network  $P$-wisely interacting, at most these network face a storage $K = \gamma N^{P-1}$ -for some positive $\gamma$  (indeed, for the Hopfield model -that  is recovered when $P=2$-  AGS theory predicts that $K \leq \gamma_c N^1$, with $\gamma_c \sim 0.138$). In this high storage regime -that we call the {\em Baldi $\&$ Venkatesh regime}- dense networks perform  standard signal-to-noise detection, namely if the pattern to be retrieved has magnitude $O(1)$, the noise can not be larger than the signal. 
\item Regarding the latter, in 2020 the existence of a completely different operational mode has been proved for these networks \cite{PRLNN}: these can sacrifice memory storage to lower their threshold for signal detection. For instance, a dense network whose neurons interact $4$-wisely (hence $P=4$) -forced to store just $K \propto N^1$ patterns (hence far from the Baldi $\&$ Venkatesh regime $K \propto N^3$)- can detect  a pattern whose intensity is $O(1)$ even when corrupted by a noise $O(\sqrt{N})$ in the large $N$ limit \cite{FrancAlberto,AgliariDeMarzo}.
\end{itemize}
We will deepen replica symmetry breaking in this {\em high-resolution regime} in a forthcoming paper, while in the present one we focus on dense networks solely in the {\em high-storage regime}.
\newline
\newline
The paper is structured as follows and presents the following results: 
\newline
Once introduced these networks, we adapt two mathematical methods for tackling their statistical mechanics description at the first step of replica symmetry breaking. At first, framing the present research within the plethora of methodologies that are raising as alternatives to the celebrated replica trick \cite{MPV} (see e.g. \cite{NPD,deepTransport,Alemannation1,barrabecca,Albert2,Antonio,Antonio2,Antonio3,Barbier,Bates1,Bates2,Bates3,Gavin,Marullo,Murrat1,MurratPanchenko,Pax,Subag1,Subag2}),  driven by calculus and analysis, in Section \ref{Sezione2}  we work out a PDE-based theory where it is possible to obtain the phase diagrams of these models  by solving suitable transport equations in the space of the control parameters,  then, grabbing from probability theory, in Section \ref{Sezione3} we adapt the celebrated Guerra's broken-replica interpolation \cite{Guerra} to the case. Beyond coherence among the results, we  also re-obtain both the Gardner picture and the  Baldi $\&$ Venkatesh  scaling,  beyond a number of new results useful for understanding the glassy nature of these neural networks, that we inspect in the second part of the paper.
\newline
By a straight comparison of the replica symmetric and broken replica symmetry phase diagrams, while the critical storage is mildly affected by the RSB phenomenon,  the glassy region -that shrinks close to disappearing in the replica symmetric description- gets actually stable by a step of replica symmetry breaking: this is discussed in Section \ref{Sezione4}. Further, in Section \ref{Sezione5}, we prove a series of representation theorems, that allow to decompose Hebbian networks into combinations of pure spin glasses, whose significance can be summarized as follows:
\begin{itemize}
\item at the replica symmetric (RS) level, the standard ($P=2$) Hopfield model (technically speaking its free energy) can be described as a linear combination of (the free energies of) two spin-glasses, the former a standard Sherrington-Kirkpatrick spin glass (that is full-RSB and where Parisi theory is exact \cite{Guerra,Talagrand}), the latter is a Gaussian (or ''spherical'' \cite{Bovier,Dembo}) spin glass (that is solely replica symmetric in the pairwise case \cite{soffice,Crisanti}).
\item at one step of replica symmetry breaking (1-RSB), the standard ($P=2$) Hopfield model (technically speaking its free energy) can still be described  by the above decomposition in terms of a hard and a soft spin glass.
\item at the replica symmetric level (RS), the dense ($P>2$) Hebbian network (technically speaking its free energy) is no longer a linear combination of (the free energies of) two spin glass, rather solely the hard part survives, namely that pertaining to a Sherrington-Kirkpatrick model with $P$-wise interactions.
\item  at one step of replica symmetry breaking (1-RSB), the dense ($P>2$) Hebbian network (technically speaking its free energy) is still no longer a linear combination of (the free energies of) two spin glass and solely the hard part survives.
\end{itemize}
The whole contribute to highlight the different glassy nature of neural networks that, in turn, helps understanding the structure and organization of the valleys in the free energy landscape where information is stored by the Hebbian mechanism (that ultimately implies a better understanding of information processing by these networks). 

\section{Generalities}\label{HopfieldSection}

In this section we provide details on the neural networks we aim to study.  We focus on Hebbian networks whose $N$ digital neurons  (i.e. Ising spins) lie on the nodes of a fully connected network and interact P-wisely via a suitable tensorial generalization of the standard Hebbian storing rule, where $K$ patterns  $\xi^{\mu}$,  $\mu \in (1,...,P)$  -all of the same length $N$, are stored. It is useful to define as control parameters $\beta$ and $\gamma$, where
\begin{align}
\begin{cases}
\beta &= \displaystyle{\frac{1}{T}} \notag \\
\gamma &= \displaystyle{\lim_{N \to \infty}} \displaystyle{\frac{K}{N^{P-1}}},
\end{cases}
\end{align}
while $\beta \in \mathbb{R}^+$  (i.e. the {\em inverse} of the temperature T in Physics) tunes the fast noise  in the network   
such that, while for $\beta \to 0$ the neural dynamics of the network becomes an uncorrelated random walk in the configuration space,   for $\beta \to \infty$ it approaches a steepest descent to the closest minimum of the cost function, that plays as a Lyapounov function in this limit (and the probability distribution $P(\sigma|\xi)$ drifts from a uniform distribution in the first case to be sharply peaked at the minima of the energy function \eqref{eq:hop_hbare} in the opposite noiseless limit).

\begin{definition} 
\label{def:pspinham} 
Set $\gamma \in \mathbb{R}^+$, $a \in \mathbb{N}$, $P \in \mathbb{N}$ even and let $\boldsymbol \sigma \in \{- 1, +1\}^N$ be a configuration of $N$ binary neurons. Given $K=\gamma N^a$ random patterns $\{\boldsymbol \xi^{\mu}\}_{\mu=1,...,K}$, each made of $N^{P/2}$ i.i.d. digital entries drawn from probability $P(\xi_{i_{_1}\cdots \,i_{_{P/2}}}^{\mu}=+1)=P(\xi_{i_{_1}\cdots \,i_{_{P/2}}}^{\mu}=-1)=1/2$, for $i=1,...,N$, the cost-function (or Hamiltonian to preserve a physical jargon) of the dense Hebbian network (DHN) is defined as
	\beq
	\label{eq:hop_hbare}
	H_N^{(P)}(\boldsymbol \sigma| \boldsymbol \xi) \coloneqq -\frac{1}{P!\,N^{P-1}}\sum_{\mu=1}^{K}\left(\sum_{i_{1},\cdots,i_{P/2}=1}^{N,\cdots,N}\xi_{i_{_1}\cdots i_{_{P/2}}}^{\mu}\sigma_i\cdots\sigma_{i_{_{P/2}}}\right)^{2}  -\frac{\gamma}{P!}\,N^{a+1-P/2}.
	\eeq
	\label{hop_hbare}
\end{definition}
Note that the last term at the r.h.s. is due to the subtraction of the diagonal term (as we wrote the summations without restrictions in the cost function itself). The normalization factor $1/N^{P-1}$ ensures the linear extensivity of the Hamiltonian, in the volume of the network $N$, as expected.
\newline
Note that we select the Hebbian structure for the  tensor accounting for the synaptic couplings  in the factorized form $\xi^1_{\boldsymbol i}\equiv\xi^1_{i_1}\cdots\xi^1_{i_{P/2}}$.
\begin{definition}The partition function related to the Hamiltonian of the DHN given by \eqref{eq:hop_hbare} reads as 
	\begin{align}
	\label{eq:hop_BareZ}
	\mathcal Z_N(\beta,\boldsymbol \xi) &\coloneqq \sum_{ \boldsymbol \sigma }^{2^N} \exp \left[ -\beta \left(H_N^{(P)}(\boldsymbol \sigma | \boldsymbol \xi)\right)\right] \notag \\
	&=\sum_{ \boldsymbol \sigma } \exp \left[ \frac{\beta}{P!\,N^{P-1}}\sum_{\mu=1}^{K}\left(\sum_{i_{1},\cdots,i_{P/2}=1}^{N,\cdots,N}\xi_{i_{_1}\cdots i_{_{P/2}}}^{\mu}\sigma_i\cdots\sigma_{i_{_{P/2}}}\right)^{2}-\dfrac{\beta\gamma}{P!}N^{a-P/2}\right],
	\end{align}
\end{definition} 

For an arbitrary observable $O(\boldsymbol \sigma)$, we introduce the \emph{Boltzmann average} induced by the partition function (\ref{eq:hop_BareZ}), denoted with $\omega_{\boldsymbol \xi}$, defined as
	\begin{equation}
	\omega_{\boldsymbol \xi} (O (\boldsymbol \sigma)) : = \frac{\sum_{\boldsymbol \sigma} O(\boldsymbol \sigma) e^{- \beta H_N(\boldsymbol \sigma| \boldsymbol \xi)}}{Z_N(\beta, \boldsymbol \xi)}.
	\end{equation}
	This can be further averaged over the realization of the $\xi_{\boldsymbol i}^{\mu}$'s (also referred to as \emph{quenched average}) to get
	\beq
	\langle O(\boldsymbol \sigma) \rangle \coloneqq \mathbb{E} \omega_{\boldsymbol \xi} (O(\boldsymbol \sigma)).
	\eeq
\begin{definition}The intensive quenched statistical pressure of the DHN (\ref{eq:hop_hbare}) is defined as
\begin{equation}
\label{PressureDef}
\mathcal A_N(\beta,\gamma) \coloneqq \frac{1}{N} \mathbb{E} \ln \mathcal Z_N(\beta, \boldsymbol \xi),
\end{equation}
and its thermodynamic limit, assuming its existence, is referred to as $\mathcal A(\beta,\gamma) \coloneqq \lim_{N \to \infty} \mathcal A_N(\beta,\gamma)$.
\end{definition}
Focusing on pure state retrieval, we assume without loss of generality \cite{AABF-NN2020,FrancAlberto,Beppe} that the candidate pattern to be retrieved  -say $\boldsymbol \xi^1$-  is a Boolean vector, while $\boldsymbol \xi^{\mu}$, $\mu=2,...,K$ are real vectors whose entries are drawn from i.i.d. standard Gaussians. Accordingly, the average $\mathbb{E}$ acts as a Boolean average over $\boldsymbol \xi^1$ and as a Gaussian average over $\boldsymbol \xi^2 \cdots \boldsymbol \xi^{K}$.
\begin{definition} The order parameters required to describe the macroscopic behavior of the model are the standard ones \cite{Amit,Coolen,FrancAlberto,barrabecca}, namely, the Mattis magnetization 
	\begin{equation}
	m \coloneqq \frac{1}{N}\sum_{i=1}^{N} \xi_i^1 \sigma_i
	\end{equation}
	necessary to quantify the retrieval capabilities of the network
	and the two-replica overlap in the $\boldsymbol \sigma$'s variables
\begin{equation}
\label{q}
q_{12} \coloneqq \frac{1}{N}\sum_{i=1}^N \sigma_i^{(1)}\sigma_i^{(2)}
\end{equation}
required to quantify the level of slow noise the network must cope with (when performing pattern recognition).
Further, as an additional set of variables $\{\tau_{\mu}\}_{\mu=1,...,P-1}$ shall be introduced (vide infra), we accordingly define their related two-replica overlaps
\begin{equation}
\label{p}
p_{11} \coloneqq  \frac{1}{P} \sum_{\mu=1}^P \tau^{(1)}_\mu \tau^{(1)}_\mu, ~~~ p_{12} \coloneqq  \frac{1}{P} \sum_{\mu=1}^P \tau^{(1)}_\mu \tau^{(2)}_\mu
\end{equation}
for mathematical convenience.
\label{hop_orderparameters}
\end{definition}

\section{First approach: transport PDE}\label{Sezione2}
As stated in the introduction, a purpose of our investigation is to paint phase diagrams for the networks in the space of the tunable parameters, en route toward an Optimized AI: to reach this goal the prescription is to obtain an explicit expression of the quenched statistical pressure in terms of the order parameters and than extremize the former over the latter. This procedure returns a system of coupled self-consistent equations that trace the evolution of the order parameters in the space of the control parameters, whose inspection ultimately allows such a desired painting. We approach this picture by providing two mathematical alternatives, the former based on mathematical physics methods -as we deepen hereafter- and the latter more grounded on a probabilistic setting (as we will see in the next section). For both the approaches we work out in full detail both the replica symmetric  and the first-step of replica symmetry breaking scenarios and compare their findings.
\newline 
\newline
In this section -at work with PDE theory- the strategy is to introduce an interpolating pressure $\mathcal A^{(P)}_N(\bold{r},t)$ living in an enlarged fictitious space-time  $(\bold{r},t)$ that actually reduces to the intensive quenched statistical pressure $\mathcal A^{(P)}_N$ of the original model in a specific point of this space-time (namely for  $(\bold{r}=0,t=1)$, i.e. $\mathcal A^{(P)}_N(\bold{r},t)=\mathcal A^{(P)}_N(\beta,\gamma)$) the plan is thus to work out explicitly the derivative of the interpolating pressure w.r.t. the space-time and to show that they fulfills a transport PDE in such a way that the solution of the statistical mechanical problem is recast in the solution of a partial differential equation, converting a problem of statistical mechanics of neural networks into a typical problem of mathematical physics.
The purpose of next two subsections (on for the RS and the other for the RSB) is thus to solve for the quenched free energies (or quenched statistical pressures) of these dense associative network through transport equation's method (whose idea has been already introduced in \cite{AABF-NN2020} for the replica symmetric scenario and in \cite{lindaRSB} for the broken replica symmetry  scenario dealing just with the classic Hopfield network). 

\subsection{RS approximation}
In this section we solve for the quenched statistical pressure of the dense associative network at the replica symmetric level of description. 

\begin{definition}
    \label{defn: RSassumption}
	Under the replica-symmetry assumption, in the thermodynamic limit the order parameters self-average around their mean values (denoted with a bar), i.e., their distributions get delta-peaked, independently of the replica considered, namely
	\begin{eqnarray}
	\label{eq:m_ter}
	\lim_{N\to \infty} \langle (m - \bar m)^2 \rangle = 0 &\Rightarrow& \lim_{N\to \infty}  \langle m \rangle = \bar m\\
	\lim_{N\to \infty} \langle (q_{12} - \bar q)^2 \rangle = 0 &\Rightarrow& \lim_{N \to \infty}  \langle q_{12} \rangle = \bar q\\
	\lim_{N\to \infty} \langle (p_{12} - \bar p)^2 \rangle = 0 &\Rightarrow& \lim_{N \to \infty}  \langle p_{12} \rangle = \bar p.
	\end{eqnarray}
	Note that, for the generic order parameter $X$, the above concentration can be rewritten as $\langle (\Delta X)^2 \rangle \overset{N\to\infty}{\longrightarrow}0$,
	where
	$$
	\Delta X \coloneqq  X - \bar{X},
	$$
	and, clearly, the RS approximation also implies that, in the thermodynamic limit, $\langle \Delta X \Delta Y \rangle = 0$ for any generic pair of order parameters $X,Y$ as well as  $\langle (\Delta X)^k \rangle \rightarrow 0$ for $k \geq 2$.
	
\end{definition}

\begin{definition} 
Given the interpolating parameter $t, x, y, z, w$, and $J_i$, $\tilde{J\,}_\mu \sim \mathcal{N}(0,1)$  standard i.i.d. Gaussian variables, the partition function in its integral representation is given by  
\begin{equation}
\footnotesize
\begin{array}{lll}
     \mathcal{Z}^{(P)}_N(t, \bm r) &\coloneqq& \sommaSigma{\boldsymbol \sigma} \displaystyle\int \mathcal{D} \bm \tau\exp{}\Bigg[t\dfrac{\beta\,'\,N}{2} m^P(\boldsymbol \sigma)+w N\,m(\boldsymbol \sigma)+\sqrt{t}\sqrt{\dfrac{\beta\,' }{N^{P-1}}}\SOMMA{\mu>1}{K}\,\left(\SOMMA{i_1,\cdots,i_{_{P/2}}=1}{N,\cdots,N}\xi^{\mu}_{i_1\cdots,i_{_{P/2}}}\sigma_{i_1}\cdots\sigma_{i_{P/2}}\right)\tau_{\mu}
     \\\\
     & &+\sqrt{N^{1-P/2}x}\SOMMA{\mu>1}{K} \tilde{J}_{\mu}\tau_{\mu}+\sqrt{y}\SOMMA{i=1}{N} J_i\sigma_i+\dfrac{zN^{1-P/2}}{2}\SOMMA{\mu>1}{K}\,{\tau^2_{\mu}}-\dfrac{\beta'\gamma}{2}N^{a-P/2}\Bigg]\,,
     \label{def:partfunct_transpRS}
\end{array}
\end{equation}
where, for any $\mu=1,...,K$, $\tau_{\mu} \sim \mathcal N [0, 1]$ and $\mathcal{D} \bm \tau \coloneqq \prod\limits_{\mu=1}^K \frac{e^{- \tau_{\mu}^2/2}}{\sqrt{2\pi} } d\tau_\mu$ 
is the related measure and we set $\beta'=2 \beta/P!$.
\end{definition}

\begin{definition} The interpolating pressure for the Dense Hebbian Network (DHN) (\ref{eq:hop_hbare}), at finite $N$, is introduced as
\begin{eqnarray}
\mathcal{A}^{(P)}_N(t, \bm r) &\coloneqq& \frac{1}{N} \mathbb{E} \left[  \ln \mathcal{Z}^{(P)}_N(t, \bm r)  \right],
\label{hop_GuerraAction}
\end{eqnarray}
where the expectation $\mathbb E$ is now meant over $\boldsymbol \xi$, $\boldsymbol J$, and $\boldsymbol{\tilde{J}}$ and, in the thermodynamic limit,
\begin{equation}
\mathcal{A}^{(P)}(t, \bm r) \coloneqq \lim_{N \to \infty} \mathcal{A}^{(P)}_N(t, \bm r).
\label{hop_GuerraAction_TDL}
\end{equation}
By setting $t=1$ and $\bold{r}=0$ the interpolating pressure recovers the original one (\ref{PressureDef}), that is $\mathcal A^{(P)}_N (\beta,\gamma) = \mathcal{A}^{(P)}_N(t=1, \bm r = \bm 0)$.
\end{definition}
\begin{remark}
	The interpolating structure implies an interpolating measure, whose related Boltzmann factor reads as
	\begin{equation}
	\mathcal B (\boldsymbol \sigma, \boldsymbol \tau; t, \bm r )\coloneqq  \exp \left[ \beta \mathcal H (\boldsymbol \sigma, \boldsymbol \tau; t, \bm r) \right];
	\end{equation}
In this way $\mathcal Z_N(t, \bm r) = \int \mathcal{D} \bm \tau \sum_{\boldsymbol \sigma} \mathcal B (\boldsymbol \sigma, \bm \tau; t, \bm r)$ and a generalized average is coupled to this generalized measure as
\beq
	\omega_{t, \boldsymbol r} (O (\boldsymbol \sigma, \boldsymbol \tau )) \coloneqq  \int \mathcal{D} \bm \tau  \sum_{\boldsymbol \sigma} O (\boldsymbol \sigma , \boldsymbol \tau) \mathcal B (\boldsymbol \sigma, \boldsymbol \tau; t)
	\eeq
	and
\beq
\langle O (\boldsymbol \sigma, \bm \tau ) \rangle_{t, \bm r}  \coloneqq \mathbb E [ \omega_{t, \bm r} (O (\boldsymbol \sigma, \bm \tau )) ].
\eeq
Of course, when $t=1$ the standard Boltzmann measure and related averages are recovered.
\newline
Hereafter, in order to lighten the notation, we drop the sub-indices $t, \bm r$.
\end{remark}

\begin{lemma} 
The partial derivatives of the interpolating pressure (\ref{hop_GuerraAction}) w.r.t. $t,x,y,z,w$ give the following expectation values:
	\bea
	\label{hop_expvalsa}
	\frac{\partial \mathcal{A}^{(P)}_N}{\partial t} &=& \dfrac{\beta '}{2}\l m^P \r +\dfrac{\beta '}{2 N^{P/2}}K\Big(\l\ppp\r-\l\pp\qq^{P/2}\r\Big),
	\\
	\label{hop_expvalsmiddle}
	\frac{\partial \mathcal{A}^{(P)}_N}{\partial x}  &=& \dfrac{K}{2N^{P/2}}\Big(\l\ppp\r-\l\pp\r\Big),\\
	\frac{\partial \mathcal{A}^{(P)}_N}{\partial y}  &=& \dfrac{1}{2}\Big(1-\l\qq\r\Big),\\
	\frac{\partial \mathcal{A}^{(P)}_N}{\partial z}  &=& \frac{K}{2N^{P/2}}  \l \ppp\r,\\
	\frac{\partial \mathcal{A}^{(P)}_N}{\partial w}  &=&  \l m\r.
	\label{hop_expvalsb}
	\eea
\end{lemma}

\begin{proof}
Since the procedures for the derivatives w.r.t. each parameter are analogous, we prove only the derivative w.r.t. $t$. The  partial  derivative  of  the  interpolating  quenched  pressure  with respect to $t$ reads as 
\begin{equation}
    \begin{array}{lll}
           \dfrac{\partial \mathcal{A}^{(P)}_N}{\partial t}&=&\dfrac{1}{N}\mathbb{E}\left[\dfrac{\beta'N}{2}\omega(m^P)\right]+ \dfrac{1}{2N\sqrt{t}}\sqrt{\dfrac{\beta'}{N^{P-1}}}\sum\limits_{\mu}\sum\limits_{\boldsymbol{i}}\mathbb{E}\left[\xi^{\mu}_{i_1\cdots,i_{_{P/2}}}\omega(\sigma_{i_1}\cdots\sigma_{i_{P/2}}\tau_{\mu})\right]\,.
    \end{array}
\end{equation}
Now, on standard Gaussian variable $\xi_{\boldsymbol{i}}^{\mu>1}$ we apply the Stein's lemma (also known as Wick's theorem), namely  
\begin{align}
\mathbb{E} \left( J f(J)\right)= \mathbb{E} \left( \frac{\partial f(J)}{\partial J}\right)
\label{eqn:gaussianrelation2}
\end{align}
to compute the derivative w.r.t. $t$ as
\begin{equation}
    \begin{array}{lll}
           \dfrac{\partial \mathcal{A}^{(P)}_N}{\partial t}=&\dfrac{\beta'}{2}\l m^P\r+ \dfrac{\beta' K}{2N^{^{P/2}}}\left(\mathbb{E}\left[\omega((\sigma_{i_1}\cdots\sigma_{i_{P/2}}\tau_{\mu})^2)\right]-\mathbb{E}\left[\omega(\sigma_{i_1}\cdots\sigma_{i_{P/2}}\tau_{\mu})^2\right]\right)
           \\\\
           \textcolor{white}{\dfrac{\partial \mathcal{A}_N}{\partial t}} =&\dfrac{\beta'}{2}\l m^P\r+ \dfrac{\beta' K}{2N^{^{P/2}}}\left(\l\ppp\r-\l\pp\qq^{P/2}\r\right)\,.
    \end{array}
\end{equation}
\end{proof}

\begin{remark}
In the next computations, we can use the following relations
\begin{align}
    \langle m_1^P \rangle &= \sum_{k=2}^P \begin{pmatrix}P\\k\end{pmatrix} \langle (m_1-\bar{m})^k \rangle \bar{m}^{P-k} + \bar{m}^P (1-P) + P\bar{m}^{P-1}\langle m_1 \rangle\,, \label{potential_m}
    \\
    \langle p_{12}q_{12}^{P/2}\rangle &= \sum_{k=1}^{P/2} \begin{pmatrix}\frac{P}{2}\\k\end{pmatrix}\bar{q}^{P/2-k} \langle (p_{12}-\bar{p})(q_{12} - \bar{q})^k \rangle + \sum_{k=2}^{P/2} \begin{pmatrix}\frac{P}{2}\\k\end{pmatrix} \bar{q}^{P/2-k} \bar{p}\langle (q_{12} - \bar{q})^k \rangle+ \notag
    \\
    &\textcolor{white}{=}+\q^{P/2}\l\pp\r+\dfrac{P}{2}\q^{P/2-1}\p\l\qq\r-\dfrac{P}{2}\q^{P/2}\p\,,
    \label{potential_pq}
\end{align}
which can be proved trivially by brute force. 
\end{remark}

\begin{proposition} \label{prop:interp_transp_RS}
The interpolating pressure \eqref{def:partfunct_transpRS} at finite size $N$ obeys the following transport-like partial differential equation:
	\begin{equation}
	\frac{d \mathcal{A}^{(P)}_N}{dt} =  \pder{\mathcal{A}_N^{(P)}}{t} + \dot{x} \pder{\mathcal{A}^{(P)}_N}{x} + \dot y \pder{\mathcal{A}^{(P)}_N}{y} + \dot z \pder{\mathcal{A}^{(P)}_N}{z} +\dot w \pder{\mathcal{A}^{(P)}_N}{w}= S(t, \bm r)+V_N(t, \bm r),
	\label{hop_GuerraAction_DE}
	\end{equation}
	where we set 
	\begin{equation}
	    \begin{array}{lll}
	         \dot x = -\beta ' \q^{P/2}\,,&&  \dot y = -\dfrac{P}{2}\beta ' \gamma \p\q^{P/2-1}\,,
	         \\\\
	         \dot z = - \beta' (1-\q^{P/2})\,, && \dot w = -\dfrac{P}{2}\beta '\m^{P-1}
	    \end{array}
	    \label{eq:deriv_r_trasp_RS}
	\end{equation}
	and the source $S$ and the potential $V$ read (respectively) as
	\begin{equation}
	\begin{array}{lll}
	     S(t, \bm r) &\coloneqq& -\dfrac{P-1}{2}\beta ' \m^P-\beta '\gamma\dfrac{P}{4}\p\q^{P/2-1}(1-\q),
	\end{array}
	\end{equation}
	\begin{equation}
	\begin{array}{lll}
	     V_N(t, \bm r) &\coloneqq& \dfrac{\beta '}{2}\SOMMA{k=2}{P}\begin{pmatrix}
	    P\\k
	\end{pmatrix} \m^{P-k}\l(\Delta m)^k\r-\dfrac{\beta '\gamma}{2 N^{P/2-a}}\SOMMA{k=2}{P/2}\begin{pmatrix}
	    \frac{P}{2}\\k
	\end{pmatrix} \p\q^{P/2-k}\l(\Delta q)^k\r+
	\\\\
	&&-\dfrac{\beta '\gamma }{2 N^{P/2-a}}\SOMMA{k=2}{P/2}\begin{pmatrix}
	    \frac{P}{2}\\k
	\end{pmatrix} \q^{P/2-k}\l\Delta p(\Delta q)^k\r.
	\end{array}
	\label{potenziale-RS-Hopfield}
	\end{equation}
	\end{proposition}
	
	\begin{proof}
	Starting to evaluate explicitly $\dt \mathcal{A}^{(P)}_N$ by using \eqref{hop_expvalsa}-\eqref{hop_expvalsmiddle} and \eqref{potential_m}-\eqref{potential_pq}, we write
\begin{equation}
\footnotesize
    \begin{array}{lll}
         \dfrac{\partial}{\partial t}\mathcal{A}^{(P)}_N=\dfrac{\beta '}{2}\left(\SOMMA{k=2}{P} \begin{pmatrix}P\\k\end{pmatrix} \langle (m_1-\bar{m})^k \rangle \bar{m}^{P-k} + \bar{m}^P (1-P) + P\bar{m}^{P-1}\langle m_1 \rangle\right) 
         \\\\
         +\dfrac{\beta '}{2 N^{P/2}}K\Big(\l\ppp\r\Big)-\dfrac{\beta '}{2 N^{P/2}}K\Big(\q^{P/2}\l\pp\r+\dfrac{P}{2}\q^{P/2-1}\p\l\qq\r-\dfrac{P}{2}\q^{P/2}\p \Big)
         \\\\
         -\dfrac{\beta '}{2 N^{P/2}}K\Big(\SOMMA{k=1}{P/2} \begin{pmatrix}\frac{P}{2}\\k\end{pmatrix}\bar{q}^{P/2-k} \langle (p_{12}-\bar{p})(q_{12} - \bar{q})^k \rangle + \SOMMA{k=2}{P/2} \begin{pmatrix}\frac{P}{2}\\k\end{pmatrix} \bar{q}^{P/2-k} \bar{p}\langle (q_{12} - \bar{q})^k \rangle \Big)
         \\\\
         =V_N(t,\boldsymbol{r})+S(t,\boldsymbol{r}) +\beta '\gamma\dfrac{P}{4}(N^{^{a-P/2}}\p)\q^{P/2-1}+\dfrac{\beta '}{2}P\bar{m}^{P-1}\langle m_1 \rangle 
         \\\\
         +\dfrac{\beta '\gamma N^{a-P/2}}{2}\l\ppp\r-\dfrac{\beta '\gamma N^{a-P/2}}{2}\Big(\q^{P/2}\l\pp\r+\dfrac{P}{2}\q^{P/2-1}\p\l\qq\r \Big)
         \\\\
         =V_N(t,\boldsymbol{r})+S(t,\boldsymbol{r}) +\dfrac{\beta '}{2}P\bar{m}^{P-1}\Big(\frac{\partial \mathcal{A}^{(P)}_N}{\partial w} \Big)+\beta ' (1-\q^{P/2})\Big(\frac{\partial \mathcal{A}^{(P)}_N}{\partial z}\Big)+
         \\\\
         +\beta '\q^{P/2}\Big(\frac{\partial \mathcal{A}^{(P)}_N}{\partial x} \Big) +\dfrac{\beta '\gamma P\,N^{a-P/2}}{2}\p\q^{P/2-1}\Big(\frac{\partial \mathcal{A}^{(P)}_N}{\partial y}\Big) 
    \end{array}
\end{equation}
Thus, by placing $\dot{\boldsymbol{r}}=(\dot{x},\dot{y},\dot{z},\dot{w})$ as in \eqref{eq:deriv_r_trasp_RS} and $N^{a-P/2}\p$ as $\p$, we reach the thesis.
\end{proof}

\begin{remark}
In the thermodynamic limit and under the assumption of replica symmetry the potential $V_N(t, \bold{r})\to 0$ (this simplifies considerably the resolution of the transport equation). 
\end{remark}

\begin{theorem}
\label{cor_carmassimoRS}
In the thermodynamic limit and under the assumption of replica symmetry, the maximum storage that the network can handle is $K \propto N^{P-1}$ -namely the Baldi-Vekatesh storage \cite{baldi})- that is achieved for 
\begin{equation}
    a=P-1.
    \label{eq:load_of_P_Spin}
\end{equation}
In this regime the quenched statistical pressure for $P\geq 4$  of the DHN becomes
\begin{equation}
\begin{array}{lll}
      \mathcal{A}^{(P)} (\gamma, \beta) & \coloneqq & \ln{2}+\left\langle\ln{\cosh{\left[\dfrac{P}{2}\beta '  \m^{P-1}+Y\sqrt{\beta ' \gamma   \dfrac{P}{2}\, \p \q^{^{P/2-1}}} \right]}}\right\rangle_Y -\dfrac{P-1}{2} \beta ' \m^P \\\\
      & &-\beta ' \gamma\dfrac{P}{4}  \, \p \q^{P/2-1}(1-\q)+\dfrac{1}{4} \gamma\beta '^2\left(1 - \q^P\right)\,.
\end{array}
\label{eq:pressure_GuerraRS}
\end{equation}
\end{theorem}

\begin{remark}
We stress that using $P=2$ and $a=1$ in the quenched pressure \eqref{eq:pressure_GuerraRS_noAPP} we recover the AGS picture \cite{AGS}.
\end{remark}

Extremizing the statistical pressure given in \eqref{eq:pressure_GuerraRS} w.r.t. the order parameters we find the following

\begin{corollary}
	The self-consistency equations ruling the evolution of the order parameters are
\begin{equation}
    \label{eq:self_GuerraRS}
    \begin{array}{lll}
         \m=\left\langle\tanh{\left[\dfrac{P}{2}\beta '  \m^{P-1}+x\sqrt{\beta ' \gamma \frac{P}{2} \p\q^{^{P/2-1}}}\,\right]}\right\rangle_x \,,
         \\\\
         \q=\left\langle\tanh{}^2{\left[\dfrac{P}{2}\beta '  \m^{P-1}+x\sqrt{\beta ' \gamma \frac{P}{2} \p \q^{^{P/2-1}}}\,\right]}\right\rangle_x \,,
         \\\\
        \p=\beta ' \q^{^{P/2}}\,.
    \end{array}
\end{equation}
\end{corollary}
\begin{remark}
We stress that the self-consistence equations obtained through our method are the same obtained by Gardner in \cite{gardner} via heuristic techniques (i.e. the replica trick).
\end{remark}
\begin{remark}
Note that the above equations are rather different w.r.t. those of the Hopfield model, in particular the equation for the overlap $\bar{q}$ does not have a denominator at the r.h.s. (as typical for pairwise models as AGS theory revealed). Actually the self-consistency for the two-replica overlap in the DHN coincides 
with the self-consistency of the two-replica overlap in the hard P-spin-glass: this suggests that the glassy structure of the dense neural networks is different w.r.t. the glassy structure of the Hopfield model. We will deepen the glassy nature of these networks in the second part of the paper (see Section \ref{sec:glassy}).
\end{remark}
By the inspection of the self-consistency, we can find regions in the space of the control parameters $\beta$ and $\gamma$ -as $P$ is varied- where the networks is ergodic (e.g. when both $\bar{m}=0$ and $\bar{q}=0$), where the network is a pure spin glass (e.g. when $\bar{m}=0$ but $\bar{q} \sim 1$) and, the most important, where the network works as an associative memory and performs spontaneously pattern recognition (e.g. when both $\bar{m} \sim 1$ and $\bar{q} \sim 1$): these phase diagrams are shown in Figure \ref{fig:diagRS0} and deepened in Figure \ref{fig:trend_mag_RS0}.  In particular, if we visually follow the red line (the boundary of the retrieval region) starting from above, we see that the curve has a point of inflection at a value of $\gamma$ that we call $\gamma_{max}$ (and then recesses to smaller critical values for $\gamma$): that flex is the point where replica symmetry gets unstable. We can quantify the evolution of this instability as $P$ grows by plotting $1-\frac{\gamma (\beta\to \infty)}{\gamma_{max}}$ (see Figure \ref{fig:trend_mag_RS0}, left panel). It is interesting to note that, for larger and larger values of $P$, the instability regions gets smaller and smaller suggesting a milder role for RSB in very dense networks: this is further corroborated by the inspection of the values of the magnetization at $\gamma_c$ that approach one as $P\to \infty$ (see Figure \ref{fig:trend_mag_RS0}, right panel) and justifies why we investigated solely the first step of RSB in the following subsection.

\begin{figure}[h!]
    \centering
    \includegraphics[scale=0.7]{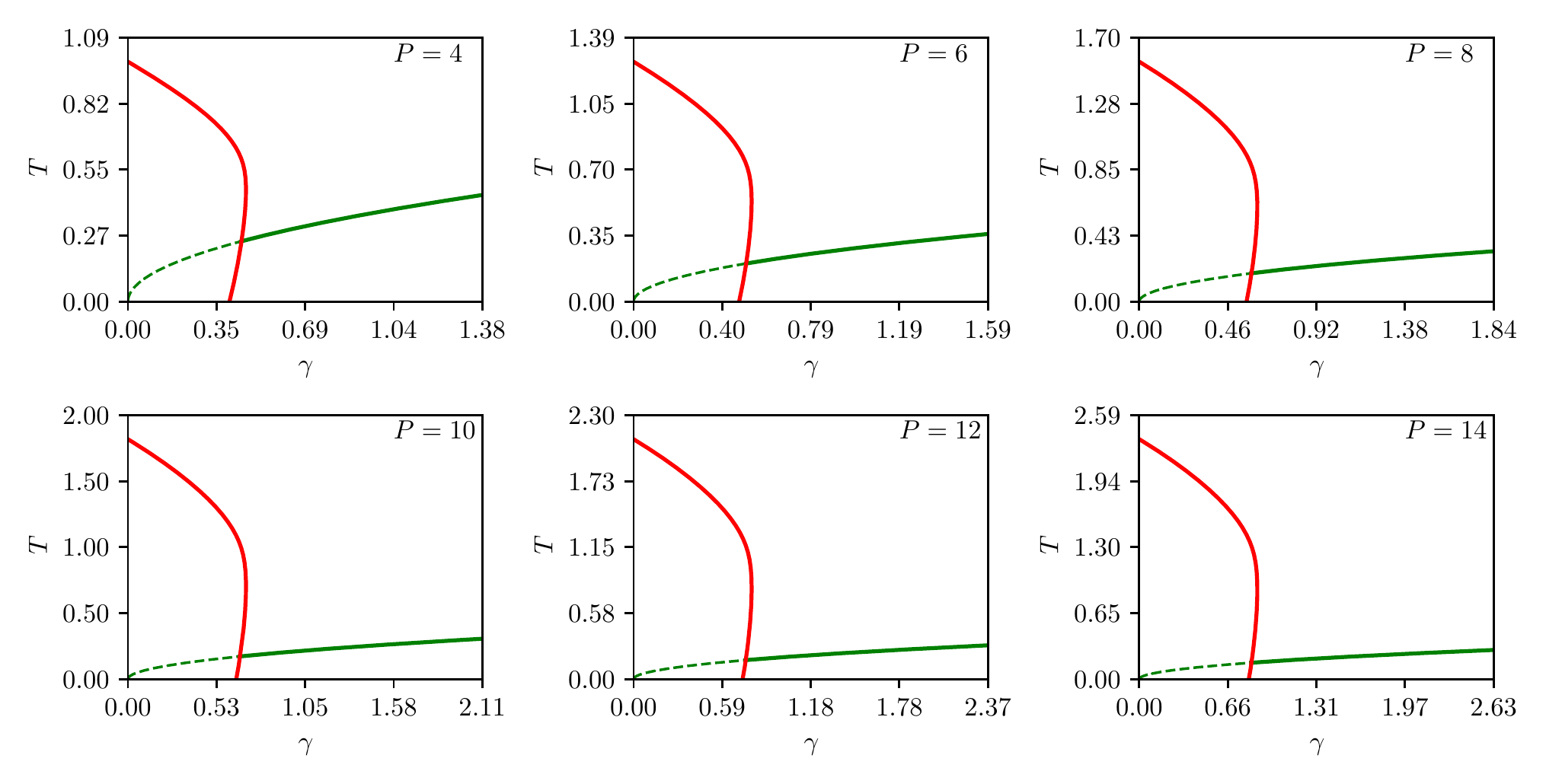}
    \caption{Replica symmetric (RS) phase diagram of the dense associative network at different values of $P$. The red curve identifies the phase transition splitting the retrieval region (on the left) from the spin glass phase (on the right), while the green curve identifies the boundary of the spin glass region (down) from the ergodic region (above). We stress that as P grows the spin glass region shrinks, as quantified in Figure \ref{fig:trend_mag_RS0} (left), further the pure spin glass solution -within the retrieval region- is always unstable and it is depicted by the dotted green curve: we call this region {\em instability region} and we inspect its evolution with $P$ in in Figure \ref{fig:trend_mag_RS0} (right).}
    \label{fig:diagRS0}
\end{figure}

\begin{figure}[h!]
     \begin{subfigure}[b]{0.5\textwidth}
         \centering
         \includegraphics[width=\textwidth]{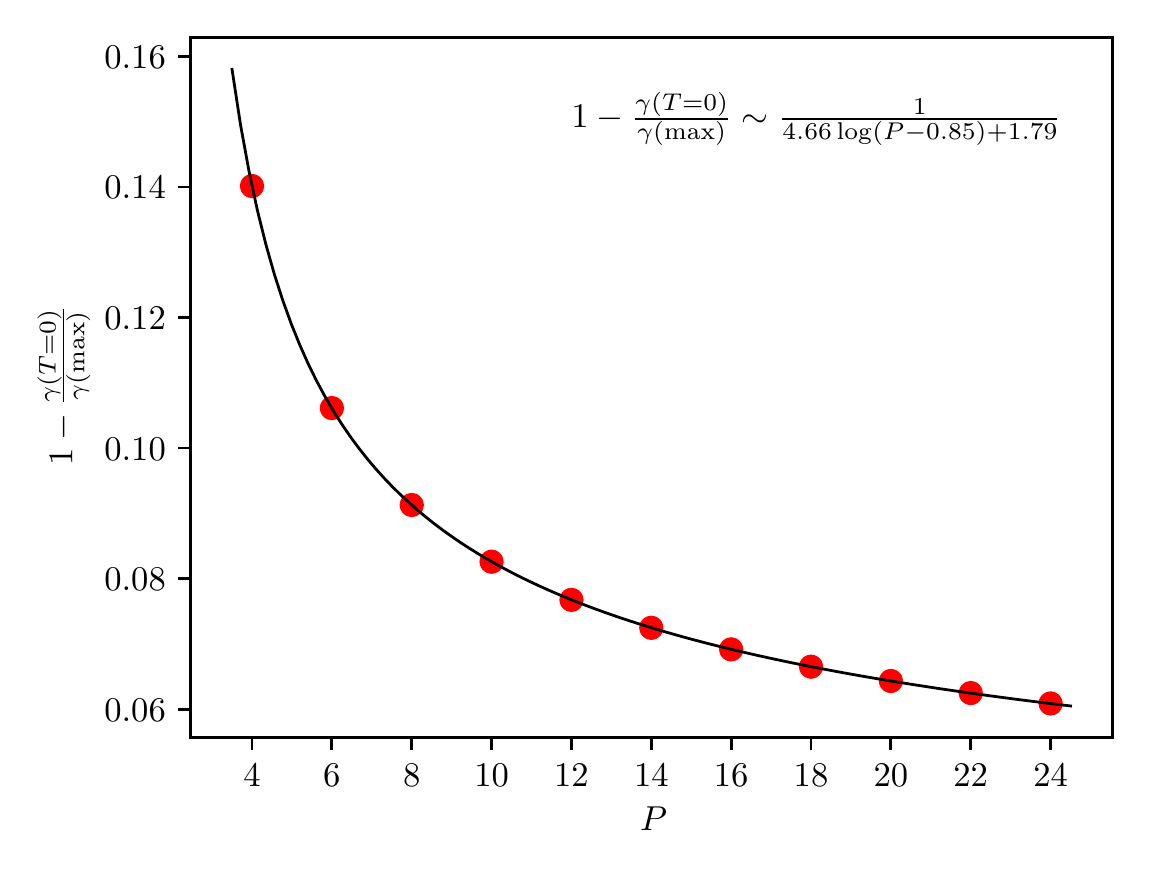}
     \end{subfigure}
     \begin{subfigure}[b]{0.5\textwidth}
         \centering
         \includegraphics[width=\textwidth]{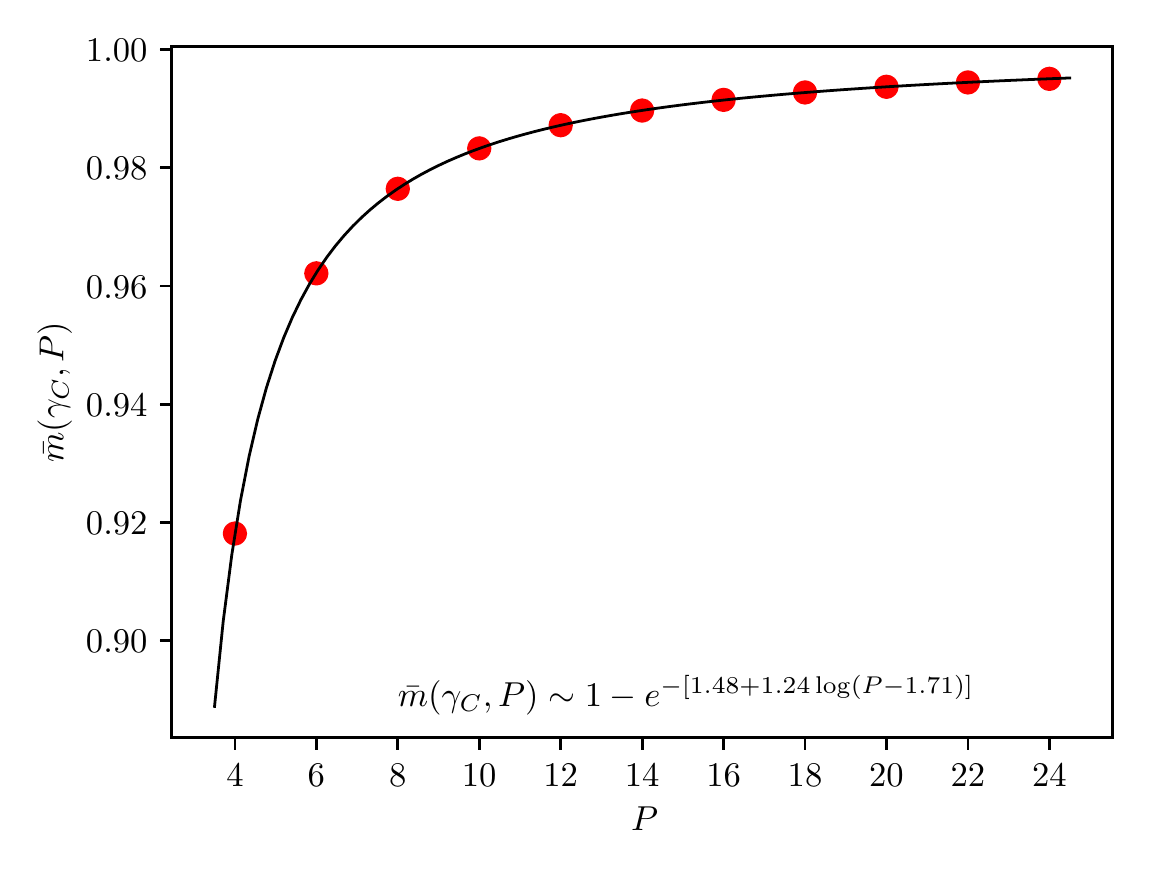}
     \end{subfigure}
        \caption{Left: Instability region w.r.t. $P$; we notice a strong reduction when $P$ increases. Right: Values of magnetization $\bar{m}$ w.r.t. P when we consider the critical capacity $\gamma_C$; we show that $\bar{m}$ reaches $1$ as P increases.}
        \label{fig:trend_mag_RS0}
\end{figure}

\subsection{1-RSB approximation}

In this subsection we turn to the solution of the quenched statistical pressure of the dense associative networks under the first step of replica symmetry breaking (1-RSB).
\newline
In the  1-RSB setting the probability distributions of the two overlaps $q$ and $p$ (see eqs. (\ref{limforq2}) and (\ref{limforp2}) respectively) display an analogous multi-modal structure as captured by the next

\begin{definition} \label{def:HM_RSB}
In the first step of replica-symmetry breaking (1-RSB), the distribution of the two-replica overlap $q$, in the thermodynamic limit, displays two delta-peaks at the equilibrium values, referred to as $\bar{q}_1,\ \bar{q}_2$, and the concentration on the two values is ruled by $\theta \in [0,1]$, namely
\begin{equation}
\lim_{N \rightarrow + \infty} P'_N(q) = \theta \delta (q - \bar{q}_1) + (1-\theta) \delta (q - \bar{q}_2). \label{limforq2}
\end{equation}
Similarly, for the overlap $p$, denoting with $\bar{p}_1,\ \bar{p}_2$ the equilibrium values, we have
\begin{equation}
\lim_{N \rightarrow + \infty} P''_N(p) = \theta \delta (p - \bar{p}_1) + (1-\theta) \delta (p - \bar{p}_2). \label{limforp2}
\end{equation}
The Mattis magnetization $m$ still self-averages at $ \bar{m}$ as in (\ref{eq:m_ter}).
\end{definition}
Note that, strictly speaking, the above ansatz for the overlaps is not the original Parisi one  (that holds for pure spin glasses, e.g. the Sherrington-Kirkpatrick model \cite{Guerra,Talagrand}), but its straightforward generalization, named {\em ziqqurat ansatz} for obvious reasons in \cite{Ziguli1,Ziguli2}. 

Following the same route pursued in the previous sections, we need an interpolating partition function $\mathcal{Z}$ and an interpolating quenched pressure $\mathcal{A}^{(P)}$,  that are defined hereafter.
\begin{definition}
Given the interpolating parameters $\bm r = (x^{(1)}, x^{(2)}, y^{(1)}, y^{(2)}, w,z), t$ and the i.i.d. auxiliary fields $\{J_i^{(1)}, J_i^{(2)}\}_{i=1,...,N}$, with $J_i^{(1,2)} \sim \mathcal N(0,1)$ for $i=1, ..., N$ and $\{\tilde J_{\mu}^{(1)}, \tilde J_{\mu}^{(2)}\}_{\mu=2,...,P}$, with $J_{\mu}^{(1,2)} \sim \mathcal N(0,1)$ for $\mu=2,...,P$, we can write the 1-RSB interpolating partition function $\mathcal Z_N(t, \boldsymbol r)$ for the dense associative network (\ref{eq:hop_hbare}) recursively, starting by

\begin{equation}
\label{eqn:Z2}
\begin{array}{lll}
     \mathcal{Z}^{(P)}_2(t, \bm r) &:=& \sommaSigma{\boldsymbol \sigma} \displaystyle\int \mathcal{D} \bm \tau\exp{}\Bigg[t\dfrac{\beta\,'\,N}{2} m^P(\boldsymbol \sigma)+wN\psi\,m(\boldsymbol \sigma)
     \\\\
     &\textcolor{white}{s} &+\sqrt{t}\sqrt{\dfrac{\beta\,' }{N^{P-1}}}\SOMMA{\mu>1}{K}\,\left(\SOMMA{i_1,\cdots,i_{_{P/2}}=1}{N,\cdots,N}\xi^{\mu}_{i_1\cdots,i_{_{P/2}}}\sigma_{i_1}\cdots\sigma_{_{P/2}}\right)\tau_{\mu} -\dfrac{\beta'\gamma}{2}N^{a-P/2}
     \\\\
     &\textcolor{white}{s} &+\SOMMA{a=1}{2}\left(\sqrt{N^{^{1-P/2}}x^{(a)}}\SOMMA{\mu>1}{K} \tilde{J\,}_{\mu}^{(a)}\tau_{\mu}+\sqrt{y^{(a)}}\SOMMA{i=1}{N} J_i^{(a)}\sigma_i\right)+\dfrac{zN^{1-P/2}}{2}\SOMMA{\mu>1}{K}\,\tau^2_{\mu}\Bigg]\,,
\end{array}
\end{equation}
where  the $\xi^{\mu}_{i_1\cdots,i_{_{P/2}}}$'s are i.i.d. standard Gaussians. 
Averaging out the fields recursively, we define
\begin{align}
\label{eqn:Z1_trans}
\mathcal Z_1^{(P)}(t, \bm r) \coloneqq& \mathbb E_2 \left [ \mathcal Z_2^{(P)}(t, \bm r)^\theta \right ]^{1/\theta} \\
\label{eqn:Z0_trans}
\mathcal Z_0^{(P)}(t, \bm r) \coloneqq&  \exp \mathbb E_1 \left[ \ln \mathcal Z_1^{(P)}(t, \bm r) \right ] \\
\mathcal Z_N^{(P)}(t, \boldsymbol r) \coloneqq& \mathcal Z_0^{(P)}(t, \bm r) ,
\end{align}
where with $\mathbb E_a$ we mean the average over the variables $J_i^{(a)}$'s and $\tilde J_\mu^{(a)}$'s, for $a=1, 2$, and with $\mathbb{E}_0$ we shall denote the average over the variables $\xi^{\mu}_{i_1\cdots,i_{_{P}}}$'s.
\end{definition}

\begin{definition}
\label{def:interpPressRSB}
The 1-RSB interpolating pressure of the DHN, at finite volume $N$, is introduced as
\begin{equation}\label{AdiSK1RSB}
\mathcal A_N^{(P)} (t) \coloneqq \frac{1}{N}\mathbb E_0 \big[ \ln \mathcal Z_0^{(P)}(t) \big],
\end{equation}
and, in the thermodynamic limit $\mathcal A^{(P)} (t) \coloneqq \lim_{N \to \infty} \mathcal A^{(P)}_N (t)$.
\newline
Note that by setting $t=1$, the interpolating pressure recovers the standard pressure (\ref{PressureDef}), that is, $A_N(\beta, \gamma) = \mathcal A^{(P)}_N (t =1)$.
\end{definition}

\begin{remark}
\label{rem:medie}
In order to lighten the notation, hereafter we use the following 
\begin{align}
\label{eq:unouno_a}
\langle m \rangle=& \mathbb E_0  \mathbb E_1  \mathbb E_2 \left[\mathcal W_2\frac{1}{N}\sum_{i=1}^N \omega(\xi_i \sigma_i) \right] \\
 \langle p_{11} \rangle=& \mathbb E_0  \mathbb E_1  \mathbb E_2 \left[ \mathcal W_2\frac{1}{P}\sum_{\mu=1}^P \omega(\tau_\mu^2) \right]\\
 \langle p_{12} \rangle_1=& \mathbb E_0  \mathbb E_1   \left[ \frac{1}{P}\sum_{\mu=1}^P \left( \mathbb E_2  \left [\mathcal W_2\omega(\tau_\mu)\right] \right)^2 \right] \\
\langle p_{12} \rangle_2=& \mathbb E_0  \mathbb E_1  \mathbb E_2 \left[ \mathcal W_2\frac{1}{P}\sum_{\mu=1}^P \omega(\tau_\mu)^2 \right] \\
\label{eqn:q121_a}
 \langle q_{12} \rangle_1=&\mathbb E_0  \mathbb E_1  \left[\frac{1}{N} \sum_{i=1}^N \left( \mathbb E_2 \left[\mathcal W_2\omega(\sigma_i)\right] \right)^2 \right] \\
\label{eqn:q122_a}
\langle q_{12} \rangle_2=& \mathbb E_0  \mathbb E_1  \mathbb E_2 \left [ \mathcal W_2\frac{1}{N}\sum_{i=1}^N \omega(\sigma_i)^2 \right] 
\end{align}
where the weight $\mathcal W_2$ is defined as
\begin{equation}
\mathcal W_2 :=\frac{{{\mathcal {Z}}_2^{(P)}}^\theta}{\mathbb E_2 \left [{{\mathcal {Z}}_2^{(P)}}^\theta\right]}.
\end{equation}
Furthermore, we define the Boltzmann factor $\mathcal B (\bm \sigma, \bm \tau ; t, \bm r)$ similarly to RS assumption. 
\end{remark}

The next step is building a transport equation for the interpolating quenched pressure, for which we preliminary need to evaluate the related partial derivatives, as discussed in the next
\begin{lemma} \label{lemma:4}
The partial derivative of the interpolating quenched pressure with respect to a generic variable $\rho$ reads as
%
%
%
\begin{equation}
\label{eqn:partialrA}
\frac{\partial }{\partial \rho} \mathcal A_N^{(P)}(t, \bm r)=\frac{1}{N} \mathbb E_0  \mathbb E_1  \mathbb E_2 \left[\mathcal W_2 \omega \big( \partial_\rho \mathcal B(\bm \sigma,  \bm \tau; t, \bm r) \big)\right].
\end{equation}
In particular,
\begin{align}
\label{eqn:partialtA}
\frac{\partial }{\partial t} \mathcal A^{(P)}_N =&\frac{\beta^{'}}{2}\langle m_1^P \rangle+ \frac{\beta^{'} K}{2N^{P/2}}\big ( \langle p_{11} \rangle -(1-\theta)\langle p_{12}q_{12}^{\frac{P}{2}} \rangle_2-\theta\langle p_{12}q_{12}^{\frac{P}{2}} \rangle_1 \big) \\
\label{eqn:partialx1A}
\frac{\partial }{\partial x^{(1)}} \mathcal A^{(P)}_N =& \frac{K}{2N^{P/2}}\big ( \langle p_{11} \rangle  -(1-\theta)\langle p_{12} \rangle_2-\theta\langle p_{12} \rangle_1 \big) \\
\label{eqn:partialx2A}
\frac{\partial }{\partial x^{(2)}} \mathcal A^{(P)}_N =& \frac{K}{2N^{P/2 }}\big ( \langle p_{11} \rangle  -(1-\theta)\langle p_{12} \rangle_2 \big) \\
\label{eqn:partialy1A}
\frac{\partial }{\partial y^{(1)}} \mathcal A^{(P)}_N =& \frac{1}{2}\big ( 1-(1-\theta)\langle q_{12} \rangle_2-\theta\langle q_{12} \rangle_1 \big) \\
\label{eqn:partialy2A}
\frac{\partial }{\partial y^{(2)}} \mathcal A^{(P)}_N =& \frac{1}{2}\big ( 1 -(1-\theta)\langle p_{12} \rangle_2
\big)\\
\label{eqn:partialzA}
\frac{\partial }{\partial z} \mathcal A^{(P)}_N =& \frac{K}{2N^{P/2}} \langle p_{11} \rangle \\
\label{eqn:partialwA}
\frac{\partial }{\partial w}  \mathcal A^{(P)}_N =&  \langle m_1 \rangle
\end{align}
\end{lemma}
\begin{proof}
The proof is pretty lengthy and basically requires just standard calculations, so it is left for the Appendix \ref{app2}. Here we just prove that, in  complete generality
\begin{align}
\label{eqn:partialrA1}
\frac{\partial }{\partial \rho}  \mathcal A^{(P)}_N(t, \bm r)=& \frac{1}{N} \mathbb E_0 \mathbb E_1 \bigg [\partial_\rho \ln\mathcal Z_1^{(P)} \bigg] \nonumber \\
=&\frac{1}{N} \mathbb E_0 \mathbb E_1 \bigg[\frac{1}{\theta}\frac{1}{\mathcal Z_1^{(P)}} \big[ {{\mathcal {Z}}_2^{(P)}}^\theta\big]^{1/\theta-1} \mathbb E_2 \big[\partial_\rho{{\mathcal {Z}}_2^{(P)}}^\theta \big] \bigg] \nonumber \\
=&\frac{1}{N} \mathbb E_0 \mathbb E_1 \mathbb E_2\bigg[\frac{{{\mathcal {Z}}_2^{(P)}}^\theta}{\mathbb E_2 {{\mathcal {Z}}_2^{(P)}}^\theta }\frac{\partial_\rho {{\mathcal {Z}}_2^{(P)}}}{\mathcal Z_2^{(P)}} \bigg]  \nonumber\\
=&\frac{1}{N} \mathbb E_0 \mathbb E_1 \mathbb E_2 \bigg[\mathcal W_2\frac{\partial_\rho \mathcal Z_2^{(P)}}{\mathcal Z_2^{(P)}} \bigg].
\end{align}
\end{proof}

\begin{remark}
As in replica symmetric case, in the next computations we can use the following relations for $a=1,2$
\begin{align}
    \langle m_1^P \rangle &= \sum_{k=2}^P \begin{pmatrix}P\\k\end{pmatrix} \langle (m_1-\bar{m})^k \rangle \bar{m}^{P-k} + \bar{m}^P (1-P) + P\bar{m}^{P-1}\langle m_1 \rangle\,, \label{potential_m_1RSB}
    \\
    \langle p_{12}q_{12}^{P/2}\rangle_a &= \sum_{k=1}^{P/2} \begin{pmatrix}\frac{P}{2}\\k\end{pmatrix}\bar{q}_a^{P/2-k} \langle (p_{12}-\bar{p}_a)(q_{12} - \bar{q}_a)^k \rangle_a + \sum_{k=2}^{P/2} \begin{pmatrix}\frac{P}{2}\\k\end{pmatrix} \bar{q}_a^{P/2-k} \bar{p}_a\langle (q_{12} - \bar{q}_a)^k \rangle_a+ \notag
    \\
    &\textcolor{white}{=}+\q_a^{P/2}\l\pp\r_a+\dfrac{P}{2}\q_a^{P/2-1}\p_a\l\qq\r_a-\dfrac{P}{2}\q_a^{P/2}\p_a\,;
    \label{potential_pq_1RSB}
\end{align}
\end{remark}

\begin{proposition}
\label{prop:9}
The $t$-streaming of the 1-RSB interpolating pressure obeys, at finite volume $N$, a standard transport equation, that reads as
\begin{align}
\label{eqn:transportequation}
\frac{d\mathcal A^{(P)}}{dt}=&\partial_t \mathcal A^{(P)}+\dot x^{(1)}\partial_{x_1} \mathcal A^{(P)} +\dot x^{(2)}\partial_{x_2} \mathcal A^{(P)} +\dot y^{(1)} \partial_{y_1} \mathcal A^{(P)} +\dot y^{(2)} \partial_{y_2} \mathcal A^{(P)} \notag \\
&+\dot z \partial_{z} \mathcal A^{(P)}+\dot w \partial_{w} \mathcal A^{(P)} = S(t, \bm r) + V_N(t, \bm r),
\end{align}
where the source $S(t, \bm r)$ and the potential $V(t, \bm r)$ read as
\begin{align}
\label{eqn:f}
S(t, \bm r)  \coloneqq& \frac{\beta^{'} \bar m^P (1-P)}{2}-{\beta^{'}\gamma (\theta-1)}\frac{P}{2}\bar{p}_2 \bar{q}_2^{P/2} +{\beta^{'}\gamma \theta}\frac{P}{2} \bar{p}_1 \bar{q}_1^{P/2} - {\beta^{'}\gamma }\frac{P}{2}\bar{p}_2 \bar{q}_2^{P/2-1} \\
\label{eqn:V}
V_N(t, \bm r) \coloneqq&  \frac{\beta^{'}K }{2N^{P/2}}\left\{ (\theta -1) \left[\sum_{k=2}^{P/2} \binom{P/2}{k} \bar{q}_2^{P/2 - k} \langle (p_{12} - \bar{p}_2)(q_{12} - \bar{q}_2)^k \rangle_2 + \right. \right. \notag \\
&\left. \left. +\sum_{k=2}^{P/2} \binom{P/2}{k} \bar{q}_2^{P/2 - k} \bar{p}_2 \langle (q_{12} - \bar{q}_2)^k \rangle_2 \right] - \theta \left[\sum_{k=2}^{P/2} \binom{P/2}{k} \bar{q}_1^{P/2 - k} \langle (p_{12} - \bar{p}_1)(q_{12} - \bar{q}_1)^k \rangle_1 \right. \right.\notag \\
&\left. \left.+ \sum_{k=2}^{P/2} \binom{P/2}{k} \bar{q}_1^{P/2 - k} \bar{p}_1 \langle (q_{12} - \bar{q}_1)^k \rangle_1 \right]\right\} + \frac{\beta^{'}}{2}\sum_{k=2}^P \binom{P}{k} \langle (m_1 - \bar{m})^k \rangle
\end{align}
\normalsize
\end{proposition}
The proof of the Proposition is provided in Appendix \ref{Appendix-Sua1}.

\begin{remark} \label{r:above}
In the thermodynamic limit, in the 1-RSB scenario, we have
\begin{align}
\label{eqn:thlimaveragem}
 \lim_{N\rightarrow \infty} \langle (m - \bar m)^2 \rangle =& 0\\
 \lim_{N\rightarrow \infty} \langle (q_{12}-\bar q_i)^2 \rangle_i=& 0; \: \: \: i=1,2\\
\label{eqn:thlimaveragep}
 \lim_{N\rightarrow \infty} \langle( p_{12}-\bar p_i)^2\rangle_i =& 0; \: \: \: i=1,2
\end{align}
Similar to the RS approximation, in the thermodynamic limit we have that the central moments greater than two tend to zero such that
\begin{equation} \label{eq:V0_HRSB}
\lim_{N \to \infty} V_N(t, \bm r) = 0.
\end{equation}
\end{remark}

Similar to Theorem \ref{cor_carmassimoRS}, we have the following
\begin{theorem}\label{Susy1}
In the thermodynamic limit, under one-step of replica symmetry breaking, the maximum storage of the dense Hebbian network scales as $K \propto N^{P-1}$, i.e. $a=P-1$.
\newline
In this regime of maximal storage, i.e. in the Baldi-Venkatesh limit, the quenched statistical pressure for even $P\geq 4$ becomes
\begin{align}
    &\mathcal{A}^{(P)}=\ln 2 + \frac{1}{\theta}\mathbb{E}_1 \ln \mathbb{E}_2 \cosh^\theta g (\bm J, \bar{m}) -\dfrac{\gamma\b}{4}\q_2^{P/2-1}\p_2\Big(P-(P-1)\q_2\Big) \notag \\
    &+ \frac{\beta^{'}}{2}\bar{m}^P (1-P) -\theta(P-1)\dfrac{\b\gamma}{4}(\q_2^{P/2}\p_2-\q_1^{P/2}\p_1) + \frac{1}{4}{\beta^{'}}^2 \gamma
    \label{A_1RSB_finalissima_trans}
\end{align}
where 
\begin{align}
    g(\bm J, \bar{m})&= \frac{\beta^{'}P}{2}\bar{m}^{P-1} + J^{(1)}\sqrt{\frac{\beta^{'}}{2}\gamma \bar{p}_1P \bar{q}_1^{P/2-1}} +J^{(2)}\sqrt{ \frac{\beta^{'}}{2} P\gamma \left[ \bar{p}_2 \bar{q}_2^{P/2-1} -\bar{p}_1 \bar{q}_1^{P/2-1}  \right]}
\end{align}
\end{theorem}
The proof is provided in Appendix \ref{SusyBreak}.
  
\begin{remark}
The above 1-RSB quenched statistical pressure, with $P=2$ and $a=1$ in \eqref{pressfinaleGuerraRSB} -namely the solution of standard Hopfield model under one step of replica symmetry breaking, coincides with that predicted heuristically by  Crisanti, Amit and Gutfreund \cite{Crisanti2, lindaRSB}.
\end{remark}
\begin{figure}[h!]
    \centering
    \includegraphics[scale=0.7]{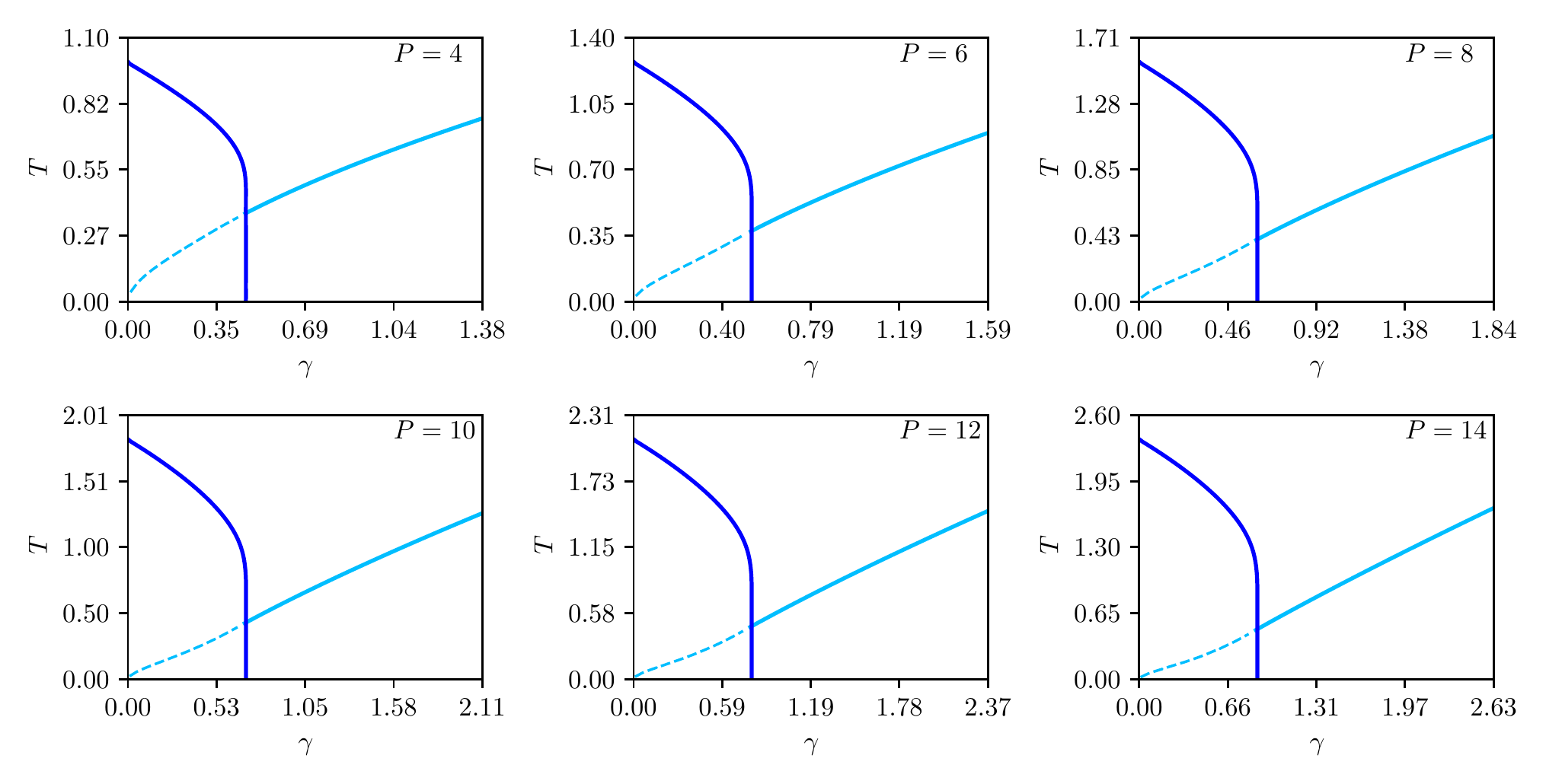}
    \caption{Broken replica symmetry (1-RSB) phase diagram of the dense associative network as different values of $P$. The dark blue phase transition identifies the retrieval region, while the light blue identifies the spin-glass region.  We stress that -outside the retrieval region- as P grows the spin-glass region gets stable in the RSB picture (while it shrinks to zero in the RS scenario). Inside the retrieval region the pure spin glass solutin is always unstable and it is detached by  a light blue dotted line.}
    \label{fig:diagRS}
\end{figure}

\begin{figure}[h!]
     \begin{subfigure}[b]{0.5\textwidth}
         \centering
         \includegraphics[width=\textwidth]{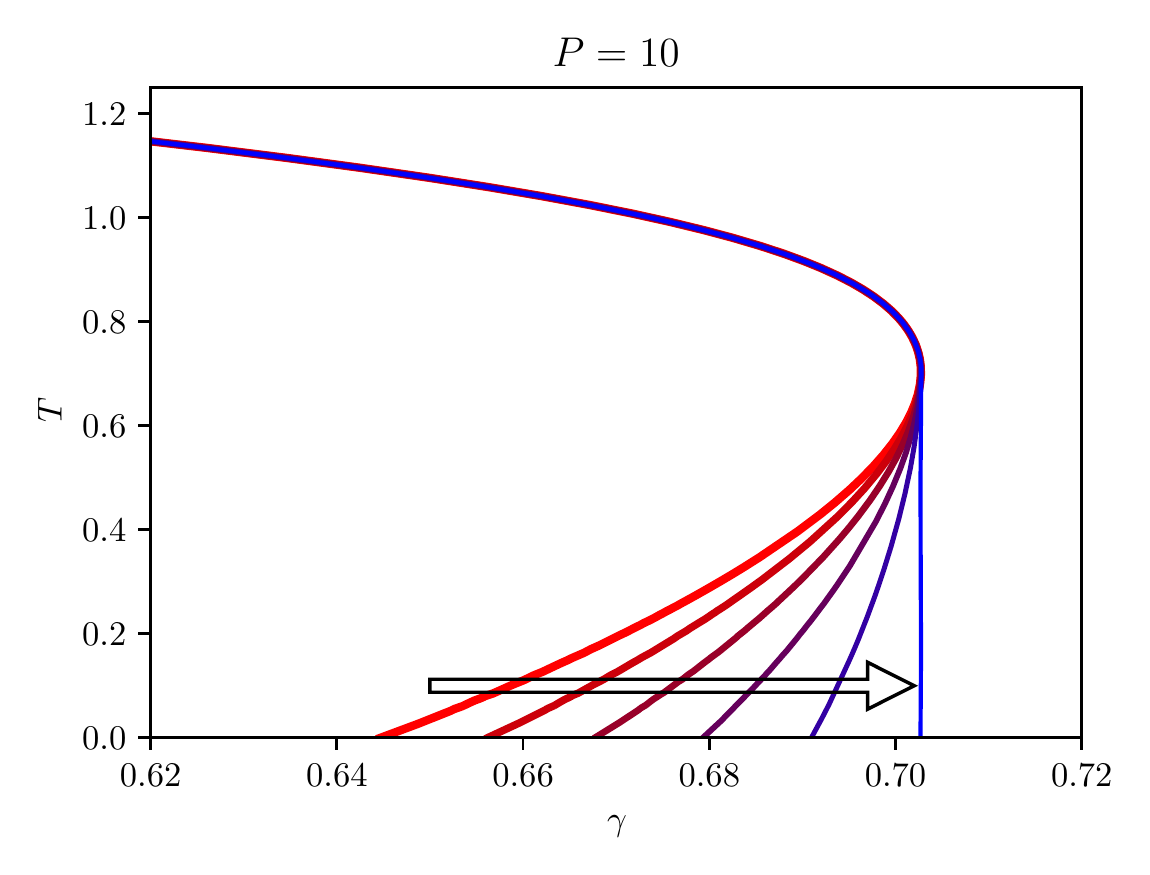}
     \end{subfigure}
     \begin{subfigure}[b]{0.5\textwidth}
         \centering
         \includegraphics[width=\textwidth]{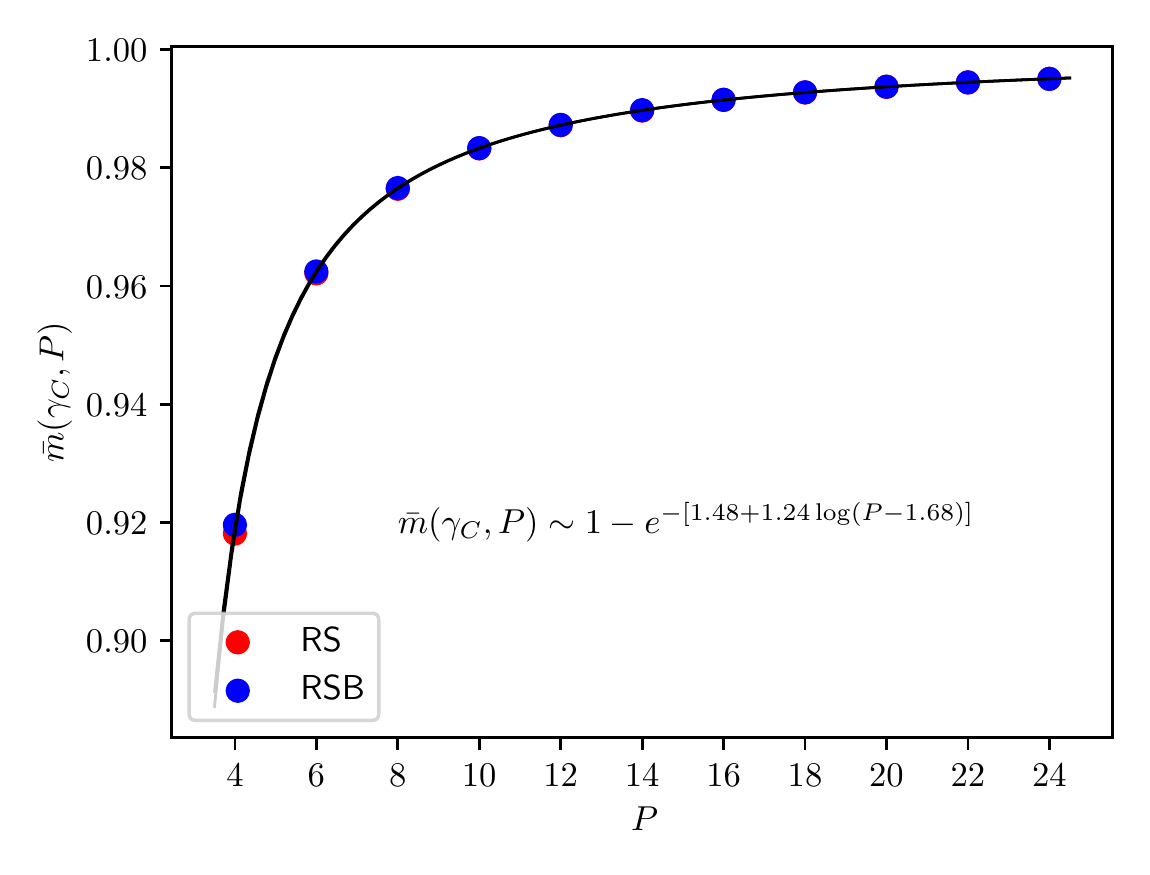}
     \end{subfigure}
        \caption{Left: Super-position of phase diagrams in $P=10$ case for RS (red) and 1RSB (blue) assumption. We highlight the fading of instability region in 1RSB case. Right: Values of magnetization $\bar{m}$ w.r.t. P when we consider the critical capacity $\gamma_C$; we note that the values of the magnetization in the RS and 1-RSB regimes coincide and as P increases, suggesting that the smaller the $P$ the stronger the effect of RSB in the network.}
        \label{fig:trend_mag_RS_2}
\end{figure}

\begin{figure}[h!]
     \begin{subfigure}[b]{0.5\textwidth}
         \centering
         \includegraphics[width=1.05\textwidth]{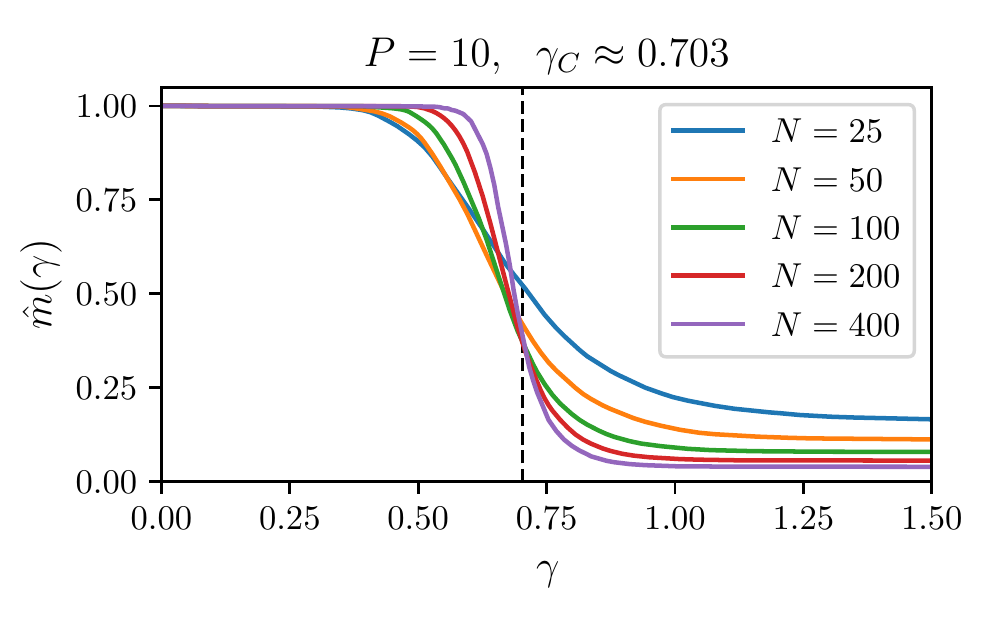}
     \end{subfigure}
     \begin{subfigure}[b]{0.5\textwidth}
         \centering
         \includegraphics[width=1.05\textwidth]{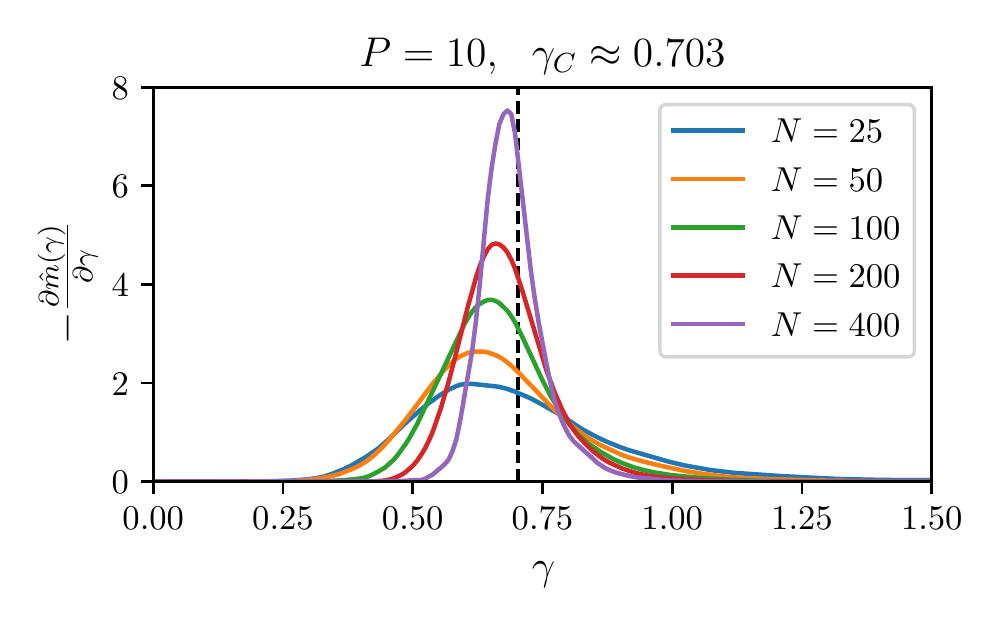}
     \end{subfigure}
        \caption{Monte Carlo numerical checks for a dense network with $P=10$: we highlight the agreement among simulations (colored lines report different simulation sizes, to facilitate a visual finize size scaling) and theory (reported as a vertical dashed bar). Left: Mattis magnetization. Right: Susceptibility (as a response function in $\gamma$).}
\end{figure}

By extremizing the quenched statistical pressure in \eqref{A_1RSB_finalissima_trans} w.r.t. the order parameters we can state the following
\begin{corollary}
\label{cor:SC_HOP_1RSB}
The self-consistent equations for the order parameters, under one step of replica symmetry breaking, read as
\begin{align}
&\bar{m}= \mathbb{E}_1 \left[ \frac{\mathbb{E}_2 \cosh^\theta g(\bm J,\bar{m})\tanh g(\bm J,\bar{m})}{\mathbb{E}_2 \cosh^\theta g(\bm J,\bar{m})}\right] \\
&\bar{p}_1= \beta^{'}\bar{q}_1^{P/2} \\
&\bar{p}_2 = \beta^{'}\bar{q}_2^{P/2}\\ \label{Senza1RSBa}
&\bar{q}_1 = \mathbb{E}_1 \left[ \frac{\mathbb{E}_2 \cosh^\theta g(\bm J,\bar{m})\tanh g(\bm J,\bar{m})}{\mathbb{E}_2 \cosh^\theta g(\bm J,\bar{m})}\right]^2 \\ \label{Senza1RSBb}
&\bar{q}_2 = \mathbb{E}_1 \left[ \frac{\mathbb{E}_2 \cosh^\theta g(\bm J,\bar{m})\tanh^2 g(\bm J,\bar{m})}{\mathbb{E}_2 \cosh^\theta g(\bm J,\bar{m})}\right]
\end{align}
where 
\begin{equation}
\begin{array}{lll}
     g(\bm J,\bar{m})&&  \dfrac{\beta^{'}P}{2}\bar{m}^{P-1} + \sqrt{\dfrac{\beta^{'}\gamma P\bar{p}_1 \bar{q}_1^{P/2-1}}{2}} J^{(1)} +\sqrt{\dfrac{\beta^{'}\gamma P\Big(\bar{p}_2 \bar{q}_2^{P/2-1}-\bar{p}_1 \bar{q}_1^{P/2-1}\Big)}{2}} J^{(2)}. 
\end{array}
\end{equation}
\end{corollary}

Now we turn to the other mathematical technique, namely in the next Section we obtain the above formulas for the quenched statistical pressure (both at the RS and 1-RSB level of approximation) via an adaptation of the Guerra's interpolation technique. Once these mathematical techniques will be exposed, we turn to understanding the information processing capabilities of these dense networks in the second part of the paper.

\section{Second approach: Guerra's interpolation technique}\label{Sezione3}
As stated, in this section we re-obtain the results achieved by the transport equation technique, this time through a suitable generalization of Guerra's interpolation technique, either in RS and in 1-RSB assumptions. 

\subsection{RS approximation}
The definition of RS assumption for the order parameters is the same as Definition \ref{defn: RSassumption}.
\begin{definition} 
Given the interpolating parameter $t \in [0,1]$, $A,\ B, \ C, \ \psi \in \mathbb{R}$  and $J_i$, $\tilde J_\mu \sim \mathcal{N}(0,1)$ for $i=1, \hdots , N$ and $\mu = 1, \hdots , K$  standard Gaussian variables i.i.d., the partition function is given as 
\begin{equation}
\footnotesize
\begin{array}{lll}
     \mathcal{Z}^{(P)}_N(t) &\coloneqq& \sommaSigma{\boldsymbol \sigma} \displaystyle\int \mathcal{D} \bm \tau\exp{}\Bigg[t\dfrac{\beta\,'\,N}{2} m^P(\boldsymbol \sigma)+(1-t)N\psi\,m(\boldsymbol \sigma)+
     \\\\
     & &+\sqrt{t}\sqrt{\dfrac{\beta\,' }{N^{P-1}}}\SOMMA{\mu>1}{K}\,\left(\SOMMA{i_1,\cdots,i_{_{P/2}}=1}{N,\cdots,N}\xi^{\mu}_{i_1\cdots,i_{_{P/2}}}\sigma_{i_1}\cdots\sigma_{_{P/2}}\right)\tau_{\mu}+
     \\\\
     & &+\sqrt{1-t}\left(A\SOMMA{\mu>1}{K} \tilde{J}_{\mu}\tau_{\mu}+B\SOMMA{i=1}{N} J_i\sigma_i\right)+\dfrac{1-t}{2}C\SOMMA{\mu>1}{K}\,\tau^2_{\mu}
-\dfrac{\beta '\gamma}{2}N^{a-P/2}\Bigg]\,,
     \label{def:partfunct_GuerraRS}
\end{array}
\end{equation}
where, for any $\mu=2,...,K$, $\tau_{\mu} \sim \mathcal N [0, 1]$ and $\mathcal{D} \bm \tau \coloneqq \prod\limits_{\mu=1}^K \frac{e^{- \tau_{\mu}^2/2}}{\sqrt{2\pi} }$ 
is the related measure and we set $\beta'=2 \beta/P!$.
\end{definition}

Similar to RS transport equation method, we can define the interpolating pressure, the Boltzmann factor and the generalized measure. 

\begin{lemma} 
\label{lemma:tderRS}
The $t$ derivative of interpolating pressure is given by 
\begin{equation}
    \begin{array}{lll}
         \dfrac{d \mathcal{A}^{(P)}(t)}{d t}\coloneqq &  \dfrac{\beta '}{2}\l m_1^P \r - \psi \l m_1 \r-\dfrac{1}{2}B^{\2}+\l\ppp\r\dfrac{K}{2 N}\left(\dfrac{\beta^{'}}{ N^{P/2}}-A^{\2}-C\right)+
         \\\\
         \textcolor{white}{\coloneqq}& -\dfrac{\beta '}{2N}\dfrac{K}{ N^{^{P/2-1}}}\Big[\l\pp\qq^{P/2}\r-\dfrac{ N^{P/2-1}}{\beta '}A^{\2}\l\pp\r-\dfrac{N^{P/2}}{ \beta ' K}B^{\2}\l\qq\r \Big]\,.
    \end{array}
    \label{eq:streaming_RS_Guerra}
\end{equation}
\end{lemma}

Since the computation is similar to that of derivative w.r.t. interpolating parameters of transport equation, we omit it. 

\begin{remark}
We stress that, for the RS assumption, we can use the relations \eqref{potential_m}, \eqref{potential_pq}.
Using these, if we fix the four constants as
\begin{equation}
    \begin{array}{lll}
         \psi=\dfrac{P}{2}\beta '  \m^{^{P-1}}\,,
         \\\\
         A^2=\dfrac{\beta '}{N^{^{P/2-1}}}\q^{^{P/2}}\,,
         \\\\
         B^2=\beta ' \gamma N^{^{a-P/2}}  \dfrac{P}{2}\, \p \q^{^{P/2-1}}\,,
         \\\\
         C= \dfrac{\beta '}{N^{^{P/2-1}}}(1-\q^{^{P/2}})\,,
    \end{array}
    \label{eq:constant_GuerraRS}
\end{equation}
and we remember Definition \ref{defn: RSassumption}, the \eqref{eq:streaming_RS_Guerra} at finite size $N$, is 
\begin{align}
	\frac{d\mathcal{A}^{(P)}(t)}{dt} &= -\dfrac{P-1}{2} \beta ' \m^P -\dfrac{\beta ' \gamma}{4} P \p \q^{P/2-1}(1-\q) + \frac{\beta^{'}}{2}\sum_{k=2}^P \begin{pmatrix}P\\k\end{pmatrix} \langle (m_1-\bar{m})^k \rangle \bar{m}^{P-k}  \notag \\
    &\textcolor{white}{sp} - \frac{\beta^{'}K}{2N^{P/2}}\sum_{k=1}^{P/2} \begin{pmatrix}\frac{P}{2}\\k\end{pmatrix}\bar{q}^{P/2-k} \langle (p_{12}-\bar{p})(q_{12} - \bar{q})^k \rangle,
\end{align}
which is independent of $t$.
\end{remark}

Applying the Fundamental Theorem of Calculus we claim the following 
\begin{proposition}
At finite size and under RS assumption applying the Fundamental Theorem of Calculus and using the suitable values of $A, B, C, \psi$, we find the quenched pressure for the P spin Hopfield model as
\begin{align}
\label{A_RS_finite}
    \mathcal{A}^{(P)} &= \ln{2} -\dfrac{\beta' \gamma}{2}N^{a-P/2} +\left\langle\ln{\cosh{\left[\dfrac{P}{2}\beta '  \m^{P-1}+Y\sqrt{\beta ' \gamma  \dfrac{P}{2}\, \p \q^{^{P/2-1}}} \right]}}\right\rangle_Y \notag \\
    &-\dfrac{\gamma N^{^{a-1}}}{2}\ln{\left(1-\beta 'N^{^{1-P/2}}\left(1-\q^{^{P/2}}\right)\right)}+\dfrac{\gamma N^{a-P/2}}{2}\dfrac{\beta '\q^{^{P/2}}}{1-\beta 'N^{^{1-P/2}}\left(1-\q^{^{P/2}}\right)} \notag \\
    &-\dfrac{P-1}{2} \beta ' \m^P -\dfrac{\beta ' \gamma}{4} P \p \q^{P/2-1}(1-\q) + \frac{\beta^{'}}{2}\sum_{k=2}^P \begin{pmatrix}P\\k\end{pmatrix} \langle (m_1-\bar{m})^k \rangle \bar{m}^{P-k} \notag \\
    & - \frac{\beta^{'}K}{2N^{P/2}}\sum_{k=1}^{P/2} \begin{pmatrix}\frac{P}{2}\\k\end{pmatrix}\bar{q}^{P/2-k} \langle (p_{12}-\bar{p})(q_{12} - \bar{q})^k \rangle 
\end{align}
\normalsize
\end{proposition}

\begin{theorem}
The derivative w.r.t. $t$ in the thermodynamical limit is 
\begin{align}
\label{dt_thermo_RS}
\frac{d\mathcal{A}^{(P)}(t)}{dt} &= -\dfrac{P-1}{2} \beta ' \m^P -\dfrac{\beta ' \gamma}{4} P \p \q^{P/2-1}(1-\q). 
\end{align}

Thus, in the thermodynamic limit and under the assumption of replica symmetry, we reach the same results we computed via transport equation's interpolation (see equation \ref{eq:pressure_GuerraRS}), namely the quenched statistical pressure for $P\geq 4$  of the DHN becomes
\begin{equation}
\begin{array}{lll}
      \mathcal{A}^{(P)} (\gamma, \beta) & \coloneqq & \ln{2}+\left\langle\ln{\cosh{\left[\dfrac{P}{2}\beta '  \m^{P-1}+Y\sqrt{\beta ' \gamma   \dfrac{P}{2}\, \p \q^{^{P/2-1}}} \right]}}\right\rangle_Y -\dfrac{P-1}{2} \beta ' \m^P \\\\
      & &-\beta ' \gamma\dfrac{P}{4}  \, \p \q^{P/2-1}(1-\q)+\dfrac{1}{4} \gamma\beta '^2\left(1 - \q^P\right)\,.
\end{array}
\label{eq:pressure_TransRS}
\end{equation}
\end{theorem}

\begin{proof}
Thanks to replica symmetry assumption Definition \eqref{defn: RSassumption} we have $\langle \Delta p_{12} \Delta q_{12}^k \rangle \rightarrow 0$ and $\langle \Delta m^k\rangle \rightarrow 0$ for $k \geq 2$, so the derivative w.r.t. $t$ becomes as in \eqref{dt_thermo_RS}.

If we apply the Fundamental Theorem in the thermodynamical limit with \eqref{dt_thermo_RS} we recover 
\begin{equation}
\begin{array}{lll}
      \mathcal{A}^{(P)}(\gamma, \beta) & \coloneqq &\ln{2} -\dfrac{\beta' \gamma}{2}N^{a-P/2}+\left\langle\ln{\cosh{\left[\dfrac{P}{2}\beta '  \m^{P-1}+Y\sqrt{\beta ' \gamma  \dfrac{P}{2}\, \p \q^{^{P/2-1}}} \right]}}\right\rangle_Y+
      \\\\
      & & -\dfrac{P-1}{2} \beta ' \m^P -\dfrac{\beta ' \gamma}{4} P\,\p \q^{P/2-1}(1-\q)-\dfrac{\gamma N^{^{a-1}}}{2}\ln{\left(1-\beta 'N^{^{1-P/2}}\left(1-\q^{^{P/2}}\right)\right)}+
       \\\\
       && +\dfrac{\gamma N^{a-P/2}}{2}\dfrac{\beta '\q^{^{P/2}}}{1-\beta 'N^{^{1-P/2}}\left(1-\q^{^{P/2}}\right)}\,.
\end{array}
\end{equation}
which is the same expression in \eqref{eq:pressure_GuerraRS_noAPP}. The proof proceeds similarly to that of transport equation's interpolation. 
\end{proof}


\subsection{1-RSB approximation}

The ansatz for the concentration of the two-replica overlap  distributions (for both $p$ and $q$) is the same as in the Definition \eqref{def:HM_RSB} and the Mattis magnetization still self-averages around its mean $\bar{m}$, hence we can directly write the next   
\begin{definition}
Given the interpolating parameter $t$ and the i.i.d. auxiliary fields $\lbrace J_i^{(1)}, J_i^{(2)}\rbrace_{i=1,...,N}$, with $J_i^{(1,2)} \sim \mathcal N(0,1)$ for $i=1, ..., N$ and $\lbrace\tilde J_{\mu}^{(1)}, \tilde J_{\mu}^{(2)}\rbrace_{\mu=2,...,P}$, with $J_{\mu}^{(1,2)} \sim \mathcal N(0,1)$ for $\mu=2,...,P$, we can write the 1-RSB interpolating partition function $\mathcal Z_N^{(P)}(t)$ for the P spin Hopfield model (\ref{eq:hop_hbare}) recursively, starting by
\begin{equation}
\begin{array}{lll}
     \mathcal{Z}^{(P)}_2(t) &:=& \sommaSigma{\boldsymbol \sigma} \displaystyle\int \mathcal{D} \bm \tau\exp{}\Bigg[t\dfrac{\beta\,'\,N}{2} m^P(\boldsymbol \sigma)+(1-t)N\psi\,m(\boldsymbol \sigma)+
     \\\\
     & &+\sqrt{t}\sqrt{\dfrac{\beta\,' }{N^{P-1}}}\SOMMA{\mu>1}{K}\,\left(\SOMMA{i_1,\cdots,i_{_{P/2}}=1}{N,\cdots,N}\xi^{\mu}_{i_1\cdots,i_{_{P/2}}}\sigma_{i_1}\cdots\sigma_{_{P/2}}\right)\tau_{\mu}+
     \\\\
     & &+\sqrt{1-t}\SOMMA{a=1}{2}\left(A^{(a)}\SOMMA{\mu>1}{K} \tilde{J\,}^{(a)}_{\mu}\tau_{\mu}+B^{(a)}\SOMMA{i=1}{N} J_i^{(a)}\sigma_i\right)
     +\dfrac{1-t}{2}C\SOMMA{\mu>1}{K}\,\tau^2_{\mu}-\dfrac{\beta '\gamma}{2}N^{a-P/2}\Bigg]\,,
     \label{def:partfunct_GuerraRSB}
\end{array}
\end{equation}
where  the $\xi^{\mu}_{i_1\cdots,i_{_{P/2}}}$'s are i.i.d. standard Gaussians. The values of the real-valued constants $A_1, A_2, B_1$, $B_2, C$ will be set a fortiori (see the remark \ref{marco13}).
\newline 
Averaging out the fields recursively, we define
\begin{align}
\label{eqn:Z1_guerra}
\mathcal Z_1^{(P)}(t) \coloneqq& \mathbb E_2 \left [ \mathcal Z_2^{(P)}(t)^\theta \right ]^{1/\theta} \\
\label{eqn:Z0_guerra}
\mathcal Z_0^{(P)}(t) \coloneqq&  \exp \mathbb E_1 \left[ \ln \mathcal Z_1^{(P)}(t) \right ] \\
\mathcal Z_N^{(P)}(t) \coloneqq & \mathcal Z_0^{(P)}(t) ,
\end{align}
where with $\mathbb E_a$ we mean the average over the variables $J_i^{(a)}$'s and $\tilde J_\mu^{(a)}$'s, for $a=1, 2$, and with $\mathbb{E}_0$ we shall denote the average over the variables $\xi^{\mu}_{i_1\cdots,i_{_{P/2}}}$'s.
\end{definition}

The definition of 1-RSB interpolating pressure at finite volume $N$ and in the thermodynamic limit is the same as transport equation technique, see Definition \eqref{def:interpPressRSB} as well as the relative notation for the generalized averages. 
\par \medskip
Now the next step is computing the $t$-derivative of the interpolating pressure. In this way we can apply the fundamental theorem of calculus and find the solution of the original model, as standard in this type of approach \cite{GuerraSum}. 

\begin{lemma}
\label{lemma:tderRSB}
The derivative w.r.t. $t$ of interpolating pressure can be written as 
\begin{align}
    &d_t \mathcal{A}^{(P)}_N = \frac{\beta^{'}}{2} \langle m_1^P \rangle + \frac{\beta^{'} K }{2N^{P/2}} \left[ \langle p_{11} \rangle + (\theta -1) \langle p_{12} q_{12}^{P/2} \rangle_2 - \theta \langle p_{12} q_{12}^{P/2} \rangle_1 \right] \notag \\
    &- \left\{ C \frac{K}{2N} \langle p_{11} \rangle + \psi \langle m_1 \rangle  + \frac{K A_1}{2N}\left[ \langle p_{11} \rangle  -(1-\theta)\langle p_{12} \rangle_2-\theta\langle p_{12} \rangle_1 \right] + \frac{K A_2}{2N}\left[ \langle p_{11} \rangle  -(1-\theta)\langle p_{12} \rangle_2 \right] \right.\notag \\
    &\left.+ \frac{B_1}{2}\left[ 1-(1-\theta)\langle q_{12} \rangle_2-\theta\langle q_{12} \rangle_1\right] + \frac{B_2}{2} \left[ 1-(1-\theta)\langle q_{12} \rangle_2\right]\right\}
\end{align}
\end{lemma}
\normalsize
Since the proof is rather lengthy but similar to that of the $t$-streaming of the transport equation approach, we omit it for the sake of simplicity. 

\begin{remark}\label{marco13}
Following the 1-RSB ansatz and the combinatorial identities provided in Definition \ref{def:HM_RSB}, we have the expressions \eqref{potential_m_1RSB} and \eqref{potential_pq_1RSB}, thus, if we fix the costants  in the recursive partition function \ref{def:partfunct_GuerraRSB} as
\begin{align}
\label{value_psi}
    \psi&= \frac{\beta^{'}}{2}P \bar{m}^{P-1} \\
    A_1^2 &= \frac{\beta^{'}}{N^{P/2-1}}\bar{q}_1^{P/2} \\
    A_2^2 &= \frac{\beta^{'}}{N^{P/2-1}}(\bar{q}_2^{P/2}-\bar{q}_1^{P/2})\\
    B_1^2 &= \frac{\beta^{'}K}{N^{P/2}}P\bar{p}_1\bar{q}_1^{P/2-1}\\
    B_2^2 &= \frac{\beta^{'}K}{N^{P/2}}P(\bar{p}_2\bar{q}_2^{P/2-1}-\bar{p}_1\bar{q}_1^{P/2-1}) \\
    C &= \frac{\beta^{'}}{N^{P/2-1}}(1-\bar{q}_2^{P/2}) \label{value_C}.
\end{align}
we compute the derivative w.r.t. $t$ at finite size as 
\footnotesize
\begin{align}
&d_t \mathcal{A}_N^{(P)}=\left\{ \frac{\beta^{'}}{2}\left[\sum_{k=2}^P \begin{pmatrix}P\\k\end{pmatrix} \langle (m_1-\bar{m})^k \rangle \bar{m}^{P-k} + \bar{m}^P (1-P) \right] -{\beta^{'} \gamma}(\theta-1) \dfrac{P}{2}\q_2^{P/2}\p_2  \right. \notag \\
&+\frac{\beta^{'} K}{2N^{P/2}}(\theta-1) \left[ \sum_{k=1}^{P/2} \begin{pmatrix}\frac{P}{2}\\k\end{pmatrix}\bar{q}_2^{P/2-k} \langle (p_{12}-\bar{p}_2)(q_{12} - \bar{q}_2)^k \rangle_2 + \sum_{k=2}^{P/2} \begin{pmatrix}\frac{P}{2}\\k\end{pmatrix} \bar{q}_2^{P/2-k} \bar{p}_2\langle (q_{12} - \bar{q}_2)^k \rangle_2 \right]\notag \\
& \left. - \frac{\beta^{'} K}{2N^{P/2}}\theta \left[ \sum_{k=1}^{P/2} \begin{pmatrix}\frac{P}{2}\\k\end{pmatrix}\bar{q}_1^{P/2-k} \langle (p_{12}-\bar{p}_1)(q_{12} - \bar{q}_1)^k \rangle_1 + \sum_{k=2}^{P/2} \begin{pmatrix}\frac{P}{2}\\k\end{pmatrix} \bar{q}_1^{P/2-k} \bar{p}_1\langle (q_{12} - \bar{q}_1)^k \rangle_1 \right] +{\beta^{'} \gamma}\theta \dfrac{P}{2}\q_1^{P/2}\p_1 \right\}.
\end{align}
\normalsize
\end{remark}

Applying the Fundamental Theorem of Calculus and computing the one-body term, we have the following 
\begin{proposition}
At finite size $N$ and under the first step of replica symmetry breaking, we can write the quenched statistical pressure of the dense Hebbian network as
\footnotesize
\begin{align}
\label{A_1RSB_finite}
    &\mathcal{A}^{(P)} = \frac{\gamma N^{a-1}}{2}\ln (1-\beta^{'}N^{1-P/2}(1-\bar{q}_2^{P/2}))  +\frac{1}{\theta} \mathbb{E}_1 \left\{ \ln \mathbb{E}_2  \cosh^\theta\left( \psi + \sum_{a=1}^2 B^{(a)}J^{(a)}\right) \right\}+ \ln 2 \notag \\
&+\frac{ \gamma \beta^{'}N^{a-P/2}\q_1^{P/2} }{2 (1-\beta^{'}N^{1-P/2}(1-\q_2^{P/2})-\theta \beta^{'}N^{1-P/2}(\q_2^{P/2} - \q_1^{P/2}))}  \notag \\
&+ \frac{\gamma N^{a-1}}{2\theta} \ln\left(\frac{1-\beta^{'}N^{1-P/2}(1-\bar{q}_2^{P/2})}{1-\beta^{'}N^{1-P/2}(1-\bar{q}_2^{P/2})-\theta \beta^{'}N^{1-P/2}(\q_2^{P/2} - \q_1^{P/2}))}\right)-\dfrac{\beta '\gamma}{2}N^{a-P/2}\notag \\
&+\left\{ \frac{\beta^{'}}{2}\left[\sum_{k=2}^P \begin{pmatrix}P\\k\end{pmatrix} \langle (m_1-\bar{m})^k \rangle \bar{m}^{P-k} + \bar{m}^P (1-P) \right] -{\beta^{'} \gamma}(\theta-1) \dfrac{P}{2}\q_2^{P/2}\p_2 \right. \notag \\
&+\frac{\beta^{'} \gamma N^{a-P/2}}{2}(\theta-1) \left[ \sum_{k=1}^{P/2} \begin{pmatrix}\frac{P}{2}\\k\end{pmatrix}\bar{q}_2^{P/2-k} \langle (p_{12}-\bar{p}_2)(q_{12} - \bar{q}_2)^k \rangle_2 + \sum_{k=2}^{P/2} \begin{pmatrix}\frac{P}{2}\\k\end{pmatrix} \bar{q}_2^{P/2-k} \bar{p}_2\langle (q_{12} - \bar{q}_2)^k \rangle_2 \right]\notag \\
& \left. - \frac{\beta^{'} \gamma N^{a-P/2}}{2}\theta \left[ \sum_{k=1}^{P/2} \begin{pmatrix}\frac{P}{2}\\k\end{pmatrix}\bar{q}_1^{P/2-k} \langle (p_{12}-\bar{p}_1)(q_{12} - \bar{q}_1)^k \rangle_1 + \sum_{k=2}^{P/2} \begin{pmatrix}\frac{P}{2}\\k\end{pmatrix} \bar{q}_1^{P/2-k} \bar{p}_1\langle (q_{12} - \bar{q}_1)^k \rangle_1 \right] +{\beta^{'} \gamma}\theta \dfrac{P}{2}\q_1^{P/2}\p_1 \right\}
\end{align}
\end{proposition}

\begin{theorem}
The derivative w.r.t. $t$ in the thermodynamical limit is 
\begin{align}
\label{dt_thermo_RSB}
\frac{d\mathcal{A}^{(P)}(t)}{dt} &= \bar{m}^P (1-P) -{\beta^{'} \gamma}(\theta-1) \dfrac{P}{2}\q_2^{P/2}\p_2 +{\beta^{'} \gamma}\theta \dfrac{P}{2}\q_1^{P/2}\p_1 .
\end{align}

Thus, in the thermodynamic limit and under the assumption of first step of replica symmetry breaking, we reach the same results we computed via transport equation's interpolation (see equation \ref{A_1RSB_finalissima_guerra}), namely the quenched statistical pressure for $P\geq 4$  of the DHN becomes
\begin{align}
    &\mathcal{A}^{(P)}=\ln 2 + \frac{1}{\theta}\mathbb{E}_1 \ln \mathbb{E}_2 \cosh^\theta g (\bm J, \bar{m}) -\dfrac{\gamma\b}{4}\q_2^{P/2-1}\p_2\Big(P-(P-1)\q_2\Big) \notag \\
    &+ \frac{\beta^{'}}{2}\bar{m}^P (1-P) -\theta(P-1)\dfrac{\b\gamma}{4}(\q_2^{P/2}\p_2-\q_1^{P/2}\p_1)+\dfrac{\gamma\beta'\,^2}{4}
    \label{A_1RSB_finalissima_guerra}
\end{align}
where 
\begin{align}
    g(\bm J, \bar{m})&= \frac{\beta^{'}P}{2}\bar{m}^{P-1} + J^{(1)}\sqrt{\frac{\beta^{'}}{2}\gamma \bar{p}_1P \bar{q}_1^{P/2-1}} +J^{(2)}\sqrt{ \frac{\beta^{'}}{2} P\gamma \left[ \bar{p}_2 \bar{q}_2^{P/2-1} -\bar{p}_1 \bar{q}_1^{P/2-1}  \right]}.
\end{align}
\end{theorem}

The proof is similar via transport equation's interpolation (see Appendix \ref{SusyBreak}), since we omit it.

\begin{remark}
Note that the above expression sharply coincides with \eqref{eq:correct_1RSB_Trasp}, hence from now on the results obtained trough the first approach automatically translate also in this setting and it is pointless to repeat the calculations: the scenario painted trough the transport-PDE approach is meticulously confirmed.
\end{remark}

\section{Ground state analysis of the maximal storage}\label{Sezione4} 
Once set the net in the Baldi-Venkatesh regime of operation (the maximal storage scaling allowed to the network, i.e. $K = \gamma N^{P-1}$), in this section we perform fine tuning, namely we search the numerical value $\gamma_c$ that sets the maximal achievable storage: this is done in the $\beta \to \infty$ limit of zero temperature of course (where no fast noise is present) and we inspect as $\gamma$ grows, the behavior of the Mattis  magnetization: as long as that observable is $\sim 1$ the network is in the retrieval operation mode -i.e., it is performing pattern recognition and associative memory- but when the magnetization suddenly drops to zero, this defines  the critical capacity $\gamma_c$: beyond that value, it is pointless to add more patterns to the network because its associative properties are lost and it behaves as a pure spin-glass with no retrieval skills (the network has a phase transitions: it escapes the retrieval region and enters the pure spin glass region). 
\newline 
Before starting the calculations we just point out that, as the  proofs of the next two theorems are short but somehow cumbersome, we prefer to keep them in the main text.

\subsection{RS approximation}
\label{sc:TzeroRS}
As standard also for the classic Hopfield model \cite{Amit}, to get the ground state solution (namely the self-consistencies for $\b\to \infty$) for the case of $P>2$, we now assume that $\lim_{\b\to\infty}\b (1-\q)$ is finite. This gives rise to the following.

\begin{theorem}
Assuming that $\lim_{\b\to\infty}\b (1-\q)$ is finite, the zero-temperature self-consistency equation for the Mattis magnetization reads as
\begin{equation}
    \begin{array}{lll}
         \m&\coloneqq&\mathrm{erf}{\left[\dfrac{1}{2}\sqrt{ \dfrac{P}{\gamma}}\,\m^{P-1}\right]}\,.
    \end{array}
    \label{eq:self_GuerraRS_Tzero}
\end{equation}
where $\mathrm{erf}$ is the error function.
\end{theorem}

\begin{proof}
We adapt the computation from \cite{Amit}. As a first step we introduce an additional term $\b y$ in the argument of the hyperbolic tangent appearing in the self-consistency equations \eqref{eq:self_GuerraRS}:
\begin{equation}
    \begin{array}{lll}
         \m=\left\langle\tanh{\left[\b\left(\dfrac{P}{2}  \m^{P-1}+x\sqrt{ \gamma \frac{P}{2} \q^{^{P-1}}}\,+y\right)\right]}\right\rangle_x \,,
         \\\\
         \q=\left\langle\tanh{}^{\2}{\left[\b\left(\dfrac{P}{2}  \m^{P-1}+x\sqrt{ \gamma \frac{P}{2} \q^{^{P-1}}}\,+y\right)\right]}\right\rangle_x \,.
    \end{array}
\end{equation}
We also recognize that as $\b\to\infty$ we have $\q\to 1$, therefore in order to perform the limit we will introduce the reparametrization
\begin{eqnarray}
\q=1-\dfrac{\delta \q}{\b}\;\;\;\mathrm{as}\;\;\;\b\to\infty\,.
\end{eqnarray}
In this way we obtain
\begin{equation}
    \begin{array}{lll}
         \m&=\left\langle\tanh{\left[\b\left(\dfrac{P}{2}  \m^{P-1}+x\sqrt{ \gamma \frac{P}{2} \left(1-\dfrac{\delta \q}{\b}\right)^{^{P-1}}}\,+y\right)\right]}\right\rangle_x \,,
         \\\\
         1-\dfrac{\delta \q}{\b}&=\left\langle\tanh{}^{\2}{\left[\b\left(\dfrac{P}{2}  \m^{P-1}+x\sqrt{ \gamma \frac{P}{2} \left(1-\dfrac{\delta \q}{\b}\right)^{^{P-1}}}\,+y\right)\right]}\right\rangle_x \,.
    \end{array}
\end{equation}
Using the new parameter $y$ we can recast the last equation in $\delta \q$ as a derivative of the magnetization
\begin{equation}
    \dfrac{\partial \m}{\partial y}=\b\left[1-\left(1-\dfrac{\delta \q}{\b}\right)\right]=\delta\q
\end{equation}
Thanks to this correspondence between $\m$ and $\q$,  we can proceed without worrying about $\q$
\begin{equation}
    \begin{array}{lll}
         \m&=\left\langle\mathrm{sign}{\left[\dfrac{P}{2}  \m^{P-1}+x\sqrt{ \gamma \frac{P}{2} }\,+y\right]}\right\rangle_x \,,
         \\\\
         \delta \q&=\dfrac{\partial \m}{\partial y} \,.
    \end{array}
\end{equation}
These equations can be simplified by evaluating the Gaussian integral in $x$, via the relation:
\begin{equation}
    \l\mathrm{sign}[A x +B]\r_x=\mathrm{erf}\left(\dfrac{B}{\sqrt{2}\,A}\right)\,,
\end{equation}
to get
\begin{equation}
    \begin{array}{lll}
         \m&=\mathrm{erf}{\left[\dfrac{\frac{P}{2}  \m^{P-1}+y}{\sqrt{ \gamma P }}\right]}\,,
         \\\\
         \delta \q&=\dfrac{2}{\sqrt{ \gamma\pi P }}\exp{\left\lbrace-\left[\dfrac{\frac{P}{2}  \m^{P-1}+y}{\sqrt{ \gamma P }}\right]^2\right\rbrace }\,.
    \end{array}
\end{equation}
Setting $y=0$ we close the proof.
\end{proof}

\begin{corollary}
As conjectured by Gardner via the replica trick \cite{gardner}, in the limit $P\to\infty$, $\gamma_{c}$ is a divergent function of P of the form
\begin{equation}
\gamma_{_C}\sim\dfrac{P}{\log{P}}\,.
\label{eq:gammaC_RS_Tzero}
\end{equation}
\end{corollary}
\begin{proof}
As numerically for $P\to\infty$ we have found that the value of the magnetization for $\gamma\leq\gamma_{_C}$ is always $\m=1$ and decays to $0$ for $\gamma>\gamma_{_C}$, to find the trend of $\gamma_{_C}$ as a function of $P$, from the \eqref{eq:self_GuerraRS_Tzero}, we have to impose the following condition
\begin{equation}
    \left|\mathrm{erf}\left[\dfrac{1}{2}\sqrt{\dfrac{P}{\gamma}\:}\right]-1\right|<\epsilon
\end{equation}
solving this equation for $P/\gamma_{_C}$ within the limit of small values of $\epsilon$, we have the approximate solution
\begin{equation}
    \dfrac{P}{\gamma_{_C}}=4\log{\left[\dfrac{1}{\epsilon}\right]}-2\log{\left[\dfrac{\pi}{2}\log{\left[\dfrac{2}{\pi \epsilon^2}\right]}\right]}+\mathcal{O}(\epsilon^2)
\end{equation}
thus, as $\epsilon\to 0$ the ratio $P/\gamma_{_C}$ must be a divergent function of the form 
\begin{equation}
\dfrac{P}{\gamma_C}\sim 4\log{\left[\dfrac{1}{\epsilon}\right]}
\end{equation}
choosing $\epsilon = 1/P$, this condition allows $\gamma_{_C}$ to be a divergent function of $P$ of the form in \eqref{eq:gammaC_RS_Tzero}.
\end{proof}


\subsection{1-RSB approximation}
\begin{theorem}
The zero-temperature self-consistency equations for the Mattis magnetization (and, technically required, also for $\Delta \bar{q}=\bar{q}_2-\bar{q}_1$), in the 1-RSB scenario, read as
\begin{equation}
    \begin{array}{lll}
         \m=1-2\,\mathbb{E}_1\left\lbrace\dfrac{1}{1+e^{2\Theta\,(A_1+A_2 J^{(1)})}\mathcal{F}(\boldsymbol{A})}\right\rbrace\,,
         \\\\
         \Delta\q= \q_2-\q_1=4\,\mathbb{E}_1\left\lbrace\dfrac{e^{2\Theta\,(A_1+A_2 J^{(1)})}\mathcal{F}(\boldsymbol{A})}{\Big[1+e^{2\Theta\,(A_1+A_2J^{(1)})}\mathcal{F}(\boldsymbol{A})\Big]^2}\right\rbrace\,,
    \end{array}
\end{equation}
where
\begin{equation}
\begin{array}{lll}
     \mathcal{F}(\boldsymbol{A})=\dfrac{1+\mathrm{erf}\Big[\mathcal{K}^+\Big]}{1+\mathrm{erf}\Big[\mathcal{K}^-\Big]} &\mathrm{with}& \mathcal{K^{\pm}}=\frac{\Theta\,A_3^2\pm(A_1+A_2 J^{(1)})}{A_3\sqrt{2}}
\end{array}
\end{equation}
and
\begin{equation}
    \begin{array}{lll}
         A_1=\dfrac{P}{2}\m^{P-1}\,,
         &
         A_2=\sqrt{\dfrac{\gamma P }{2}}\,,
         &
         A_3=\sqrt{\dfrac{\gamma P(P-1)}{2}\Delta\q}\,.
    \end{array}
\end{equation}
\end{theorem}

\begin{proof}
Following the same steps presented in the RS assumption, we introduce the additional term $\b y$ in the expression of $g(\boldsymbol{J},\m)$, the self consistent equations in Corollary \ref{cor:SC_HOP_1RSB} read as
\begin{align}
&\bar{m}= \mathbb{E}_1 \left[ \frac{\mathbb{E}_2 \cosh^\theta g(\bm J,\bar{m})\tanh g(\bm J,\bar{m})}{\mathbb{E}_2 \cosh^\theta g(\bm J,\bar{m})}\right] \\
&\bar{q}_1 = \mathbb{E}_1 \left[ \frac{\mathbb{E}_2 \cosh^\theta g(\bm J,\bar{m})\tanh g(\bm J,\bar{m})}{\mathbb{E}_2 \cosh^\theta g(\bm J,\bar{m})}\right]^2 \\
&\bar{q}_2 = \mathbb{E}_1 \left[ \frac{\mathbb{E}_2 \cosh^\theta g(\bm J,\bar{m})\tanh^2 g(\bm J,\bar{m})}{\mathbb{E}_2 \cosh^\theta g(\bm J,\bar{m})}\right]
\end{align}
where 
\begin{equation}
\small
\begin{array}{lll}
     g(\bm J,\bar{m})= \beta^{'} \left( \dfrac{P}{2}\bar{m}^{P-1} + \sqrt{\dfrac{\gamma P \q_1^{P-1}}{2}} J^{(1)} +\sqrt{\dfrac{\gamma P(\q_2^{P-1}-\q_1^{P-1})}{2}} J^{(2)} +y\right). 
\end{array}
\label{eq:g_zeroT}
\end{equation}
We recognize that as $\b\to\infty$, we have $\q_2\to 1$, therefore in order to perform the limit we will introduce the reparametrization
\begin{equation}
\begin{array}{lll}
     \q_2=1-\dfrac{\delta\q_2}{\b}&\mathrm{as}&\b\to\infty
\end{array}
\end{equation}
Using the new parameter $y$, we can recast the equation for $\q_2$ as a derivative of the magnetization
\begin{equation}
\begin{array}{lll}
     \dfrac{\partial\m}{\partial y}=\delta\q_2-\Theta\Delta\q\Longrightarrow\delta\q_2=\dfrac{\partial\m}{\partial y}+\Theta\Delta\q
\end{array}
\end{equation}
where we have used $\Delta\q=\q_2-\q_1$ and, as $\b\to\infty$,  $\b\theta\to \Theta\in\mathbb{R}$. Thus, in the zero temperature limit the previous equations become
\begin{equation}
\begin{array}{lll}
     \bar{m}&\to& \mathbb{E}_1 \left\lbrace \dfrac{\mathbb{E}_2 \left[\mathrm{sign}{[g(\boldsymbol{J},\m)]}\:e^{\Theta|g(\boldsymbol{J},\m)|}\,\right]}{\mathbb{E}_2 \left[e^{\Theta|g(\boldsymbol{J},\m)|}\right]}\right\rbrace \\
\Delta\q &\to& 1-\mathbb{E}_1 \left\lbrace \dfrac{\mathbb{E}_2 \left[\mathrm{sign}{[g(\boldsymbol{J},\m)]}\:e^{\Theta|g(\boldsymbol{J},\m)|}\,\right]}{\mathbb{E}_2 \left[e^{\Theta|g(\boldsymbol{J},\m)|}\right]}\right\rbrace^2 \\
\q_2 &\to& 1 
\end{array}
\end{equation}
Now, if we suppose $\Delta \q\ll 1$ the \eqref{eq:g_zeroT} reduces to
\begin{equation}
\label{gmodified}
    g(\boldsymbol{J},\m)=\b\left[A_1+A_2 J^{(1)}+A_3 J^{(2)}+\mathcal{O}(\Delta\q)\right]
\end{equation}
where
\begin{equation}
    \begin{array}{lll}
         A_1=\dfrac{P}{2}\m^{P-1}\,,
         &
         A_2=\sqrt{\dfrac{\gamma P }{2}}\,,
         &
         A_3=\sqrt{\dfrac{\gamma P(P-1)}{2}\Delta\q}\,.
    \end{array}
    \label{eq:A_parameters_RSB}
\end{equation}
Performing the integral over $J^{(2)}$ we get the proof.
\end{proof}

\begin{remark}
Note that, as $\Delta\q\to 0$, the whole above construction collapses to the replica symmetric picture as it should.
\end{remark}
\begin{proof}
For $\Delta\q\to 0$ from \eqref{gmodified} and \eqref{eq:A_parameters_RSB}, we have
\begin{equation}
\begin{array}{lll}
     g(\bm J,\bar{m})\to\b\left[  \dfrac{P}{2}\m^{P-1} + \sqrt{\dfrac{\gamma P }{2}} J^{(1)} \right]
\end{array}
\end{equation}
and so
\begin{equation}
\begin{array}{lll}
     \bar{m}&\to& \mathbb{E}_1 \left[  \mathrm{sign}\left(  \dfrac{P}{2}\m^{P-1} + \sqrt{\dfrac{\gamma P }{2}} J^{(1)} \right)\right]
     \\\\
\Delta\q &\to& 0 
\\
\q_1 &\to& \q_2\;\to\;1 
\end{array}
\end{equation}
which are the equations in the zero-temperature limit of RS assumption.
\end{proof}

\begin{remark}
We checked numerically the behavior of critical capacity, both  in the RS and 1-RSB assumptions, and  -as reported in the plots of Figure \ref{fig:gammaC}, we can appreciate that their trends are similar, almost identica:  also in the 1-RSB  scenario $\gamma_C$ is a divergent function of $P$ of the form $\frac{P}{\log P}$ and, as expected, in these regards replica symmetry breaking plays a minor role.
\begin{figure}[h!]
         \centering
         \includegraphics[width=0.4\textwidth]{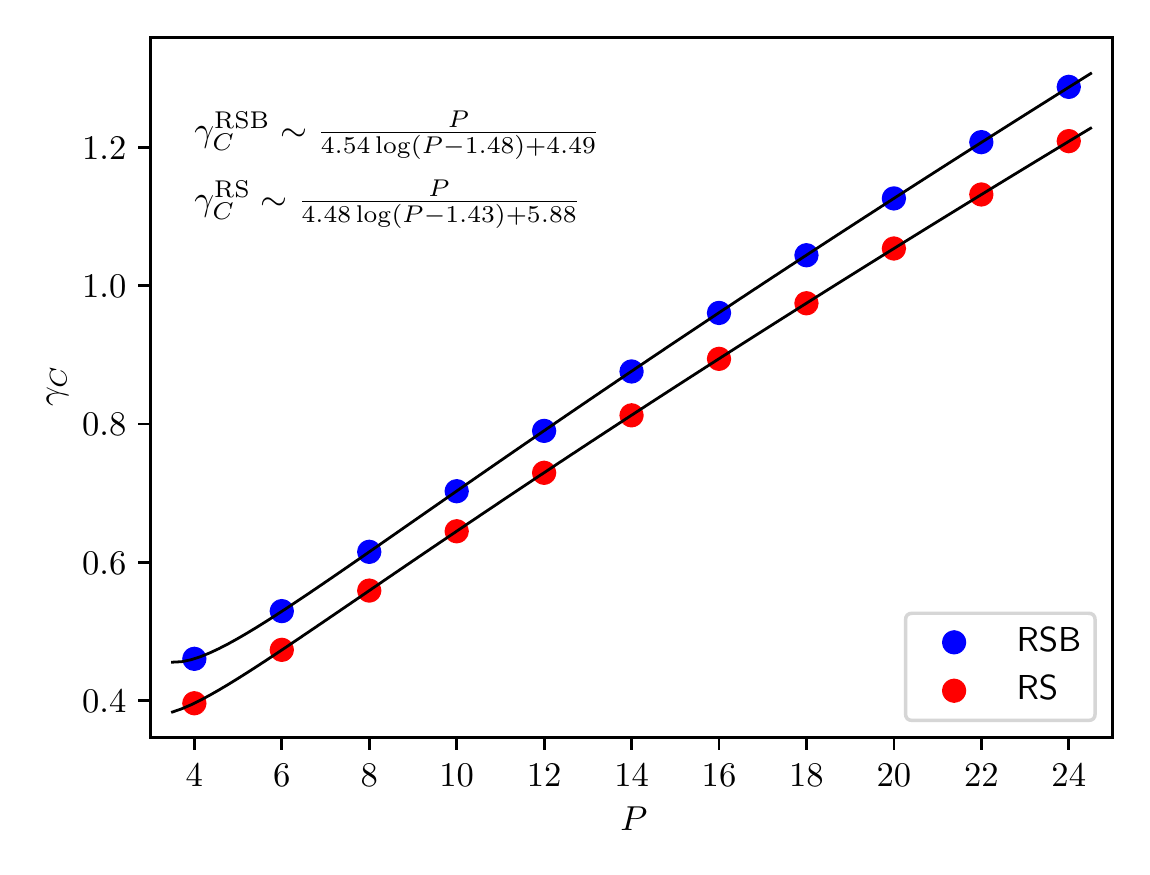}
           \includegraphics[width=0.4\textwidth]{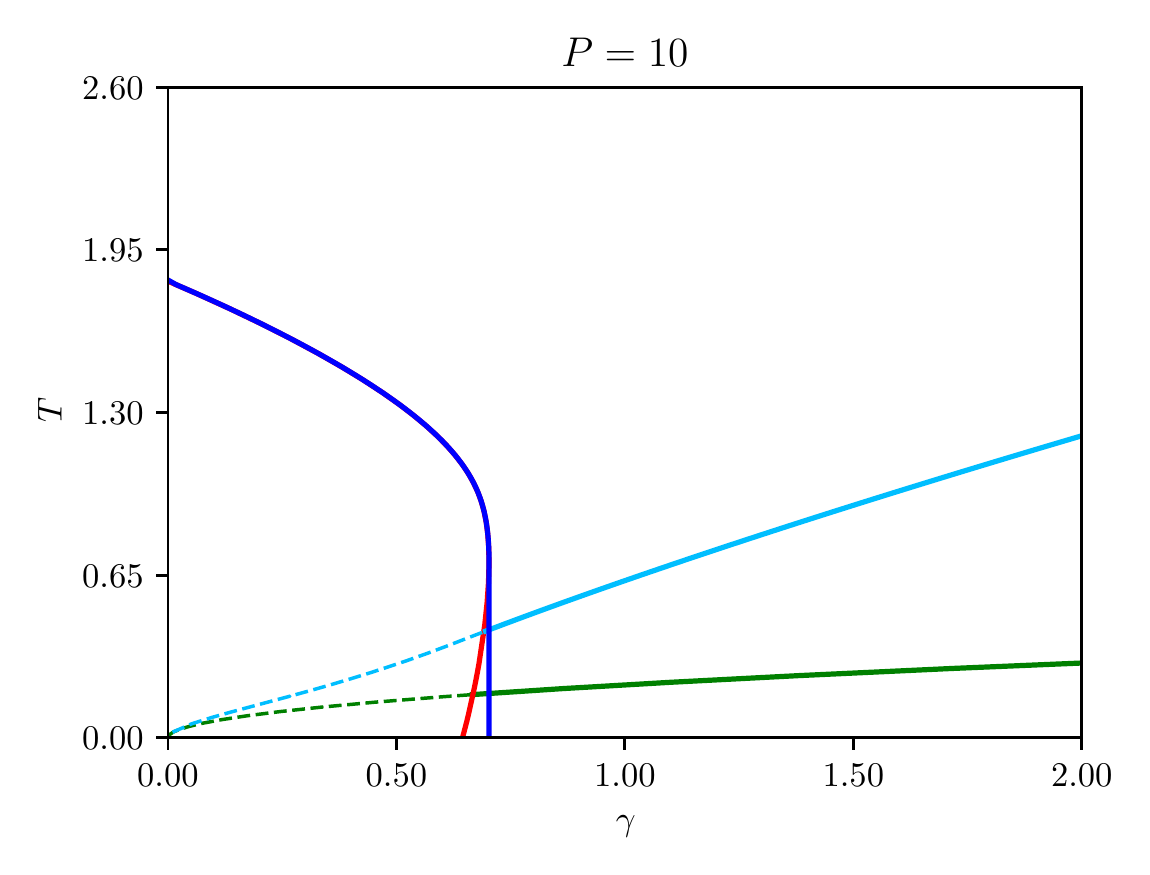}
        \caption{Left: $\gamma_c$ as a function of $P$; we note that  -for all values of $P$- the RSB maximal capacity is systematically  larger than its replica symmetric counterpart. Right: Superposition of the RS and 1-RSB phase diagrams for a given $P$ -i.e. $P=10$, the same of the Monte Carlo runs reported in Figure \ref{fig:trend_mag_RS_2} - to facilitate visual comparison of the various regions: we note that the spin-glass phase is systematically larger in the RSB scenario (light blue) rather than in the RS counterpart (green). Within the retrieval region the spin glass solution is always unstable, both in the RS and in the 1-RSB approximations. }
        \label{fig:gammaC}
\end{figure}
\end{remark}

\section{The structure of the glassiness}\label{Sezione5}
\label{sec:glassy}

\begin{figure}[h!]
     \begin{subfigure}[b]{0.5\textwidth}
         \centering
         \includegraphics[width=\textwidth]{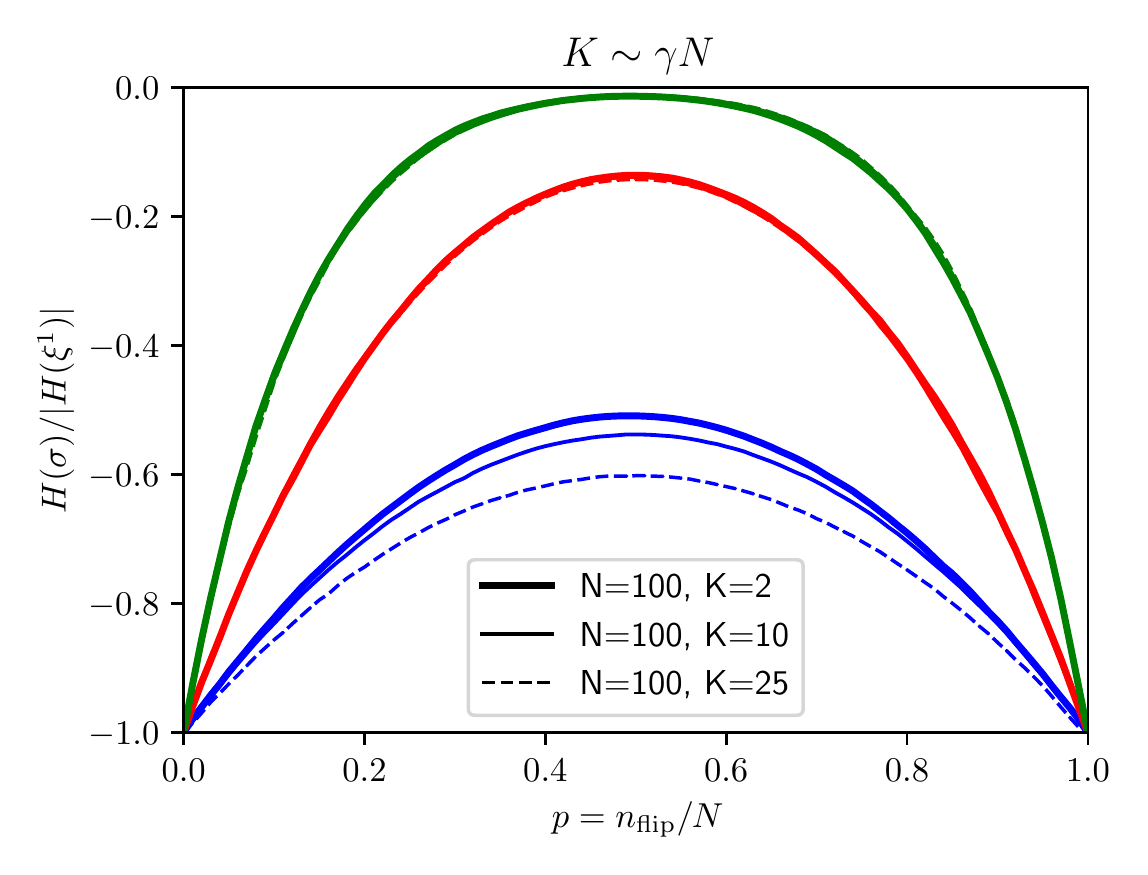}
     \end{subfigure}
     \begin{subfigure}[b]{0.5\textwidth}
         \centering
         \includegraphics[width=\textwidth]{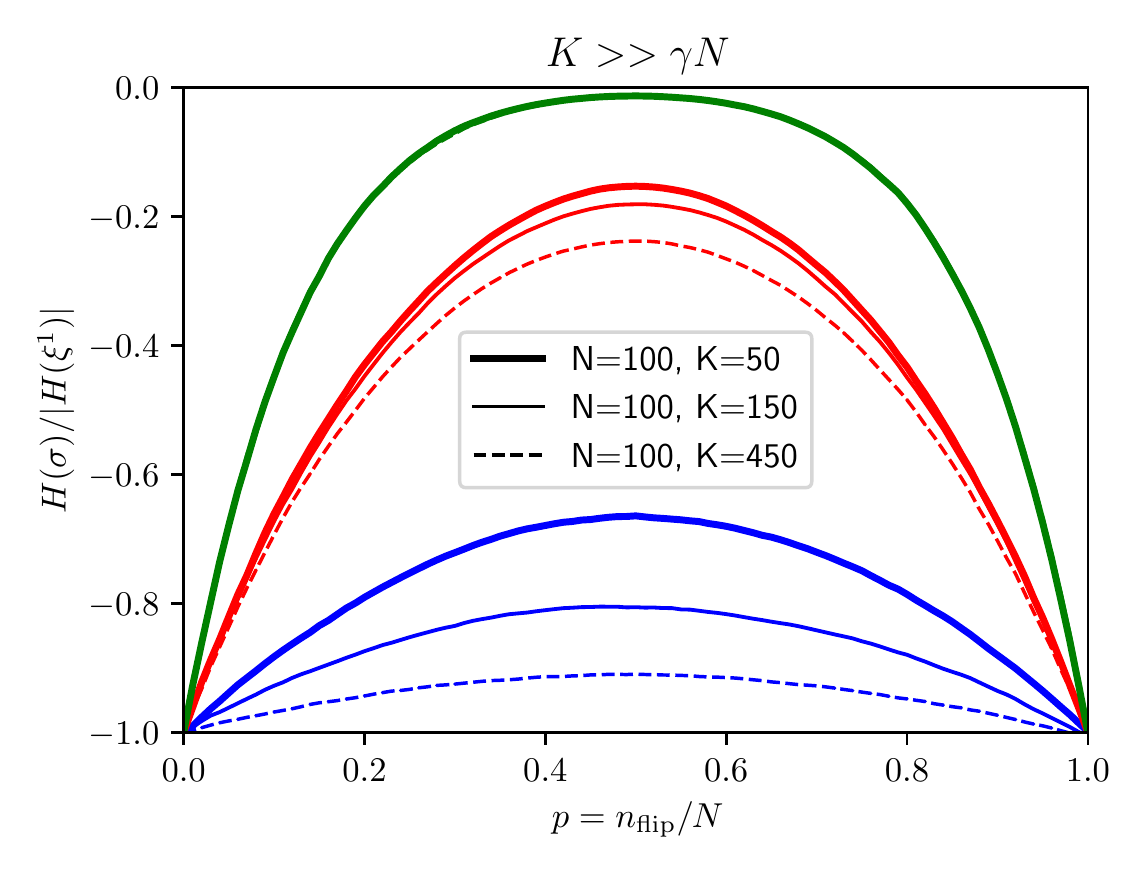}
     \end{subfigure}
        \caption{Comparison of the structure of the landscape in the {\em high resolution} regime \cite{PRLNN} (left) and in the {\em high storage} regime (right). In the vertical axes we plot the ratio where in the denominator there is the Hamitonian evaluated in the minimum corresponding to the pattern  $\xi^1$ -and it is fixed-  and in the numerator we plot the value of the Hamiltonian where we perform ground state spin flips to step away from $\xi^1$ toward $\xi^2$.  In the horizontal  axes we plot the number of spin flips required to move from $\xi^1$ to $\xi^2$. 
\newline
In blu we report $P=2$ (standard Hopfield), in red $P=4$ and in green $P=8$. It shines that in dense networks minima are more profound w.r.t. the shallow limit and energy barrier are higher (hence trapping in spurious states become less probable for dense networks). Note that Hopfield has a parabolic shape as expected being a quadratic Hamiltonian (see also \cite{Agliari-Barattolo,leonelli}).   
\newline
Selected a network (i.e. selected a color in the plot), as the storage grows we see that the maxima of these curves -that happen on the mixture of $\xi^1$ and $\xi^2$- the contribution of the quenched noise increases and the corresponding energy of the maximum gets lower. Further, whatever the storage, we highlight that the basin of attractions of the minima gets steeper as $P$ grows, suggesting both a higher critical storage value as well as their flat structure  (see also \cite{Zecchina-New}).}
        \label{fig:trend_mag_RS_1}
\end{figure}

In order to deepen the glassy structure of these neural networks it is instructive to start with a glance at the pairwise reference. Remembering that the Hamiltonian of the Sherrington-Kirkpatrick (SK) spin glass reads as
$$
H_{SK}=\frac{-1}{\sqrt{N}}\sum_{i<j}^{N,N}J_{ij}\sigma_i \sigma_j,
$$
with $J_{ij}$ quenched random couplings i.i.d. accordingly to $\mathcal{N}[0,1]$, if we consider the standard Hopfield limit (i.e. we set $P=2$ in the dense Hebbian network), we can write the related Hamiltonian and partition function as
\begin{eqnarray}
H_{Hopfield}(\bold{\sigma}|\bold{\xi})=\frac{-1}{N}\sum_{i<j}^{N,N}\sum_{\mu=1}^K \xi_i^{\mu}\xi_j^{\mu}\sigma_i\sigma_j,\\
Z_{Hopfield}=\sum_{\sigma}^{2^N}\exp\left(-\beta H_{Hopfield}(\bold{\sigma}|\bold{\xi})\right).
\end{eqnarray}
In turn, these can be rewritten, after minimal manipulations -i.e., for the former splitting the signal (i.e. the pattern to be retrieved, say $\mu=1$) from the quenched noise (i.e. all the other patterns) and for the latter using its integral representation \'a la Hubbard-Stratonovich, as
\begin{eqnarray}\label{ducojoni1}
H_{Hopfield}(\bold{\sigma}|\bold{\xi})=-m_1^2/2 + \frac{-1}{\sqrt{N}}\sum_{i<j}^{N,N}J_{ij}\sigma_i\sigma_j, \textit{with} \ \  J_{ij}= (\frac{1}{\sqrt{N}}\sum_{\mu=1}^K \xi_i^{\mu}\xi_j^{\mu}) \\ \label{ducojoni2}
Z_{Hopfield}=\sum_{\sigma}^{2^N}e^{\beta m_1^2}\int_{-\infty}^{+\infty} \prod_{\mu=2}^P dz_{\mu}e^{-z^2/2}\exp\left(\frac{1}{\sqrt{N}}\sum_{i,\mu}^{N,P}\xi_i^{\mu}\sigma_i z_{\mu}\right),
\end{eqnarray}
hence, it shines  that, if naively we send $N \to \infty$ in eq. (\ref{ducojoni1}) we note that $J_{ij} \to \mathcal{N}[0,1]$ -as in the Sherrington-Kirkpatrick model- and, correspondingly,   the normalization of the Hopfield Hamiltonian turns  to the Sherrington-Kirkpatrick one (i.e. $\sqrt{N}$ rather than $N$): certainly we are dealing with a spin-glass, we must now study what kind of spin glass it is. A glance at eq. (\ref{ducojoni2}) suggests a bipartite spin-glass made of by one party with $N$ Ising spins (binary neurons) $\sigma_i = \pm 1$ and one party with $K$ Gaussian spins (real valued neurons equipped with a Gaussian prior). 
Indeed, in a couple of recent papers \cite{bipartito-mio,glassy}, Guerra and coworkers provided -at the replica symmetric level of description only- a representation theorem for the standard Hopfield quenched statistical pressure in terms of the related quenched statistical pressures of an hard spin glass (i.e. the Sherrington-Kirkpatrick model) and a soft one (i.e. the Gaussian or {\em spherical} model): as the former is full-RSB (it is the archetype of models where Parisi theory is correct) \cite{BGDiBiasio,Guerra,Talagrand}, while the latter is replica symmetric \cite{soffice,Dembo}, the interplay among them confers a glassiness to the Hopfield model that is typical of that kind of neural network and it is not the same nor of the hard spin glass alone  neither of the soft one alone. 
\newline
Does the glassiness of the Hopfield neural network hold also for dense networks?  
\newline
A glance at the self-consistencies for the overlap both at the replica symmetric level -see equation (\ref{eq:self_GuerraRS})- as well as under the first step of RSB -see equations (\ref{Senza1RSBa}-\ref{Senza1RSBb})- seems to suggest that this is no longer the case as the self-consistencies for the overlap are the same of the standard hard P spin glass (namely the Sherrington-Kirkpatrick model with P-wise interactions \cite{conBurioni,GuysAlone}) both in the RS and in the 1-RSB scenarios.
\newline
To prove this conjecture, in this section we generalize the Guerra's representation theorem in various directions: at first we focus on the standard pairwise Hopfield model to inspect if such a decomposition holds also within a broken replica framework and we prove that it keeps holding. Then we focus on dense networks and   we prove that such a decomposition theorem does not hold, rather these networks have quenched statistical pressures related solely to those pertaining to the hard spin glasses. The soft part disappears and this turns to be true both at the replica symmetric and within the first step of replica symmetry breaking: Let us prove these statements and deepen their consequences .

\subsection{RS scenario}
\subsubsection{Case $P=2$ (standard Hopfield reference)}

For sake of completeness, in this subsection we report the decomposition theorem for $P=2$ case, namely standard Hopfield model, claimed in \cite{glassy}. 

\begin{theorem}
Fixed at noise level $\beta,\ \beta_1$ and $\beta_2$ as
\begin{align}
\beta_1 &= \frac{\sqrt{\gamma \beta}}{1-\beta(1-\q)} \\
\beta_2 &= 1-\beta(1-\q),
\end{align}
the replica
symmetric approximation of the quenched free energy of the analogical neural network can be linearly decomposed in terms of the replica symmetric approximation of the Sherrington–Kirkpatrick quenched free energy, at noise level $\beta_1$, and the replica symmetric approximation of the quenched free energy of the Gaussian spin glass, at noise level $\beta_2$, such that 
\begin{align}
\mathcal{A}_{NN}^{RS} (\beta,\gamma) &= \mathcal{A}_{SK}^{RS} (\beta, \beta_1) + \gamma \mathcal{A}_{Gauss} (\beta_2, \beta)-\frac{1}{4} \beta^2_1.
\end{align}
\end{theorem}

\subsubsection{Case $P>2$ (dense Hebbian network)}

In this subsection we show that, as long as $P > 2$, the above representation does not hold any longer and the decomposition reduces to a simpler version (where solely the hard spin glass is involved).  This is captured by the next
\begin{theorem}
Let us fix the noise levels $\beta_1$ and $\beta_2$ as follow
\begin{equation}
    \begin{array}{lll}
         \beta_1&=& \dfrac{\b\sqrt{\gamma}}{1-\b N^{1-P/2}(1-\q^{P/2})}\,,
         \\\\
         \beta_2&=&1-\b N^{1-P/2}(1-\q^{P/2})\,,
    \end{array}
    \label{eq:b1_b2}
\end{equation}
and recall the finite size expressions for the quenched statistical pressures of the Hopfield model $ \mathcal{A}_{NN}^{(P)}$, the hard P-spin glass $\mathcal{A}_{SK}^{(P)}$ and the soft P-spin glass $\mathcal{A}_{Gauss}^{(P)}$ obtained with Guerra's interpolation technique, that read as\footnote{While extensive statistical mechanical treatments of both the hard and soft P-spin glass are extensively available in the Literature \cite{conBurioni,Barrat,Crisanti,Guerra,Dmitry}, in \cite{GuysAlone} we re-obtained sharply the expressions (\ref{eq:SK_finite_size_RS}) and (\ref{eq:Gauss_finite_size_RS}) via the two techniques developed in this paper.}
\begin{equation}
\footnotesize
    \begin{array}{lll}
       &  \mathcal{A}_{NN}^{(P)}(\b,\gamma)=\ln{2}-\dfrac{\b\gamma}{2}N^{P/2-1}+\dfrac{\gamma N^{P/2-1}}{2}\dfrac{\b\q^{P/2}}{1-\b N^{1-P/2}(1-\q^{P/2})}-\dfrac{\gamma N^{P-2}}{2}\ln{\left[1-\b N^{1-P/2}(1-\q^{P/2})\right]}
         \\\\
         &+\left\langle\ln{}\cosh{\left[\dfrac{P}{2}\b\m^{P-1}+Y\sqrt{\b\gamma\dfrac{P}{2}N^{P/2-1}\p\q^{P/2-1}}\right]}\right\rangle_Y-\dfrac{P-1}{2}\b\m^{P}-\b\gamma\dfrac{P}{4}\p N^{P/2-1}\q^{P/2-1}(1-\q)+V_N^{(NN)}\,,
         \label{eq:NN_finite_size_RS}
    \end{array}
\end{equation}
\begin{equation}
\footnotesize
    \begin{array}{lll}
         \mathcal{A}_{SK}^{(P)}(\b,\beta_1)= & \ln{2}+\left\langle\ln{\cosh{\left[\dfrac{P}{2}\b\m^{P-1}+Y\sqrt{\dfrac{P}{2}\beta_1^2 \q_{SK}^{^{P-1}}} \right]}}\right\rangle_Y-\dfrac{P-1}{2}\b \m^P
         \\\\
         &+\dfrac{1}{4}\beta_1^2\Big(1 -P\q_{SK}^{P-1}+(P-1)\q_{SK}^P\Big)+V_N^{(SK)}\,,
    \end{array}
    \label{eq:SK_finite_size_RS}
\end{equation}
\begin{equation}
\footnotesize
    \begin{array}{lll}
         \mathcal{A}_{Gauss}^{(P)}(\lambda,\beta_2)=&\dfrac{1}{2}\dfrac{\beta_2^2\frac{P}{2}\q_{G}^{P-1}}{1-\lambda+\beta_2^2\frac{P}{2}\q_{G}^{P-1}}-\dfrac{1}{2}\ln{\left[1-\lambda+\beta_2^2\frac{P}{2}\q_{G}^{P-1}\right]} +(P-1)\dfrac{\beta_2^2}{4}\q_{G}^P+V_N^{(Gauss)}
         \end{array}
    \label{eq:Gauss_finite_size_RS}
\end{equation}
where we used
\begin{equation}
\footnotesize
\begin{array}{lll}
V_N^{(NN)}=\dfrac{\b}{2}\SOMMA{k=2}{P}\begin{pmatrix}
P\\ k
\end{pmatrix}\l(\Delta m)^k\r\m^{P-k}-\dfrac{\b \gamma N^{P/2-1}}{2}\left[\SOMMA{k=1}{P/2}\begin{pmatrix}
\frac{P}{2}\\ k
\end{pmatrix}\l\Delta p(\Delta q)^k\r \q^{P/2-k}+\SOMMA{k=2}{P/2}\begin{pmatrix}
\frac{P}{2}\\ k
\end{pmatrix}\l(\Delta q)^k\r\p\q^{P/2-k}\right]\,,
\\\\
V_N^{(SK)}=\dfrac{\b}{2}\SOMMA{k=2}{P}\begin{pmatrix}
P\\ k
\end{pmatrix}\l(\Delta m)^k\r\m^{P-k}-\dfrac{\beta_1^2}{4}\SOMMA{k=2}{P}\begin{pmatrix}
P\\ k
\end{pmatrix}\l(\Delta q_{SK})^k\r\q_{SK}^{P-k}\,,
\\\\
V_N^{(Gauss)}=-\dfrac{\beta_2^2}{4}\SOMMA{k=2}{P}\begin{pmatrix}
P\\ k
\end{pmatrix}\l(\Delta q_G)^k\r\q_G^{P-k}\,.
\end{array}
\label{eq:potential_glassy_RS}
\end{equation}
We can write the following decomposition of the finite size quenched statistical pressure of the dense Hebbian network in terms of the replica symmetric quenched pressures of the Sherrington-Kirkpatrick P-spin glass, at noise level $\beta_1$, and the replica symmetric quenched statistical pressure of the Gaussian P-spin glass, at noise level $\beta_2$:
\begin{equation}
\footnotesize
\begin{array}{lll}
      \mathcal{A}_{NN}^{(P)}(\b,\gamma)\coloneqq&\mathcal{A}_{SK}^{(P)}(\b,\beta_1)-\dfrac{\b\gamma}{2}N^{^{P/2-1}}-\dfrac{1}{4}\beta_1^2+\gamma N^{P-2}\mathcal{A}_{Gauss}^{(P)}(\b,\beta_2)
    \\\\
    &+\gamma N^{P-2}\dfrac{2-P}{4P}\left(\dfrac{\beta_2-(1-N^{^{1-P/2}}\b)}{\beta_2}\right)^2-\Big(V_N^{(SK)}+\gamma N^{P-2}V_N^{(Gauss)}-V_N^{(NN)}\Big)\,,
\end{array}
    \label{eq:SK_Gauss}
\end{equation}
\end{theorem}
\begin{proof}
The proof for $P=2$ is presented in \cite{glassy}. The generalization to $P>2$ is obtained following the same steps but taking care of using the new definitions of the noise in \eqref{eq:b1_b2}.
\end{proof}

\begin{remark}
\label{rm:thermod_limit_potential_RS}
Note that, in the thermodynamic limit, in the replica symmetric framework, $V^{(NN)}_N$, $V_{N}^{(SK)}$ and $ V_{N}^{(Gauss)}$ presented in \eqref{eq:potential_glassy_RS} vanish.
\end{remark}

\begin{corollary}
In the thermodynamic limit, for the case of $P>2$ the glassy nature of the  dense Hebbian network is  equivalent to that of a P-spin Sherrington-Kirkpatrick model with a noise level $\b\sqrt{\gamma}$:
\begin{equation}
    \mathcal{A}_{NN}^{(P)}(\b,\gamma)=\mathcal{A}_{SK}^{(P)}(\b,\b\sqrt{\gamma})
\end{equation}
\end{corollary}

\begin{proof}
As we set $P>2$, in the thermodynamic limit ($N\to\infty$), the definitions \eqref{eq:b1_b2} reads as
\begin{equation}
    \begin{array}{lll}
         \beta_1= \b\sqrt{\gamma}\,,&&\beta_2=1\,.
         \label{eq:noise_SK_P}
    \end{array}
\end{equation}
Using Remark \ref{rm:thermod_limit_potential_RS}, from the replica symmetric expression  of the quenched statistical pressure of P-spin  Sherrington-Kirkpatrick model presented in \eqref{eq:SK_finite_size_RS} with the new noise \eqref{eq:noise_SK_P}, the self consistent equation for $\q_{SK}$ in the SK P-spin model, 
in the thermodynamic limit, coincides with the one for $\q$ in the dense Hebbian network:
\begin{equation}
\footnotesize
    \q_{SK}=\left\langle\tanh{}^2{\left[\dfrac{P}{2}\b\m^{P-1}+Y\beta_1\sqrt{\dfrac{P}{2} \q_{SK}^{^{P-1}}} \right]}\right\rangle_Y=\q\,.
\end{equation}
Similarly, we can verify that, with the new noise \eqref{eq:noise_SK_P}, the self equation for $\q_G$ in  Spherical P-spin glass coincides with the one for $\p$ in the dense Hebbian network
\begin{equation}
\footnotesize
\begin{array}{lll}
     \q_G=\p=\b \q^{P/2}\,.
\end{array}
\end{equation}
where, we scaled $N^{P/2-1}\p$ as $\p$.  It can also be shown that in the thermodynamic limit for $P>2$ the Spherical P-spin glass model reduces to
\begin{equation}
\footnotesize
\begin{array}{lll}
     \gamma N^{P-2}\mathcal{A}_{Gauss}(\beta,\beta_2)-\dfrac{\gamma\b}{2}N^{P/2-1}-\dfrac{1}{4}\beta_1^2\xrightarrow[P>2]{\;\;N\to\infty\;\;}&\dfrac{1}{4}\gamma\beta'\,^2\q^{P}-\dfrac{1}{2P}\gamma\beta'\,^2\q^{P}\,.
\end{array}
\end{equation}
Moreover, the last term of \eqref{eq:SK_Gauss} becomes 
\begin{equation}
\footnotesize
\begin{array}{lll}
    \gamma N^{P-2}\dfrac{2-P}{4P}\left(\dfrac{\beta_2-(1-N^{^{1-P/2}}\b)}{\beta_2}\right)^2\xrightarrow[P>2]{\;\;N\to\infty\;\;}\dfrac{1}{2P}\gamma\beta'\,^2\q^{P}-\dfrac{1}{4}\gamma \beta'\,^2\q^{P}\,.
\end{array}
\end{equation}
So putting all together in relation \eqref{eq:SK_Gauss}, for $P>2$ in the thermodynamic limit, we have 
\begin{equation}
\footnotesize
    \mathcal{A}_{NN}^{(P)}(\b,\gamma)=\mathcal{A}_{SK}^{(P)}(\b,\b\sqrt{\gamma})\,.
\end{equation}
\end{proof}

\begin{remark}
This corollary is also verified by a direct calculation of the quenched statistical pressure in the replica symmetric scenario of the  dense Hebbian network in \eqref{eq:pressure_GuerraRS}, that perfectly coincides with the replica symmetric one of the Sherrington-Kirkpatrick P-spin glass presented in \eqref{eq:SK_finite_size_RS} if we set $\beta_1=\b\sqrt{\gamma}$.
\end{remark}

\subsection{1-RSB scenario}

Do the above representations  generalize to a broken replica picture? Yes, in the next subsections we prove that these theorems can be generalized to the 1-RSB scenario:   we show separately the $P=2$ (standard Hopfield in the broken replica regime) and even $P>2$ case (dense networks in the same broken replica regime).

\subsubsection{Case $P=2$ (standard Hopfield reference)}
We claim the following
\begin{theorem}
\label{th:P2_1RSB}
If we set
\begin{equation}
\small
    \begin{array}{lll}
        \bg= 1-(1-\q_2)\beta-\theta\beta(\q_2-\q_1)\,,&&\bsk=\dfrac{\beta\sqrt{\gamma}}{\beta_{G}^{(1)}}\,,
        \\\\
           \bgg=\bg\sqrt{\dfrac{1-\b(1-\q_2)}{\bg+\theta\b\q_1(1-\frac{\q_1}{\q_2})}} \,,&& \bskk=\dfrac{\beta\sqrt{\gamma}}{\beta_{G}^{(2)}}\,,
    \end{array}
    \label{eq:beta_interp_RSB}
\end{equation}
recalling the 1RSB expressions (see e.g. \cite{lindaRSB}) of Hopfield, SK and Spherical model, that read as
\begin{equation}
\footnotesize
\begin{array}{lll}
     \mathcal{A}_{NN}^{(1RSB)}(\beta,\gamma)=&\ln 2 + \frac{1}{\theta}\mathbb{E}_1 \left[\ln \mathbb{E}_2 \cosh^\theta g (\bm J, \bar{m})\right] -\beta\dfrac{\m^2}{2}+\dfrac{\gamma\beta}{2}\dfrac{\q_1}{1-(1-\q_2)\beta-\theta\beta(\q_2-\q_1)}-\dfrac{\gamma\beta}{2}\p_2(1-\q_2)
     \\\\
     &-\dfrac{\gamma\beta}{2}\theta(\p_2\q_2-\p_1\q_1)+\dfrac{\gamma}{2\theta}\ln{\left[1+\theta\dfrac{\beta(\q_2-\q_1)}{1-(1-\q_2)\beta-\theta\beta(\q_2-\q_1)}\right]}-\dfrac{\gamma}{2}\ln{\left[1-(1-\q_2)\beta\right]} \,,
\end{array}  
\label{eq:NN_1RSB_glassy}  
\end{equation}
\begin{equation}
\footnotesize
\begin{array}{lll}
     \mathcal{A}_{SK}^{(1RSB)}(\b,\beta_{SK}^{(1)},\beta_{SK}^{(2)})=&\ln 2 + \dfrac{1}{\theta}\mathbb{E}_1 \left[\ln \mathbb{E}_2 \cosh^\theta g_{_{SK}} (\bm J, \bar{m})\right] -\dfrac{\b}{2}\m^2+\dfrac{\left(\bskk\right)^2}{4}-\dfrac{\left(\bskk\right)^2}{2}\q_{SK}^{(2)}+{4}
     \\\\
     &+\theta\dfrac{\left(\bsk\q_{SK}^{(1)}\right)^2}+(1-\theta)\dfrac{\left(\bskk\q_{SK}^{(2)}\right)^2}{4} \,,
\end{array}
\label{eq:sk_1RSB_2}
\end{equation}
\begin{equation}
\footnotesize
\begin{array}{lll}
     \mathcal{A}_{Gauss}^{(1RSB)}(\bg,\bgg)=&-\theta\dfrac{\left(\bgg\q_{G}^{(2)}\right)^2}{4}+\theta\dfrac{\left(\bg\q_{G}^{(1)}\right)^2}{4}+\dfrac{\left(\bgg\q_{G}^{(2)}\right)^2}{4}-\dfrac{1}{2}\ln{\left[1-\lambda+\left(\bgg\right)^2\q_{G}^{(2)}\right]}
     \\\\
     &+\dfrac{1}{2\theta}\ln{\left[1+\theta\dfrac{\left(\bgg\right)^2\p_{G}^{(2)}-\left(\bg\right)^2\q_{G}^{(1)}}{1-\lambda+\left(\bgg\right)^2\q_{G}^{(2)}-\theta\left(\bgg\right)^2\q_{G}^{(2)}+\theta\left(\bg\right)^2\q_{G}^{(1)}}\right]}
     \\\\
     &+\dfrac{1}{2}\dfrac{\left(\bg\right)^2\q_{G}^{(1)}}{1-\lambda+\left(\bgg\right)^2\q_{G}^{(2)}-\theta\left(\bgg\right)^2\q_{G}^{(2)}+\theta\left(\bg\right)^2\q_{G}^{(1)}}\,,
\end{array}
\label{eq:Gauss_1RSB_2}
\end{equation}
where we used
\begin{equation}
\footnotesize
    \begin{array}{lll}
         g_{_{NN}}(\mathbf{J},\m)=&\beta\m+J^{(1)}\beta\dfrac{\sqrt{\gamma\q_1}}{1-(1-\q_2)\beta-\theta\beta(\q_2-\q_1)}+J^{(2)}\beta\sqrt{\dfrac{\gamma(\q_2-\q_1)}{[1-(1-\q_2)\beta-\theta\beta(\q_2-\q_1)][1-(1-\q_2)\beta]}}\,,
         \\\\
g_{_{SK}}(\mathbf{J},\m)=&\b\m+J^{(1)}\sqrt{\left(\bsk\right)^2\q_{SK}^{(1)}}+J^{(2)}\sqrt{\left[\left(\bskk\right)^2\q_{SK}^{(2)}-\left(\bsk\right)^2\q_{SK}^{(1)}\right]}\,.
    \end{array}
\end{equation}
We can have the following representation of the broken replica quenched statistical pressure of the Hopfield neural network:
\begin{equation}
    \mathcal{A}^{(1RSB)}_{NN}(\b,\gamma)=\mathcal{A}^{(1RSB)}_{SK}(\b,\bsk,\bskk)+\gamma\mathcal{A}^{(1RSB)}_{Gauss}(\b,\bg,\bgg)-\dfrac{\left(\bskk\right)^2}{4} \,.
\end{equation}
\end{theorem}

\begin{proof}
Starting with the 1-RSB SK model with the two noise level $\bsk$ and $\bskk$ defined in \eqref{eq:sk_1RSB_2}, using the definitions \eqref{eq:beta_interp_RSB}, we can verify the following relations
\begin{equation}
\footnotesize
    \begin{array}{lll}
         \left(\bsk\right)^2\q_{SK}^{(1)}=\gamma\b\p_1  \,,&& \left(\bskk\right)^2\q_{SK}^{(2)}=\gamma\b\p_2\,.
    \end{array}
\end{equation}
Putting these equations into \eqref{eq:sk_1RSB_2}, we get
\begin{equation}
\footnotesize
\begin{array}{lll}
     \mathcal{A}_{SK}^{(1RSB)}(\b,\bsk,\bskk)=&\ln 2 + \dfrac{1}{\theta}\mathbb{E}_1 \left[\ln \mathbb{E}_2 \cosh^\theta g_{_{(NN)}} (\bm J, \bar{m})\right] +\dfrac{\left(\bskk\right)^2}{4}-\dfrac{\b}{2}\m^2-\dfrac{\b\gamma}{2}\p_2
     \\\\
     & +(1-\theta)\dfrac{\b\gamma}{4}\p_2\q_2+\theta\dfrac{\b\gamma}{4}\p_1\q_1.
\end{array} 
\label{eq:SK_1RSB_final_2}
\end{equation}
Now using the 1-RSB quenched pressure for the soft model with the two noise $\bg$ and $\bgg$ (Eq. \ref{eq:Gauss_1RSB_2}), if we set $\lambda=\b$, using the definition \eqref{eq:beta_interp_RSB} we can verify the following relations
\begin{equation}
\footnotesize
    \begin{array}{lll}
         \left(\bg\right)^2\q_{G}^{(1)} =\b\q_1 && \left(\bgg\right)^2\q_{G}^{(2)}=\b\q_2
    \end{array}
\end{equation}
so the expression \eqref{eq:Gauss_1RSB_2} for the soft model quenched pressure becomes
\begin{equation}
\footnotesize
\begin{array}{lll}
     \mathcal{A}_{Gauss}^{(1RSB)}(\beta)=&-\theta\dfrac{\b}{4}\q_2\p_2+\theta\dfrac{\b}{4}\q_1\p_1+\dfrac{\b}{4}\q_2\p_2+\dfrac{1}{2}\dfrac{\b\q_1}{1-\b(1-\q_2)-\theta\b(\p_2-\p_1)}
     \\\\
     &+\dfrac{1}{2\theta}\ln{\left[1+\theta\dfrac{\b(\q_2-\q_1)}{1-\b(1-\q_2)-\theta\b(\p_2-\p_1)}\right]}-\dfrac{1}{2}\ln{\left[1-\b(1-\q_2)\right]}\,.
\end{array}
\label{eq:Gauss_1RSB_final_2}
\end{equation}
Now, using the equations \eqref{eq:SK_1RSB_final_2} and \eqref{eq:Gauss_1RSB_final_2}, if we compute $\mathcal{A}^{(1RSB)}_{SK}+\gamma\mathcal{A}^{(1RSB)}_{Gauss}-\frac{\left(\beta_{SK}^{(2)}\right)^2}{4}$ we get the proof.
\end{proof}

\subsubsection{Case $P>2$ (dense Hebbian network)}
We claim the following
\begin{theorem}
If we set
\begin{equation}
\footnotesize
    \begin{array}{lll}
        \bg= 1-(1-\q_2^{P/2})\b N^{1-P/2}-\theta \b N^{1-P/2}(\q_2^{P/2}-\q_1^{P/2})\,,&&  \bsk=\dfrac{\beta \sqrt{\gamma}}{\beta_{G}^{(1)}}\,,
        \\\\
        \bgg=\bg\sqrt{\dfrac{1-\b N^{1-P/2}(1-\q_2^{P/2})}{\bg+\theta\b N^{1-P/2}\left(1-\frac{\q_1^{P/2}}{\q_2^{P/2}}\right)}}\,,
          && \bskk=\dfrac{\beta \sqrt{\gamma}}{\beta_{G}^{(2)}} \,,
    \end{array}
    \label{eq:beta_interp_RSB_Pspin}
\end{equation}
recalling the finite size Guerra's expression for the 1RSB quenched statistical pressure of the Hopfield, hard P-spin glass and soft P-spin glass model, that read as
\begin{equation}
\footnotesize
\begin{array}{lll}
   &  \mathcal{A}_{NN}^{(P)(1RSB)}(\b,\gamma) = \ln 2-\dfrac{\beta '\gamma}{2}N^{P/2-1}+\dfrac{1}{\theta} \mathbb{E}_1  \left[\ln \mathbb{E}_2  \cosh^\theta g_{_{NN}}(\boldsymbol{J},\m)\right] -(P-1) \dfrac{\beta^{'}}{2} \bar{m}^P-\dfrac{\b\gamma}{4}P\p_2\q_2^{P/2-1}
      \\\\
&+\dfrac{\gamma\b}{2}\dfrac{N^{P/2-1}\q_1^{P/2} }{ 1-\beta^{'}N^{1-P/2}(1-\q_2^{P/2})-\theta \beta^{'}N^{1-P/2}(\q_2^{P/2} - \q_1^{P/2})}  +\dfrac{\gamma N^{P-2}}{2}\ln \left[1-\beta^{'}N^{1-P/2}(1-\bar{q}_2^{P/2})\right]  
\\\\
&+ \dfrac{\gamma N^{P-2}}{2\theta} \ln\left[\dfrac{1-\beta^{'}N^{1-P/2}(1-\bar{q}_2^{P/2})}{1-\beta^{'}N^{1-P/2}(1-\bar{q}_2^{P/2})-\theta \beta^{'}N^{1-P/2}(\q_2^{P/2} - \q_1^{P/2})}\right]-{\beta^{'} \gamma}(\theta-1) \dfrac{P}{4}\q_2^{P/2}\p_2
\\\\
      & +{\beta^{'} \gamma}\theta \dfrac{P}{4}\q_1^{P/2}\p_1+V_{N}^{(NN)(1RSB)}\,,
\label{eq:P_spin_Hopfield_1RSB}
\end{array}
\end{equation}
\vspace*{0.5cm}
\begin{equation}
\footnotesize
\begin{array}{lll}
     &\mathcal{A}_{SK}^{(P)(1RSB)}(\b,\beta_{SK}^{(1)},\beta_{SK}^{(2)})=\ln 2 + \dfrac{1}{\theta}\mathbb{E}_1 \left[\ln \mathbb{E}_2 \cosh^\theta g_{_{SK}} (\bm J, \bar{m})\right] -\b\dfrac{P-1}{2}\m^P+\dfrac{\left(\bskk\right)^2}{4}
     \\\\
     &+\theta(P-1)\dfrac{\left(\bsk\right)^2}{4}\left(\q_{SK}^{(1)}\right)^P-P\dfrac{\left(\bskk\right)^2}{4}\left(\q_{SK}^{(2)}\right)^{P-1} +(1-\theta)(P-1)\dfrac{\left(\bskk\right)^2}{4}\left(\q_{SK}^{(2)}\right)^P+V_N^{(SK)(1RSB)}\,,
\end{array}
\label{eq:SK_finite_size_1RSB}    
\end{equation}
\\\vspace*{0.5cm}
\begin{equation}
\footnotesize
\begin{array}{lll}
&\mathcal{A}_{Gauss}^{(P)(1RSB)}(\lambda,\bg,\bgg) =\dfrac{1}{2\theta} \log \left[\dfrac{1-\lambda + {(\b)}^2 \frac{P}{2}\left(\q_{G}^{(2)}\right)^{P-1}}{1-\lambda + \left(\bgg\right)^2\frac{P}{2} \left(\q_{G}^{(2)}\right)^{P-1} (1-\theta) +\left(\bg\right)^2\frac{P}{2}\left(\q_{G}^{(1)}\right)^{P-1} } \right] 
\\\\
&+ \dfrac{{\b}^2 \frac{P}{4}\left(\q_{G}^{(1)}\right)^{P-1}}{1-\lambda + \left(\bgg\right)^2\frac{P}{2} \left(\q_{G}^{(2)}\right)^{P-1} (1-\theta) +\left(\bg\right)^2\frac{P}{2}\left(\q_{G}^{(1)}\right)^{P-1}} - \dfrac{1}{2} \log \left(1-\lambda + \left(\bgg\right)^2 \frac{P}{2}\left(\q_{G}^{(2)}\right)^{P-1} \right)
 \\\\
&+(1-\theta)\dfrac{P-1}{4}\left(\bgg\right)^2\left(\q_{G}^{(2)}\right)^P+\theta\dfrac{P-1}{4}\left(\bg\right)^2\left(\q_{G}^{(1)}\right)^P +V_{N}^{(Gauss)(1RSB)}\,,
\\\\
\end{array}
\label{eq:soft_1RSB_P}
\end{equation} 
where we scaled $N^{P/2-1}\p\to\p$ and we used
\begin{equation}
\label{eq:g_SK_1RSB}
\footnotesize
\begin{array}{lll}
g_{_{NN}}(\bm J, \bar{m})=\dfrac{\beta^{'}P}{2}\bar{m}^{P-1} + J^{(1)}\sqrt{\dfrac{\beta^{'}}{2}\gamma \bar{p}_1P \bar{q}_1^{P/2-1}} +J^{(2)}\sqrt{ \dfrac{\beta^{'}}{2} P\gamma \left[ \bar{p}_2 \bar{q}_2^{P/2-1} -\bar{p}_1 \bar{q}_1^{P/2-1}  \right]}\,,
\\\\
      g_{_{SK}}(\mathbf{J},\m)=\b\dfrac{P}{2}\m^{P-1}+J^{(1)}\sqrt{\left(\bsk\right)^2\dfrac{P}{2}\left(\q_{SK}^{(1)}\right)^{P-1}}+J^{(2)}\sqrt{\dfrac{P}{2}\left[\left(\bskk\right)^2\left(\q_{SK}^{(2)}\right)^{P-1}-\left(\bsk\right)^2\left(\q_{SK}^{(1)}\right)^{P-1}\right]}\,.
\end{array}
\end{equation}
We can have the following finite size representation of the broken replica quenched statistical pressure of the dense Hebbian network:
\begin{equation}
\footnotesize
\begin{array}{lll}
    \mathcal{A}^{(P)(1RSB)}_{NN}(\b,\gamma)=&\mathcal{A}^{(P)(1RSB)}_{SK}(\b,\bsk.\bskk)-\dfrac{\left(\bskk\right)^2}{4}+\gamma  N^{P-2}\mathcal{A}^{(P)(1RSB)}_{Gauss}(\b,\bg,\bgg)+
    \\\\
    &+\gamma\b \dfrac{2-P}{4P}\left[\p_2\q_2^{P/2}-\theta\left(\p_2\q_2^{P/2}-\p_1\q_1^{P/2}\right)\right]
    \\\\
    &-\Big(V_N^{(SK)(1RSB)}+\gamma N^{P-2}V_N^{(Gauss)(1RSB)}-V_N^{(NN)(1RSB)}\Big)
\end{array}    
\label{eq:interpolation_1RSB_P}
\end{equation}
where the 1-RSB quenched pressure at finite size for $P\geq 2$ of Hopfield, SK and Gaussian spin glass models are indicated respectively with $\mathcal{A}^{(P)(1RSB)}_{NN}$, $\mathcal{A}^{(P)(1RSB)}_{SK}$ and $\mathcal{A}^{(P)(1RSB)}_{Gauss}$, and we scaled $N^{P/2-1}\p_{1/2}$ as $\p_{1/2}$.
\end{theorem}

\begin{proof}
The generalization to $P >2$ is obtained following the same steps presented in the proof for $P=2$ in Theorem \ref{th:P2_1RSB}, but taking care of using the new definitions of the noise in \eqref{eq:beta_interp_RSB_Pspin}.
\end{proof}

\begin{corollary}
In  the thermodynamic limit, for the case  of $P >2$, also in a broken replica symmetry framework, the  dense Hebbian network's quenched statistical  pressure is equivalent to that of the hard P-spin glass model.
\end{corollary}

\begin{proof}
To simplify the notation in this proof, since all the formulas will refer to the 1-RSB assumption, from now on we will omit the label \textit{1-RSB} for the sake of clearness.
\\
First of all, we can verify that, as $P>2$ in the thermodynamic limit, from the definitions of noise \eqref{eq:beta_interp_RSB_Pspin}, we get
\begin{equation}
    \footnotesize
    \begin{array}{lll}
         \bsk\to \b\sqrt{\gamma}\,,&& \bskk\to\b\sqrt{\gamma} \,.
    \end{array}
\end{equation}
From the quenched pressure, in the 1-RSB assumption, of the SK model with $P$ spin interactions  presented in \eqref{eq:SK_finite_size_1RSB}-\eqref{eq:g_SK_1RSB}, if we scaled, as usual, $N^{P/2-1}\p_{1/2}$ as $\p_{1/2}$, we can find the relations
\begin{equation}
\footnotesize
    \begin{array}{lll}
         \left(\bsk\right)^2\left(\q_{SK}^{(1)}\right)^{P-1}=\gamma\b \p_1\q_1^{P/2-1}\, , && \left(\bskk\right)^2\left(\q_{SK}^{(2)}\right)^{P-1}=\gamma\b \p_2\q_2^{P/2-1}\,.
    \end{array}
\end{equation}
Thus, (remembering that, in the thermodynamic limit, $V_N^{(NN)}$, $V_N^{(SK)}$ and $V_N^{(Gauss)}$ will vanish) for the 1-RSB assumption the SK quenched pressure reads as
\begin{equation}
\footnotesize
\begin{array}{lll}
     \mathcal{A}_{SK}^{(P)}(\b,\b\sqrt{\gamma},\b\sqrt{\gamma})=&\ln 2 + \dfrac{1}{\theta}\mathbb{E}_1 \left[\ln \mathbb{E}_2 \cosh^\theta g_{_{NN}} (\bm J, \bar{m})\right] -\dfrac{\gamma\b}{4}\q_2^{P/2-1}\p_2\Big(P-(P-1)\q_2\Big)
     \\\\
     &-\b\dfrac{P-1}{2}\m^P-\theta(P-1)\dfrac{\b\gamma}{4}(\q_2^{P/2}\p_2-\q_1^{P/2}\p_1)+\dfrac{\beta'\,^2\gamma}{4}.
\end{array}    
\end{equation}
Focusing, now, on the soft model in the 1-RSB assumption in the case of P-spin interactions (Eq. \ref{eq:soft_1RSB_P}), if we set $\lambda=\b N^{1-P/2}$,  using the definition \eqref{eq:beta_interp_RSB_Pspin}, we can get the relations
\begin{equation}
\footnotesize
    \begin{array}{lll}
         \left(\bg\right)^2\frac{P}{2}\left(\q_{G}^{(1)}\right)^{P-1} =N^{1-P/2}\b\q_1^{P/2} \,,&& \left(\bgg\right)^2\frac{P}{2}\left(\q_{G}^{(2)}\right)^{P-1}=N^{1-P/2}\b\q_2^{P/2}\,.
    \end{array}
\end{equation}
Thus, in the thermodynamic limit, for $P>2$ the quenched pressure for the soft model in the 1-RSB assumption reads as
\begin{equation}
\footnotesize
\begin{array}{lll}
     \gamma N^{P-2}\mathcal{A}^{(P)}_{Gauss}(\beta,\bg,\bgg)-\dfrac{\gamma\b}{2}N^{P/2-1}\xrightarrow[P>2]{\;\;N\to\infty\;\;} &
     \gamma\beta'^2\dfrac{1}{4}\q_2^P-\gamma\beta'^2\dfrac{1}{4}\theta (\q_2^{P}-\q_1^{P})+
     \\\\
     &+\gamma\beta '^2\dfrac{1}{2P}\theta (\q_2^{P}-\q_1^{P})-\gamma\beta '^2\dfrac{1}{2P}\q_2^{P}.
\end{array}
\end{equation}
Moreover, the last term of \eqref{eq:interpolation_1RSB_P} becomes 
\begin{equation}
\footnotesize
\begin{array}{lll}
    \gamma\b N^{P/2-1}\dfrac{2-P}{4P}\left[\p_2\q_2^{P/2}-\theta\left(\p_2\q_2^{P/2}-\p_1\q_1^{P/2}\right)\right]\xrightarrow[P>2]{\;\;N\to\infty\;\;}&-\gamma\beta '^2 \dfrac{1}{4}\q_2^{P}+\gamma\beta '^2 \dfrac{1}{4}\theta\left(\q_2^{P}-\q_2^{P}\right)
    \\\\
    & +\gamma\beta '^2 \dfrac{1}{2P}\q_2^{P}-\gamma\beta '^2 \dfrac{1}{2P}\theta\left(\q_2^{P}-\q_2^{P}\right)\,.
\end{array}
\end{equation}
So putting all together in relation \eqref{eq:interpolation_1RSB_P}, for $P >2$ in the thermodynamic limit, we have
\begin{equation}
     \mathcal{A}_{NN}^{(P)}(\b,\gamma)={\mathcal{A}^{(P)}}^{(1RSB)}_{SK}(\b,\b\sqrt{\gamma},\b\sqrt{\gamma})\,.
\end{equation}
\end{proof}

\section{Conclusions and outlooks} \label{conclusions} 

In this paper we focused on replica symmetry breaking in dense Hebbian networks (namely generalized Hopfield networks whose neural dialogues are broader than pairwise) and we gave both mathematical instruments to address this phenomenon as well as physical insights. In particular, regarding the methodology,  we adapted to these dense networks two different approaches, the first consists in constructing effective PDE -the transport equation in particular- in the space of the coupling constants and then relying upon the arsenal of mathematical results available in PDE theory while the second  is a generalization of the celebrated Guerra's interpolation scheme \cite{Guerra,GuerraSum}, more grounded on Probability Theory. Whatever the route, at the end of the calculations, we obtained a set of self-consistent equations for the order parameters whose solutions trace their evolution  in the control parameter space, ultimately allowing the construction of phase diagrams that we provided both at the replica symmetric level and under one step of replica symmetry breaking level of description.
\newline
Restricting to the Baldi-Venkatesh regime \cite{baldi}, i.e. the {\em high storage regime} for dense networks, at first we recovered in the replica symmetric scenario, the Gardner's picture \cite{gardner} (achieved in the eighties via heuristic techniques, i.e. the replica trick) in every detail, even the scaling of the divergence of the critical storage $\gamma_c(P) \sim P/\ln(P)$ as $P \to \infty$, then we inspected the replica symmetry breaking phenomenon just at the first step of symmetry breaking: as expected, the critical storage is mildly affected by RSB, however a glance at the phase diagrams in the two frameworks (RS and 1-RSB) immediately reveals that the spin-glass phase (that naively shrinks to zero in the RS picture) gets stabilized and actually enlarged by the RSB phenomenon.
\newline
Indeed the type of glassiness underlying these networks is rather different w.r.t. the type of glassiness of the standard Hopfield neural network:  while the quenched statistical pressure (or free energy) of the latter can be written (both in the RS and 1-RSB scenarios) as a weighted linear combination of the quenched statistical pressures of two spin glasses, one being the hard spin glass (the Sherrington-Kirkpatrick model) and the other being the soft spin glass (the Gaussian model), this is no longer true in dense networks where the soft contribution disappears: as the Sherrington-Kirkpatrick model is full-RSB (it is the harmonic oscillator for Parisi theory) while the Gaussian model is solely replica symmetric, the disappearance of the soft contribution makes dense networks different w.r.t. the Hopfield reference.  A subtle point is that, in the 1-RSB picture  for neural networks,  the ziqqurat  prescription introduced in \cite{Ziguli1,Ziguli2} (that naturally generalizes Parisi's ansatz to the case) breaks the permutational-invariance both for $q_{12}$ as well as for $p_{12}$: the breaking of self-averaging of $p_{12}$ is not a propriety of the pairwise Gaussian spin-glass per se, but -rather- a consequence of the interactions among these spin glasses (interactions that  shines in the integral representation of the partition function of these Hebbian neural networks, see e.g. eq. (\ref{ducojoni2}), hence for the soft overlap the transition is not spontaneous, but driven by the hard one \cite{Driven} in the shallow limit.
\newline
The inspection of replica symmetry breaking phenomenon in dense networks in the {\em high resolution regime} \cite{PRLNN} -rather than in the {\em high storage regime}-is entirely missing at present: we plan to report in a separate paper investigations in that regime.

\appendix

%

\section{Proof of Theorem One}  

In this appendix we provide the explicit calculations behind the proof of Theorem \ref{cor_carmassimoRS}.
\begin{proof}
Since the potential $V_N(t, \bm r)$ vanishes in the thermodynamical limit, we solve the following transport equation 
\begin{equation}
	\begin{array}{lll}
	     \pder{\mathcal{A}^{(P)}_{\textrm{RS}}}{t}-{\beta '\q^{P/2}} \left(\pder{\mathcal{A}^{(P)}_{\textrm{RS}}}{x}\right)-\dfrac{P}{2}{\beta '\gamma\,\p\q^{P/2-1}} \left(\pder{\mathcal{A}^{(P)}_{\textrm{RS}}}{y}\right)- \beta'(1-\q^{P/2})\left(\pder{\mathcal{A}^{(P)}_{\textrm{RS}}}{z} \right)+
	     \\\\-\dfrac{P}{2}\beta '\m^{P-1} \left(\pder{\mathcal{A}^{(P)}_{\textrm{RS}}}{w}\right)=-\dfrac{P-1}{2}\beta ' \m^P-\dfrac{P\, \beta '\gamma }{4}{\p\q^{^{P/2-1}}(1-\q)},
	\end{array}
	\label{hop_GuerraAction_RSDE}
	\end{equation}
We compute the solution using the characteristic method on the transport equation: 
\begin{align}
    \mathcal{A}^{(P)}_{RS}(t, \bm r) = \mathcal{A}^{(P)}_{RS}(0, \bm r-\dot{\boldsymbol{r}}t) + S(t, \bm r)t.
\end{align}
where $\dot{\boldsymbol{r}}=(\dot{x},\dot{y},\dot{z},\dot{w})$. Along the characteristics, the fictitious motion in the $(t,\boldsymbol{r})$ time-space is linear and returns
\begin{equation}
    \begin{array}{lll}
         x=x_0-\beta '  \q^{P/2}t && y=y_0-\dfrac{P}{2}\beta ' \gamma \p\q^{P/2-1} t   
         \\\\
         z=z_0-\beta' (1-\q^{P/2})t && w=w_0-\dfrac{P}{2}\beta '\m^{P-1}t
    \end{array}
\end{equation}
where $\boldsymbol{r}_0=(x_0,y_0,z_0,w_0)=(x(t=0),y(t=0),z(t=0),w(t=0))$. The Cauchy condition at $t=0$ is given by a direct computation at finite $N$ as
\footnotesize
\begin{equation}
\begin{array}{lll}
     & &\mathcal A^{(P)}(0,\boldsymbol{r}-\dot{\boldsymbol{r}}t)=\mathcal A^{(P)}(0,\boldsymbol{r}_0)
     \\\\
     & =&\dfrac{1}{N}\mathbb{E}\Bigg\lbrace\sommaSigma{\boldsymbol \sigma} \displaystyle\int \mathcal{D} \bm \tau\exp{}\Bigg[w_0 N\psi\,m(\boldsymbol \sigma)+\sqrt{x_0 N^{1-P/2}}\SOMMA{\mu>1}{K} \tilde{J}_{\mu}\tau_{\mu}
     +\sqrt{y_0}\SOMMA{i=1}{N} J_i\sigma_i+\dfrac{z_0 N^{1-P/2}}{2}\SOMMA{\mu>1}{K}\,{\tau^2_{\mu}}-\dfrac{\beta'\gamma}{2}N^{a+1-P/2}\Bigg]\Bigg\rbrace
     \\\\
     &=& \dfrac{1}{N}\mathbb{E}\ln{}\Bigg\lbrace\sum\limits_{\lbrace\sigma\rbrace} \exp{}\Big[\sum_i(w_0 +B J_i)\sigma_i\Big]\Bigg\rbrace+\dfrac{1}{N}\mathbb{E}\ln{}\Bigg\lbrace\prod\limits_{\mu>1}^{K}\displaystyle\int \dfrac{d\tau_{\mu}}{\sqrt{2\pi}}\,e^{(1-z_0N^{1-P/2})\tau^{\2}_{\mu}/2+\sqrt{x_0N^{1-P/2}}\tilde{J\,}_{\mu}\tau_{\mu}}\Bigg\rbrace -\dfrac{\beta '\gamma}{2}N^{a-P/2}
         \\\\
    & =&\ln{2}-\dfrac{\beta'\gamma}{2}N^{a-P/2}+\left\l\ln{\cosh{\left[\omega_0+J\sqrt{y_0}\right]}}\right\r_{J}-\dfrac{K}{2N}\ln{\left(1-z_0N^{1-P/2}\right)}+\dfrac{K}{2N}\dfrac{x_0N^{1-P/2}}{1-z_0N^{1-P/2}}\,.
\end{array}
\end{equation}
\normalsize

Giving the suitable values of parameters, namely $t=1$ and $\bm r=0$ we have the following 

\begin{equation}
\begin{array}{lll}
      \mathcal{A}^{(P)}(\gamma, \beta) & \coloneqq &\ln{2} -\dfrac{\beta' \gamma}{2}N^{a-P/2}+\left\langle\ln{\cosh{\left[\dfrac{P}{2}\beta '  \m^{P-1}+Y\sqrt{\beta ' \gamma  \dfrac{P}{2}\, \p \q^{^{P/2-1}}} \right]}}\right\rangle_Y+
      \\\\
      & & -\dfrac{P-1}{2} \beta ' \m^P -\dfrac{\beta ' \gamma}{4} P\,\p \q^{P/2-1}(1-\q)-\dfrac{\gamma N^{^{a-1}}}{2}\ln{\left(1-\beta 'N^{^{1-P/2}}\left(1-\q^{^{P/2}}\right)\right)}+
       \\\\
       && +\dfrac{\gamma N^{a-P/2}}{2}\dfrac{\beta '\q^{^{P/2}}}{1-\beta 'N^{^{1-P/2}}\left(1-\q^{^{P/2}}\right)}\,.
\end{array}
\label{eq:pressure_GuerraRS_noAPP}
\end{equation}

Now, expanding the two last member of \eqref{eq:pressure_GuerraRS_noAPP} for large value of $N$, (remembering the conditions $a\geq 1$ and $P\geq 2$ given in Definition \ref{def:pspinham}) it is possible to write the expression 
 \begin{equation}
 \begin{array}{lll}
       \mathcal{A}^{(P)}(\gamma, \beta) & \coloneqq & \ln{2}+\left\langle\ln{\cosh{\left[\dfrac{P}{2}\beta '  \m^{P-1}+Y\sqrt{\beta ' \gamma  \dfrac{P}{2}\,\p \q^{^{P/2-1}}} \right]}}\right\rangle_Y+
      \\\\
       & & -\dfrac{P-1}{2} \beta ' \m^P -\beta ' \gamma\dfrac{P}{4} \,\p \q^{P/2-1}(1-\q)+\dfrac{1}{4} N^{a+1-P}\gamma\beta '^2\left(1 - \q^P\right)+\mathcal{O}\left(N^{a-P}\right)\,.
 \end{array}
 \end{equation}
Thus, since we need that, in the thermodynamic limit, the quenched statistical pressure must be intensive in $N$ (we want that the terms in $\mathcal{O}(N^{a-P})$ vanish), it is necessary to ensure that
\begin{equation}
    a\leq P-1
\end{equation}
Moreover, it is easy to check that the only no trivial case is $a=P-1$, otherwise our model turns into a ferromagnet \cite{fachechi} with polynomial interaction of degree $P$. 

So, if we consider $P \geq 4$ (since we have rescaled $N^{P/2 -1} \p$ in $\p$), the previous expression reads as
\begin{equation}
\begin{array}{lll}
      A^{(P)} (\gamma, \beta) & \coloneqq &\ln{2}+\left\langle\ln{\cosh{\left[\dfrac{P}{2}\beta '  \m^{P-1}+Y\sqrt{\beta ' \gamma \dfrac{P}{2}\,   \p \q^{^{P/2-1}}} \right]}}\right\rangle_Y+
     \\\\
      & & -\dfrac{P-1}{2} \beta ' \m^P -\beta ' \gamma\dfrac{P}{4} \, \p \q^{P/2-1}(1-\q)+\dfrac{1}{4}\gamma\beta '^2\left(1 - \q^P\right)+\mathcal{O}\left(N^{-1}\right)\,;
\end{array}
\end{equation}
in the thermodynamics limit ($N\to\infty$) the correction terms will vanish.
\end{proof}

\section{Proof of Proposition \ref{prop:9}}\label{Appendix-Sua1}

Hereafter we give all the details regarding the proof of Proposition \ref{prop:9}.

\begin{proof}
Similar to the case RS (\ref{potential_m}-\ref{potential_pq}) and we have for $a=1,2$
\begin{align}
    \langle p_{12}q_{12}^{P/2}\rangle_a =& \sum_{k=1}^{P/2} \begin{pmatrix}\frac{P}{2}\\ k\end{pmatrix}\bar{q}_a^{P/2-k} \langle (\Delta p_a)(\Delta q_a)^k \rangle_a + \sum_{k=2}^{P/2} \begin{pmatrix}\frac{P}{2}\\ k\end{pmatrix} \bar{q}_a^{P/2-k} \bar{p}_a\langle (\Delta q_a)^k \rangle_a \notag \\
    &+\bar q_a^{P/2}\l\pp\r_a+\dfrac{P}{2}\bar q_a^{P/2-1}\bar p_a \l\qq\r_a-\dfrac{P}{2}\bar q_a^{P/2}\bar p_a
\end{align}
\normalsize
where we use $\Delta X_a=X-\bar{X}_a$. Now, starting to evaluate explicitly $\dt \mathcal{A}^{(P)}_N$ by using (\ref{eqn:partialx1A} - \ref{eqn:partialwA}) we write
\footnotesize
\begin{align}
\label{A_1RSB_thermo}
\dt \mathcal{A}^{(P)}_N=& \frac{\beta}{2} \left[ \sum_{k=2}^P \binom{P}{K} \langle (m_1-\bar{m})^k \rangle \bar{m}^{P-k} + \bar{m}^P (1-P) + P\bar{m}^{P-1}\langle m_1 \rangle \right] + \frac{\beta^{'}K}{2N^{P/2}} \Bigg\{ \langle p_{11} \rangle   \notag \\
\textcolor{white}{\dt \mathcal{A}^{(P)}_N=}&+ (\theta-1)\left[ \sum_{K=1}^{P/2} \binom{P/2}{K}\bar{q}_2^{P/2-K} \langle (p_{12}-\bar{p}_2)(q_{12} - \bar{q}_2)^K \rangle_2 + \sum_{K=2}^{P/2} \binom{P/2}{K} \bar{q}_2^{P/2-K} \bar{p}_2\langle (q_{12} - \bar{q}_2)^K \rangle_2\right.\notag \\
\textcolor{white}{\dt \mathcal{A}^{(P)}_N=}&\left.+\bar q_2^{P/2}\l\pp\r_2+\dfrac{P}{2}\bar q_2^{P/2-1}\bar p_2 \l\qq\r_2-\dfrac{P}{2}\bar q_2^{P/2}\bar p_2\right] - \theta \left[ \sum_{K=1}^{P/2} \binom{P/2}{K}\bar{q}_1^{P/2-K} \langle (p_{12}-\bar{p}_1)(q_{12} - \bar{q}_1)^K \rangle_1 \right.\notag \\
\textcolor{white}{\dt \mathcal{A}^{(P)}_N=}&\left.+ \sum_{K=2}^{P/2} \binom{P/2}{K} \bar{q}_1^{P/2-K} \bar{p}_1\langle (q_{12} - \bar{q}_1)^K \rangle_1+\bar q_1^{P/2}\l\pp\r_1+\dfrac{P}{2}\bar q_1^{P/2-1}\bar p_1 \l\qq\r_1-\dfrac{P}{2}\bar q_1^{P/2}\bar p_1\right]\Bigg\} = \notag \\
=&V_N(t, \bm r) + \frac{\beta^{'} \bar m^P (1-P)}{2}-\frac{\beta^{'}K (\theta-1)}{2N^{P/2}}\frac{P}{2}\bar{p}_2 \bar{q}_2^{P/2} + \frac{\beta^{'}K \theta}{2N^{P/2}}\frac{P}{2} \bar{p}_1 \bar{q}_1^{P/2} + \frac{\beta^{'}P}{2}\bar{m}^{P-1} \partial_w \mathcal{A}^{(P)}_N \notag \\
&+ \frac{\beta^{'}K}{2N^{P/2}} \langle p_{11} \rangle + \frac{\beta^{'}K}{2N^{P/2}} (\theta-1) \bar{q}_2^{P/2} \langle p_{12} \rangle_2 + \frac{\beta^{'}K}{2N^{P/2}} (\theta -1) \frac{P}{2} \bar{p}_2 \bar{q}_2^{P/2-1} \langle q_{12} \rangle_2 \notag \\
&+ {\beta^{'}\bar{q}_1^{P/2}} \left[ \partial_{x^{(1)}} \mathcal{A}^{(P)}_N - \frac{K}{2N^{P/2}} \langle p_{11} \rangle - \frac{K}{2N^{P/2}} (\theta-1) \langle p_{12} \rangle_2 \right] \notag \\
&+ \frac{\beta^{'} K P \bar{p}_1 \bar{q}_1^{P/2-1}}{2N^{P/2}} \left[\partial_{y^{(1)}} \mathcal{A}^{(P)}_N - \frac{1}{2} - \frac{1}{2} (\theta-1) \langle q_{12} \rangle_2 \right]   \notag \\
=& V_N(t, \bm r) + S(t, \bm r) + \frac{\beta^{'}P}{2}\bar{m}^{P-1} \partial_w \mathcal{A}^{(P)}_N + {\beta^{'} \bar{q}_1^{P/2}}\partial_{x^{(1)}} \mathcal{A}^{(P)}_N + {\beta^{'} (\bar{q}_2^{P/2}-\bar{q}_1^{P/2})}\partial_{x^{(2)}} \mathcal{A}^{(P)}_N \notag \\
&+ \frac{\beta^{'} KP }{2 N^{P/2}}\bar{p}_1 \bar{q}_1^{P/2-1}\partial_{y^{(1)}}\mathcal{A}^{(P)}_N+\frac{\beta^{'} KP }{2 N^{P/2}}(\bar{p}_2\bar{q}_2^{P/2-1} - \bar{p}_1 \bar{q}_1^{P/2-1})\partial_{y^{(2)}}\mathcal{A}^{(P)}_N + {\beta^{'}}(1-\bar q_2^{P/2})\partial_z \mathcal{A}^{(P)}_N
\end{align}
\normalsize
Thus, by placing
\begin{align}
 \label{eqn:dotx1}
\dot{x}^{(1)}&= -\beta^{'} \bar{q}_1^{P/2} \\
\dot{x}^{(2)}&= - \beta^{'} (\bar{q}_2^{P/2}-\bar{q}_1^{P/2})\\
\dot{y}^{(1)}&= -  \frac{\beta^{'} KP }{2 N^{P/2}}\bar{p}_1 \bar{q}_1^{P/2-1} \\
\dot{y}^{(2)}&= - \frac{\beta^{'} KP }{2 N^{P/2}}(\bar{p}_2\bar{q}_2^{P/2-1} - \bar{p}_1 \bar{q}_1^{P/2-1}) \\
\dot{z}&= -\beta^{'}(1-\bar q_2^{P/2})\\
\dot{w} &= - \frac{\beta^{'}P}{2}\bar{m}^{P-1} \label{eq:dotw}
\end{align}
and $N^{a-P/2}\p_2 $, $N^{a-P/2}\p_1 $ with $\p_2$, $\p_1$ we reach the thesis.

\end{proof}

\section{Proof of Theorem \ref{Susy1}}\label{SusyBreak}
In this section we provide all the details regarding the proof of Theorem \ref{Susy1}.
\begin{proof}
Since the potential $V_N(t, \bm r)$ vanishes in the thermodynamical limit, we can apply Remark \ref{r:above} and solve the following equation
\begin{align}
&\partial_t \mathcal A^{(P)}+\dot x^{(1)}\partial_{x_1} \mathcal A^{(P)} +\dot x^{(2)}\partial_{x_2} \mathcal A^{(P)} +\dot y^{(1)} \partial_{y_1} \mathcal A^{(P)} +\dot y^{(2)} \partial_{y_2} \mathcal A^{(P)}+\dot z \partial_{z} \mathcal A^{(P)}+\dot w \partial_{w} \mathcal A^{(P)}  \notag \\
&\textcolor{white}{+\dot z \partial_{z} \mathcal A^{(P)}+\dot w \partial_{w} \mathcal A^{(P)} }= \frac{\beta^{'} \bar m^P (1-P)}{2}-{\beta^{'}\gamma (\theta-1)}\frac{P}{2}\bar{p}_2 \bar{q}_2^{P/2} +{\beta^{'}\gamma \theta}\frac{P}{2} \bar{p}_1 \bar{q}_1^{P/2} - {\beta^{'}\gamma }\frac{P}{2}\bar{p}_2 \bar{q}_2^{P/2-1}.
\end{align}
We use the characteristic method to solve it and, after one body computation in the similar way as RS assumption, we find the explicit solution 
\footnotesize
\begin{align}
\label{pressfinaleGuerraRSB}
    &\mathcal{A}^{(P)} = \frac{1}{\theta} \mathbb{E}_1 \left\{ \ln \mathbb{E}_2  \cosh^\theta\left( \beta^{'}\frac{P}{2} \bar{m} ^{P-1}+  \sqrt{\beta^{'}\gamma P \p_1 \q_1^{P/2-1}} J^{(1)}+  \sqrt{\beta^{'}\gamma P (\p_2\q_2^{P/2-1} - \p_1 \q_1^{P/2-1})} J^{(2)} \right) \right\}\notag \\
&+\frac{ \gamma \beta^{'}N^{a-P/2}\q_1^{P/2} }{2 \left[1-\beta^{'}N^{1-P/2}(1-\q_2^{P/2})-\theta \beta^{'}N^{1-P/2}(\q_2^{P/2} - \q_1^{P/2})\right]}+ \ln 2 +\frac{\gamma N^{a-1}}{2}\ln (1-\beta^{'}N^{1-P/2}(1-\bar{q}_2^{P/2}))    \notag \\
&+ \frac{\gamma N^{a-1}}{2\theta} \ln\left(\frac{1-\beta^{'}N^{1-P/2}(1-\bar{q}_2^{P/2})}{1-\beta^{'}N^{1-P/2}(1-\bar{q}_2^{P/2})-\theta \beta^{'}N^{1-P/2}(\q_2^{P/2} - \q_1^{P/2}))}\right)-\dfrac{\beta '\gamma}{2}N^{a-P/2}\notag \\
&+ \frac{\beta^{'}}{2}\bar{m}^P (1-P)  -{\beta^{'} \gamma}(\theta-1) \dfrac{P}{2}\q_2^{P/2}\p_2  +{\beta^{'} \gamma}\theta \dfrac{P}{2}\q_1^{P/2}\p_1 .
\end{align}
\normalsize

Expanding some factors of \eqref{pressfinaleGuerraRSB} for large value of $N$, (remembering the conditions $a\geq 1$ and $P\geq 2$) it is possible to write the following expression for the quenched  pressure
\begin{align}
\mathcal{A}^{(P)} =& \ln 2 +\beta^{'}\gamma N^{a-P/2} + \frac{1}{\theta} \mathbb{E}_1 \ln \mathbb{E}_2 \cosh^\theta g( \bm J, \bar{m}) -  N^{a+1-P}\frac{\beta^{'}}{4}\gamma P\bar{p}_2\bar{q}_2^{P/2-1} + \frac{\beta^{'}}{2}\bar{m}^P (1-P) \notag \\
&-\theta(P-1)\dfrac{\b\gamma}{4}(\q_2^{P/2}\p_2-\q_1^{P/2}\p_1) + \frac{1}{4} {\beta^{'}}^2 \gamma.
\end{align}
Similar to RS assumption, a must satisfy the condition $a=P-1$ and
the previous expression for even $P \geq 4$ reads as
\begin{align}
 &\mathcal{A}^{(P)}=\ln 2 + \frac{1}{\theta}\mathbb{E}_1 \ln \mathbb{E}_2 \cosh^\theta g (\bm J, \bar{m}) -\dfrac{\gamma\b}{4}\q_2^{P/2-1}\p_2\Big(P-(P-1)\q_2\Big) \notag \\
    &+ \frac{\beta^{'}}{2}\bar{m}^P (1-P) -\theta(P-1)\dfrac{\b\gamma}{4}(\q_2^{P/2}\p_2-\q_1^{P/2}\p_1) + \frac{1}{4} {\beta^{'}}^2 \gamma+ \mathcal{O}(N^{-1})
    \label{eq:correct_1RSB_Trasp}
\end{align}
in the thermodynamics limit ($N\to\infty$) the correction terms will vanish.
\end{proof}

\section{Proof of Lemma \ref{lemma:4}} \label{app2}
We prove only (\ref{eqn:partialtA}) regarding Lemma \ref{lemma:4}, being the proofs for the others obtained in the same way. First of all, using (\ref{eqn:partialrA}) we see that
\begin{align}
\label{eqn:partialtAproof1}
\partial_t \mathcal{A}_N =\frac{\beta^{'}}{2}\langle m_1^P \rangle+\sqrt{\frac{\beta^{'}}{N^{P-1}}}\frac{1}{2N \sqrt{t}}\mathbb{E}_0\mathbb{E}_1\mathbb{E}_2 \left[ \mathcal{W}_2\sum_{ \bm i,\mu}\xi_i^\mu\omega( \sigma_{i_1} \hdots \sigma_{i_{P/2}}\tau_\mu) \right]
\end{align}
Now, using Stein's lemma (\ref{eqn:gaussianrelation2}), we may rewrite the second member of (\ref{eqn:partialtAproof1}) as
\begin{align}
\label{eqn:partialtAproof2}
\sqrt{\frac{\beta^{'}}{N^{P-1}}}\frac{1}{2N \sqrt{t}}\sum_{\bm i,\mu}\mathbb{E}_0\mathbb{E}_1 \mathbb{E}_2 \left[\partial_{\eta_{\bm i}^\mu}\bigg(\mathcal{W}_2\omega( \sigma_{i_1} \hdots \sigma_{i_{P/2}}\tau_\mu) \bigg)\right] =D_1+D_2+D_3
\end{align}
Let's investigate those three terms:
\footnotesize
\begin{align}
\label{eqn:D1}
D_1 =&\sqrt{\frac{\beta^{'}}{N^{P-1}}}\frac{1}{2N \sqrt{t}}\sum_{ \bm i,\mu}\mathbb{E}_0 \mathbb{E}_1  \mathbb{E}_2 \bigg[\mathcal{W}_2\partial_{\xi_i^\mu}\omega( \sigma_{i_1} \hdots \sigma_{i_{P/2}}\tau_\mu)\bigg] = \notag \\
=&\frac{\beta^{'}}{2N^P}\sum_{\bm i,\mu}\mathbb{E}_0 \mathbb{E}_1  \mathbb{E}_2 \left[\mathcal{W}_2\omega( (\sigma_{i_1} \hdots \sigma_{i_{P/2}}\tau_\mu)^2) \right]  -\frac{\beta^{'}}{2N^P}\sum_{\bm i,\mu}\mathbb{E}_0 \mathbb{E}_1 \mathbb{E}_2 \left[\mathcal{W}_2\omega(\sigma_{i_1} \hdots \sigma_{i_{P/2}}\tau_\mu)^2 \right]=\nonumber \\
=&\frac{\beta^{'}K}{2N^{P/2}} \left[ \langle p_{11} \rangle-\langle p_{12}q_{12}^{P/2} \rangle_2 \right] \\
D_2=&\sqrt{\frac{\beta^{'}}{N^{P-1}}}\frac{1}{2N \sqrt{t}}\sum_{\bm i,\mu}\mathbb{E}_0 \mathbb{E}_1  \mathbb{E}_2 \left[\omega( \sigma_{i_1} \hdots \sigma_{i_{P/2}}\tau_\mu)\frac{\partial_{\eta_{\bm i}^\mu}\mathcal Z_2^\theta}{\mathbb{E}_2 \left(\mathcal Z_2^\theta\right)} \right]= \notag \\
=&\frac{\beta^{'} \theta}{2N^P}\sum_{\bm i,\mu}\mathbb{E}_0 \left\{ \mathbb{E}_1 \left[ \mathbb{E}_2 \left(\mathcal{W}_2\omega( \sigma_{i_1} \hdots \sigma_{i_{P/2}}\tau_\mu)^2 \right) \right] \right\}=  \frac{\beta^{'}K}{2N^{P/2}}\theta \langle p_{12}q_{12}^{P/2} \rangle_2  \label{eqn:D2} \\
D_3=&\sqrt{\frac{\beta^{'}}{N^{P-1}}}\frac{1}{2N \sqrt{t}}\sum_{\bm i,\mu}\mathbb{E}_0 \mathbb{E}_1 \mathbb{E}_2 \left[\omega( \sigma_{i_1} \hdots \sigma_{i_{P/2}}\tau_\mu){\mathcal Z_2^{(P)}}^\theta\partial_{\eta_{\bm i}^\mu}\frac{1}{\mathbb{E}_2 \left({\mathcal Z_2^{(P)}}^\theta\right)} \right]= \notag \\
=&-\frac{\beta^{'}\theta}{2N^P}\sum_{\bm i,\mu}\mathbb{E}_0 \mathbb{E}_1 \mathbb{E}_2 \left[\omega( \sigma_i\tau_\mu)\mathcal{W}_2\mathbb{E}_2 \left(\mathcal{W}_2\frac{\partial_{\eta_{\bm i}^\mu}\mathcal Z_2^{(P)}}{\mathcal Z_2^{(P)}}\right)\right] = \nonumber \\
=&-\frac{\beta^{'}\theta}{2N^P}\sum_{i,\mu}\mathbb{E}_0\mathbb{E}_1 \mathbb{E}_2 \left[\omega( \sigma_{i_1} \hdots \sigma_{i_{P/2}}\tau_\mu)\mathcal{W}_2\mathbb{E}_2 \left(\omega(\sigma_{i_1} \hdots \sigma_{i_{P/2}}\tau_\mu)\mathcal{W}_2\right) \right]=-\frac{\beta^{'}K}{2N^{P/2}}\theta\langle p_{12}q_{12}^{P/2} \rangle_1
\label{eqn:D3}
\end{align}
\normalsize
Putting (\ref{eqn:D1}),  (\ref{eqn:D2}) and (\ref{eqn:D3}) inside (\ref{eqn:partialtAproof2}), and (\ref{eqn:partialtAproof2}) inside (\ref{eqn:partialtAproof1}) we find (\ref{eqn:partialtA}).

\section*{Acknowledgments}
The Authors acknowledge  INFN (FIELDTURB) and INdAM (GNFM) for providing computational facilities and the grants  by MUR (PRIN 2017, Project
no. 2017JFFHSH, {\em Stochastic Models for Complex Systems}) and by MEACI (Project ''BULBUL'': {\em Scientific, technological and industrial cooperation between Italy and Israel}).


\begin{thebibliography}{99}

\bibitem{NoiConLenka} E. Agliari, A. Barra, P. Sollich, L. Zdeborova, {\em Machine learning and statistical physics: theory, inspiration, application}, J. Phys. A: Special Volume (2020).

\bibitem{NPD} E. Agliari, F.E. Leonelli, C. Marullo, {\em Storing, learning and retrieving biased patterns}, Appl. Math. $\&$ Comp. \textbf{415}, 126716, (2021).


\bibitem{deepTransport}  E. Agliari, L. Albanese, F. Alemanno, A. Fachechi, {\em A transport equation approach for deep neural networks}, arXiv preprint arXiv:2106.08978. 

\bibitem{Alemannation1} E. Agliari, F. Alemanno, A. Barra, A. Fachechi, {\em Dreaming neural networks: rigorous results},  J. Stat. Mech. 083503 (2019). 

\bibitem{Agliari-Barattolo} E. Agliari, A. Barra, C. Longo, D. Tantari, {\em Neural Networks retrieving binary patterns in a sea of real ones}, J. Stat. Phys. \textbf{168}, 1085, (2017).

\bibitem{lindaRSB} E. Agliari, L. Albanese, A.Barra, G. Ottaviani {\em Replica symmetry breaking in neural networks: a few steps toward rigorous results}, J. Phys. A: Math. $\&$ Theor. \textbf{53}(41), (2020).

\bibitem{conBurioni} E. Agliari, et al., {\em Notes on the p-spin glass studied via Hamilton-Jacobi and smooth-cavity techniques}, J. Math. Phys. \textbf{53}.6:063304, (2012).

\bibitem{PRLNN} E. Agliari, et al. {\em Neural networks with a redundant representation: detecting the undetectable.} Phys. Rev. Lett. \textbf{124}.2, 028301, (2020).

\bibitem{AABF-NN2020} E. Agliari, F. Alemanno, A. Barra, A. Fachechi, {\em Generalized Guerra's interpolating techniques for dense associative memories}, Neural Networks \textbf{128},  254-267, (2020).

\bibitem{AgliariDeMarzo} E. Agliari, G. De Marzo, {\em Tolerance versus synaptic noise in dense associative memories}, Europ. Phys. J. Plus \textbf{135}.11, 1-22, (2020).

\bibitem{GuysAlone} L. Albanese, A. Alessandrelli, {\em Rigorous approaches for spin glass and Gaussian spin glass with P-wise interactions}, available at arXiv \textit{http://arxiv.org/abs/2111.12569} (2021).

\bibitem{FrancAlberto} F. Alemanno, M. Centonze, A. Fachechi, {\em Interpolating between Boolean and extremely high noisy patterns through minimal dense associative memories},  J. Phys. A: Math. $\&$ Theor.  \textbf{53}, 7, (2020).


\bibitem{Amit} D.J. Amit, {\em Modeling brain functions}, Cambridge Univ. Press (1989).

\bibitem{AGS} D.J. Amit, H. Gutfreund, H. Sompolinsky. {\em Storing infinite numbers of patterns in a spin-glass model of neural networks}, Phys. Rev. Lett. \textbf{55}(14) (1985).

\bibitem{Antonio} A. Auffinger,  W.K. Chen, {\em The Parisi formula has a unique minimizer}, Comm. Math. Phys. \textbf{335}.3:1429-1444, (2015).

\bibitem{Antonio2} A. Auffinger, W.K. Chen, {\em Free Energy and Complexity of Spherical Bipartite Models}, J. Stat. Phys. \textbf{157}, 1, 40–59, (2014).

\bibitem{Antonio3} A. Auffinger, Q. Zeng, {\em Existence of two-step replica symmetry breaking for the spherical mixed p-spin glass at zero temperature}, Comm. Math. Phys. \textbf{370}.1:377-402, (2019).

\bibitem{Zecchina-New} C. Baldassi, F. Pittorino, R. Zecchina, {\em Shaping the learning landscape in neural networks around wide flat minima}, Proc. Natl. Acad. Sci. \textbf{117}.1:161-170, (2020).

\bibitem{baldi} P. Baldi, S.S. Venkatesh. {\em Number of stable points for spin-glasses and neural networks of higher orders}, Phys. Rev. Lett. \textbf{58}.9, 913, (1987).

\bibitem{Barbier} J. Barbier, N. Macris, {\em The adaptive interpolation method: a simple scheme to prove replica formulas in Bayesian inference}, Prob. Th. Rel. Fi. \textbf{174}, 1133, (2017).

\bibitem{DAM-C}  H. Bao, R. Zhang, Y. Mao,  {\em The Capacity of The Dense Associative Memory Networks}, Neurocomputing  -in press- (2021). 

\bibitem{Driven} A. Barra,  {\em Driven transitions at the onset of ergodicity breaking in complex networks}, Int. J. Mod. Phys. B \textbf{24}, 1-17, (2010).  

\bibitem{barrabecca} A. Barra, M. Beccaria, A. Fachechi, {\em A new mechanical approach to handle generalized Hopfield neural networks}, Neural Networks \textbf{106}, 205-222, (2018).


\bibitem{Ziguli1} A. Barra, P. Contucci. E. Mingione, D. Tantari, {\em Multi-Species mean-field spin-glasses: Rigorous results}, Ann. H. Poincar\'e   \textbf{16}(3), 691, (2015).

\bibitem{bipartito-mio} A. Barra, G. Genovese, F. Guerra, {\em Equilibrium statistical mechanics of bipartite spin systems},  J. Phys. A: Math. $\&$ Theor.  \textbf{44}.24:245002, (2011).

\bibitem{glassy} A. Barra et al. {\em How glassy are neural networks?}, J. Stat. Mech. P07009, (2012).

\bibitem{soffice} A. Barra, et al. {\em About a solvable mean field model of a Gaussian spin glass}  J. Phys. A: Math. $\&$ Theor.  \textbf{47.15}(155002) (2014).

 \bibitem{BGDiBiasio} A. Barra, A. Di Biasio, F. Guerra, {\em Replica symmetry breaking in mean field spin glasses trough Hamilton-Jacobi technique}, J. Stat. Mech. P09006, (2010).

\bibitem{Albert1} A. Barra, M. Beccaria, A. Fachechi, {\em A new mechanical approach to handle generalized Hopfield neural networks}, Neural Networks (2018).

\bibitem{Barrat} A. Barrat, {\em The p-spin spherical spin glass model}, arXiv preprint cond-mat/9701031, (1997).

\bibitem{Bates1} E. Bates, L. Sloman, Y. Sohn,  {\em Replica symmetry breaking in multi-species Sherrington-Kirkpatrick model}, J. Stat. Phys. \textbf{174}(2):333–350, (2019).
\bibitem{Bates2} E. Bates, Y. Sohn, {\em Crisanti-Sommers formula and simultaneous symmetry breaking in multi-species spherical spin glasses},  arXiv:2109.14791, (2021).
\bibitem{Bates3} E. Bates, Y. Sohn, {\em Free energy in multi-species mixed p-spin spherical models}, arXiv:2109.14790, (2021).

\bibitem{Bovier}  A. Bovier, B. Niederhauser, {\em The spin-glass phase transition in the Hopfield model with p-spin interactions}, Adv. Theor. Math. Phys. \textbf{5}:1001-1046, (2001). 

\bibitem{Dembo} G. Ben Arous, A. Dembo, A. Guionnet, {\em Aging of spherical spin glasses}, Probab. Theor. Relat. Fields \textbf{120}, 1, (2001).

\bibitem{Crisanti2} A. Crisanti, D.J. Amit, H. Gutfreund, {\em Saturation level of the Hopfield model for neural network}, EPL  \textbf{2}.4:337, (1986).

\bibitem{Crisanti} A. Crisanti, H.J. Sommers, {\em The spherical p-spin interaction spin glass model: the statics}, Zeitschrift  Phys. B \textbf{87}.3:341-354, (1992). 

\bibitem{Coolen} A.C.C. Coolen, R. Kuhn, P. Sollich, {\em Theory of neural information processing systems}, Oxford Press (2005).

\bibitem{fachechi} A. Fachechi, {\em PDE/Statistical Mechanics Duality: Relation Between Guerra’s Interpolated p-Spin Ferromagnets and the Burgers Hierarchy}, J. Stat. Phys. \textbf{183.1}(1-28) (2021).

\bibitem{Albert2}  A. Fachechi, E. Agliari, A. Barra, {\em Dreaming neural networks: forgetting spurious memories and reinforcing pure ones}, Neural Networks \textbf{112}, 24, (2019). 

\bibitem{gardner}E. Gardner, {\em Multiconnected neural network models},  J. Phys. A: Math. $\&$ Theor.  20(11), 3453 (1987).

\bibitem{Beppe} G. Genovese, {\em Universality in bipartite mean field spin glasses}, J. Math. Phys. \textbf{53}.12:123304, (2012). 

\bibitem{Guerra} F. Guerra, {\em Broken replica symmetry bounds in the mean field spin glass model}, Comm. Math. Phys. \textbf{233}(1), 1, (2003).

\bibitem{GuerraSum} {F. Guerra, {\em Sum rules for the free energy in the mean field spin glass model}, Fiel. Inst. Comm. \textbf{30}, 11, (2001).}

\bibitem{Gavin} G.S. Hartnett, E. Parker, E. Geist, {\em Replica symmetry breaking in bipartite spin glasses and neural networks}, Phys. Rev. E \textbf{98}.2:022116, (2018).

\bibitem{Nishimori} Y. Kabashima, D. Saad, {\em Statistical mechanics of error-correcting codes}, EPL \textbf{45}.1:97, (1999).

\bibitem{Pax} P. Kivimae, {\em The Ground State Energy and Concentration of Complexity in Spherical Bipartite Models},  arXiv:2107.13138v1 (2021)

\bibitem{denseHop1} D. Krotov, J. J. Hopfield, {\em Dense associative memory for pattern recognition}, Adv. Neural Inf. Proc. Sys. \textbf{29}, 1172-1180, (2016).

\bibitem{denseHop2} D. Krotov, J. J. Hopfield, {\em Dense associative memory is robust to adversarial inputs}, Neural Comp. \textbf{30.12}, 3151-3167, (2018).

\bibitem{leonelli} F.E. Leonelli, et al. {\em On the effective initialisation for restricted Boltzmann machines via duality with Hopfield model.}, Neural Networks \textbf{143}, 314, (2021). 

\bibitem{Marullo} C. Marullo, E. Agliari, {\em Boltzmann Machines as Generalized Hopfield Networks: A Review of Recent Results and Outlooks}, Entropy \textbf{23}(1), 34, (2021).

\bibitem{Ksat} M.Mezard, G. Parisi, R. Zecchina, {\em Analytic and algorithmic solution of random satisfiability problems}, Science \textbf{297}, 812-815, (2002).

\bibitem{postino} M. Mezard, G. Parisi, {\em A replica analysis of the travelling salesman problem}, J. de Phys. \textbf{47}(8):1285-1296, (1986).

\bibitem{MezardMontanari} M. Mezard, A. Montanari, {\em Information, physics, and computation}, Oxford Univ. Press (2009).

\bibitem{Steven} C. Moore, S. Mertens, {\em The nature of computation}, Oxford Univ. Press (2010). 

\bibitem{Remi}  R. Monasson, D. O'Kane, {\em Domains of solutions and replica symmetry breaking in multilayer neural networks}, EPL \textbf{27}.2:85, (1994). 

\bibitem{MPV} M. M\'ezard, G. Parisi, M.A. Virasoro {\em Spin Glass Theory and Beyond}, World Scientific, Singapore (1987). 

\bibitem{ZecchinaOld} R. Monasson, R. Zecchina, {\em Weight space structure and internal representations: a direct approach to learning and generalization in multilayer neural networks}, Phys. Rev. Lett. \textbf{75}.12:2432, (1995).

\bibitem{Murrat1} J.C. Mourrat, {\em Parisi's formula is a Hamilton-Jacobi equation in Wasserstein space}, arXiv preprint arXiv:1906.08471, (2019). 

\bibitem{MurratPanchenko} J.C. Mourrat, D. Panchenko, {\em Extending the Parisi formula along a Hamilton-Jacobi equation}, Electron. J. Probab. \textbf{25}(23), 1, (2020). 

\bibitem{Ziguli2} D. Panchenko, {\em The free energy in a multi-species Sherrington–Kirkpatrick model}, Ann. of Prob. \textbf{43}(6), 3494, (2015).

\bibitem{Dmitry} D. Panchenko, {\em The Sherrington-Kirkpatrick model}, Springer Science $\&$ Business Media, (2013).

\bibitem{Kuhn} H. Steffan, R. Kuhn, {\em Replica symmetry breaking in attractor neural network models}, Z. Phys. B \textbf{95}, 249, (1994). 

\bibitem{Subag1} E. Subag, {\em TAP approach for the multi-species spherical spin glasses I: general theory}, arXiv preprint arXiv:2111.07132, (2021).

\bibitem{Subag2} E. Subag, {\em TAP approach for the multi-species spherical spin glasses II: the free energy of the pure models}, arXiv preprint arXiv:2111.07134, (2021). 

\bibitem{Talagrand} M. Talagrand, {\em The Parisi formula}, Ann. Math. \textbf{163}(1), 221–263, (2006).

\bibitem{Talabook} M. Talagrand, {\em Spin glasses: a challenge for mathematicians}, Springer Science $\&$ Business Media, (2003).

\bibitem{zhang} Y. Zhang et al. {\em Residual dense network for image super-resolution}, Proc. IEEE conf. comp. vis. pattern rec. (2018).

 
\end{thebibliography}
\end{document}